\documentclass[journal]{IEEEtran}
\usepackage{cite}
\usepackage{amsmath,amssymb,amsfonts,amsthm}
\usepackage{mathtools}
\usepackage{algorithmic}
\usepackage{graphicx}
\usepackage{textcomp}
\usepackage{bm}
\usepackage{svg}
\usepackage{xcolor}
\usepackage[table]{xcolor}
\usepackage[dvipsnames]{xcolor}
\usepackage{algorithm}
\usepackage{array}
\usepackage[caption=false,font=footnotesize]{subfig}
\usepackage{acronym}
\usepackage{stfloats}
\usepackage{url}
\usepackage{verbatim}
\usepackage{booktabs,ragged2e}
\usepackage[flushleft]{threeparttable}
\usepackage{hyperref}
\usepackage{bookmark}
\usepackage{siunitx}

\newtheorem{theorem}{Theorem}%[section]
\newtheorem{corollary}{Corollary}%[theorem]
\newtheorem{remark}{Remark}
\newtheorem{proposition}{Proposition}
\newtheorem{example}{Example}

\usepackage[capitalise]{cleveref}
\usepackage[inline]{enumitem}
\usepackage{bbm}

\usepackage{multirow}
\usepackage{tikz}
\usetikzlibrary{tikzmark,calc,decorations.pathreplacing,patterns}
\usepackage{tikzsymbols}
\usepackage{pgfplots} 
\pgfplotsset{compat=1.16,every x tick label/.append style={font=\tiny},every y tick label/.append style={font=\tiny}}

\acrodef{ula}[ULA]{uniform linear array}
\acrodef{upa}[UPA]{uniform planar array}
\acrodef{crb}[CRB]{Cramér-Rao bound}
\acrodef{crbm}[CRBM]{\ac{crb} matrix}
\acrodef{fim}[FIM]{Fisher information matrix}
\acrodef{sos}[SoS]{sums of squares}
\acrodef{mle}[MLE]{maximum-likelihood estimator}
\acrodef{mra}[MRA]{minimum-redundancy array}
\acrodef{mimo}[MIMO]{multiple-input multiple-output}
\acrodef{snr}[SNR]{signal-to-noise ratio}
\acrodef{sinr}[SINR]{signal-to-interference-plus-noise ratio}
\acrodef{jsac}[ISAC]{integrated sensing and communication}
\acrodef{wwb}[WWB]{Weiss-Weinstein bound}
\acrodef{zzb}[ZZB]{Ziv–Zakai bound}
\acrodef{bzb}[BZB]{Bobrovski-Zakai bound}
\acrodef{kkt}[KKT]{Karush-Kuhn-Tucker}
\acrodef{ma}[MA]{movable antenna}
\acrodef{dof}[DoFs]{degrees-of-freedom}
\acrodef{doa}[DoA]{direction-of-arrival}
\acrodef{mse}[MSE]{mean squared error}
\acrodef{kpis}[KPIs]{key performance indicators}
\acrodef{tx}[Tx]{transmit}
\acrodef{rx}[Rx]{receive}
\acrodef{rf}[RF]{radio-frequency}
\acrodef{sdp}[SDP]{semidefinite program}

\newcommand{\txt}[1]{\text{\normalfont #1}}

\DeclareMathOperator{\tx}{t}
\DeclareMathOperator{\rx}{r}
\DeclareMathOperator{\T}{\top}
\DeclareMathOperator{\HT}{\mathrm{H}}
\DeclareMathOperator{\F}{\mathrm{F}}

\DeclareMathOperator*{\argmax}{arg\,max}

\DeclareMathOperator{\trace}{\txt{trace}}

\DeclareMathOperator{\Sorth}{\mathbf{S}_{\txt{orth}}}

\newcommand{\icassp}{black}
\newcommand{\edited}{black}

\begin{document}

\title{
%On the Optimal Design of Arrays and Waveforms for Active Sensing
%Rethinking Array-Waveform Design in Active Sensing: Spatial Covariance, Redundancy, and Sensor Allocation in CRB-Optimal Designs
% Rethinking CRB-Based Array Geometry and Waveform Design: Role of Spatial Covariance, Redundancy, and Optimal Sensor Allocation
CRB-Optimal Arrays
%, Sensor Allocations, 
and Waveforms in Active Sensing:
Role of Redundancy and Spatial Covariance of Array Geometry
}
% Alternative titles: 
% - Joint Array–Waveform Design in Active Sensing: \\ Fundamental Limits and Optimality Conditions

\author{Ids van der Werf,~\IEEEmembership{Graduate Student Member,~IEEE,} Robin Rajamäki, ~\IEEEmembership{Member,~IEEE,} and Geert Leus,~\IEEEmembership{Fellow,~IEEE} 
\vspace{-.5cm}
        % <-this % stops a space
%\thanks{Manuscript received 22 October 2025; revised 26 April 2026; accepted 28 May 2026.}
\thanks{\textcopyright\ 2026 IEEE. Personal use of this material is permitted. Permission from IEEE must be obtained for all other uses, in any current or future media, including reprinting/republishing this material for advertising or promotional purposes, creating new collective works, for resale or redistribution to servers or lists, or reuse of any copyrighted component of this work in other works.}
\thanks{This work was partly funded by the Netherlands Organisation for Applied Scientific Research (TNO) and the Netherlands Defence Academy (NLDA), reference no. TNO-10026587 and by Business Finland 6G-ISAC, Research Council of Finland FUN-ISAC 359094, and EU INSTINCT 101139161.}
\thanks{Ids van der Werf and Geert Leus are affiliated with the Microelectronics department, Delft University of Technology, Delft, The Netherlands (email: \{i.vanderwerf, g.j.t.leus\}@tudelft.nl). Robin Rajamäki is with the Signal Processing Research Centre, Tampere University, Finland (email: robin.rajamaki@tuni.fi). He was previously with the Department of Information and Communications Engineering, Aalto University, Finland.}}
% <-this % stops a space

%\markboth{Journal of \LaTeX\ Class Files,~Vol.~X, No.~X, September~2025}%


\maketitle

\begin{abstract}
% The performance of active sensing systems depends critically on the joint design of array geometries and transmit waveforms, yet most prior work has optimized these elements in isolation. This paper establishes fundamental limits and optimality conditions for jointly designed array–waveform pairs in MIMO radar. We show that, with orthogonal waveforms, the single-target Cramér–Rao bound (CRB) depends on the spatial variance of both the transmit and receive array, and equivalently on the multiplicity weighted sum co-array variance, leading to redundant optimal designs. In contrast, for beamforming waveforms, the CRB was shown to depend only on the receive array, allowing nonredundant optima. We derive optimal sensor distributions under both regimes and reveal that symmetric allocations are not always optimal. Extending to planar arrays, we generalize the role of spatial variance to spatial covariance and establish a condition under which beamforming is optimal for a broad class of objective functions. Finally, we present a novel connection between sparse array design and Diophantine equations, showing that different geometries can have identical CRBs, but may exhibit distinct maximum-likelihood performance in the low-SNR regime.
This paper characterizes the performance limits of optimal array designs using orthogonal and coherent waveforms for both linear and planar arrays. 
For orthogonal waveforms, we show that the single-target Cramér–Rao Bound (CRB) depends on the sum of the so-called spatial variances of the transmit (Tx) and receive (Rx) arrays, or equivalently, the spatial variance of the sum co-array weighted by the multiplicities of the virtual sensors. This reveals that CRB-optimal geometries are inherently redundant, highlighting a fundamental trade-off between mean squared error (MSE) and identifiability in parameter estimation. Moreover, we derive optimal Tx-Rx sensor allocations given a total sensor budget and show that unequal allocation (favoring the Rx) is optimal even for nonredundant arrays, 
%challenging 
\textcolor{\edited}{questioning} 
conventional designs. We extend our results to planar arrays, providing a new general condition that the spatial covariances of the Tx and Rx arrays should satisfy for the optimal waveforms to direct %nonnegligible 
power in the target direction. % and thereby %be physically relevant. 
Additionally, we establish a connection between Diophantine equations and array geometries with equal CRB, along with a constructive method for designing such arrays. Our work provides new guidelines for and insights into optimal array and waveform design with relevance in emerging active sensing multiple-input multiple-output systems.
\end{abstract}


% \textcolor{\icassp}{[Ids: Text in this color is text that \emph{literally} comes from the ICASSP paper.]}

\section{Introduction}

\IEEEPARstart{A}{ctive} sensing systems, including automotive radar~\cite{sun2020mimoradar}, unmanned aerial systems~\cite{hugler2018radar}, and sonar systems~\cite{sharaga2015optimal,liu2021highresolution}, 
%rely on the joint design of sensor arrays and transmit waveforms. Many of these applications 
make increasing use of \ac{mimo} 
arrays. 
%radar, 
Unlike phased arrays, these can transmit several independent waveforms 
to improve
%This design directly impacts 
angular resolution, parameter identifiability, and detection performance \cite{li2007mimoradar,li2007onparameter}. %while often being subject to space, hardware, or energy constraints. 
%Given these challenges, a fundamental question arises: 
Given the importance of the \emph{array geometry} and \emph{transmit waveforms} to the performance of \ac{mimo} systems,
the \emph{optimal} design of these key spatio-temporal resources becomes critical.
The \ac{crb} remains perhaps the most widely used tool for multisensor system performance analysis \cite{mirkin1991cramerraobounds,yau1992worst,gershman1997anote,nielsen1994azimuth,baysal2003onthegeometry,oktel2005abayesian,gazzah2006Cramerraobounds,he2010target,tohidi2019sparse,amin2024sparsearrays,werf2024transmit,shahsavari2025cramerrao}, due to being both estimator-independent and analytically tractable. 
Early works on \ac{crb}-optimal array design mainly considered \textit{passive} sensing \cite{mirkin1991cramerraobounds,yau1992worst,gershman1997anote,nielsen1994azimuth,baysal2003onthegeometry,oktel2005abayesian,gazzah2006Cramerraobounds}.
%In \ac{crb}-optimal array design, 
%prior work has mainly considered either (i) optimizing array geometries~\cite{he2010target,tohidi2019sparse,tabrikian2021cognitive,amin2024sparsearrays,werf2024transmit} or (ii) designing transmit waveforms~\cite{li2008range,vanderwerf2023transmit}, often under specific operating conditions such as clutter~\cite{stoica2012optimization} or \ac{jsac} requirements~\cite{bica2019radar,liu2020jointtransmit,liu2022cramerrao,li2023optimal}. 
%We now briefly review these two directions.
% array design passive sensing
%especially %, the impact of array geometry on 
%\textit{passive} sensing has been studied %extensively~\cite{mirkin1991cramerraobounds,gershman1997anote,nielsen1994azimuth,baysal2003onthegeometry}.
%The work~
In particular, \cite{gershman1997anote} demonstrated that the optimal receive array geometry tends to form 
\emph{clusters} of sensors\textcolor{\edited}{,} an insight later corroborated by \cite{baysal2003onthegeometry} which showed the dependence of the single-source \ac{crb} on the so-called \emph{moment of inertia} of the array geometry (sum of squared sensor positions). 
%$n\!+\!1$ clusters, where $n$ denotes the number of sources. This was corroborated mathematically in the single-target case by connecting the \ac{crb} to the so-called \emph{moment of inertia} of the array geometry (sum of squared sensor positions)~\cite{baysal2003onthegeometry}. 
%
More recent approaches relax %cast 
array design 
to 
%as 
a 
%(relaxed) 
convex sensor selection problem, promoting sparsity %while achieving 
with 
a desired
%guaranteeing estimation performance~\cite{chepuri2015sparsitypromoting}, %. Extensions include 
sidelobe suppression~\cite{roy2013sparsity}, %worst-case %two-target design~\cite{kokke2025arraydesign}, and applications to \ac{jsac}~\cite{werf2024receiver}. 
average / worst-case estimation error \cite{chepuri2015sparsitypromoting,kokke2025arraydesign}, or \ac{jsac} 
constraints
%performance 
\cite{werf2024receiver}. 
%Results for two-target scenarios corroborate earlier findings on clustered geometries.
% waveform design active sensing
In \textit{active} sensing, extensive work exists on both 
%not only the array geometry but also the transmit waveforms play a critical role in overall performance. 
array \cite{he2010target,tohidi2019sparse,amin2024sparsearrays,werf2024transmit} and waveform design using \ac{crb} \cite{li2008range,vanderwerf2023transmit}, often under additional constraints imposed by clutter~\cite{stoica2012optimization} or \ac{jsac} requirements~\cite{bica2019radar,liu2020jointtransmit,liu2022cramerrao,li2023optimal}. 
%We now briefly review these two directions.
% array design passive sensing
%Exploiting waveform diversity offers significant advantages~\cite{li2007mimoradar}, such as improved angular resolution, parameter identifiability~\cite{li2007onparameter}, and target detection. 
% While fully orthogonal waveforms are most common, partially coherent waveforms have also been investigated~\cite{stoica2007onprobing,fuhrmann2008transmitbeamforming}, and shown to provide performance gains in various tasks. 
Of particular relevance is \cite{li2008range}, which %investigated 
%characterized 
showed that 
\ac{crb}-optimal waveform matrices %matrix 
%\emph{given} an %linear array geometry. 
%It was shown that, when minimizing the single-target \ac{crb}, the 
%showing 
%them to be 
%that optimal waveforms %is 
are 
confined to a low-dimensional subspace determined by the unknown parameters of interest. %(specifically, the array manifold and its first derivatives). 
In the single-target case, these correspond 
% a 
% convex combination of 
% signal components in the range space of the array manifold and its derivative, 
%\textcolor{\edited}{the coherent component and its orthogonal (nullspace) counterpart (often referred to as the sum and difference beams)}, 
% \textcolor{red}{[undefined, non-standard terminology]} 
%with proportions determined by 
to beamforming in the target direction %in the single-target case 
under certain conditions on the moments of inertia of the \ac{tx} and \ac{rx} arrays \cite{forsythe2005waveform,li2008range,vanderWerf2025Jointly}. %Jointly optimal array–waveform pairs, however, were not explored in~\cite{forsythe2005waveform, li2008range}. 
% array design active sensing
% Additionally, transmit and receive array geometries are inherently coupled, making optimal design for active sensing considerably more intricate than in the passive case. For example, \cite{rajamaki2023importance} showed that two array geometries, despite having the same number of sensors and the same sum co-array, may yield different identifiability, illustrating the tight coupling between transmit and receive geometry. To simplify this coupling, most existing works on array design assume orthogonal waveforms and rely on (sub-optimal) iterative procedures~\cite{roberts2011sparse,werf2024transmit} or submodular optimization~\cite{tohidi2019sparse}, leaving the problem of jointly optimal transmit–receive designs open. 
%Additionally, transmit and receive array geometries are inherently coupled, making optimal design in active sensing considerably more challenging than in the passive case.
%Most existing works simplify this coupling by assuming orthogonal waveforms and employing iterative or submodular optimization methods~\cite{roberts2011sparse,werf2024transmit,tohidi2019sparse}, leaving the problem of jointly optimal transmit–receive array design open. 
%Recently, \ac{crb}-optimal array-waveform pairs were identified for linear arrays~\cite{vanderWerf2025Jointly}, %but 
%\textcolor{red}{\textbf{Some argument about multi-target including reference to \cite{werf2025on}?}}
Recently, \cite{vanderWerf2025Jointly} also %characterized optimal linear array geometries in this case, 
derived optimal linear array geometries that both minimize the single-target \ac{crb} when transmitting optimal (coherent) waveforms and guarantee \textit{identifiability} of the largest possible number of targets when transmitting \emph{orthogonal} waveforms.
% However, 
% open questions regarding 
% %the use of 
% \ac{crb}-optimal array geometries remain in case of \textcolor{\edited}{both coherent and} 
% %how the single-target \ac{crb} depends on the \ac{tx} and \ac{rx} array geometries when transmitting 
% orthogonal waveforms---the latter being crucial for initial target search. %as opposed to coherent (\ac{crb}-optimal) waveforms requiring 
% when prior knowledge of the %(initially unknown) 
% target directions are unavailable. 
Orthogonal waveforms 
are crucial for initial target search, and 
allow fully utilizing the spatial \ac{dof} provided by the so-called \emph{sum co-array} \cite{hoctor1990theunifying}\textcolor{\edited}{,} a virtual array model arising from the sums of \ac{tx}-\ac{rx} sensor position pairs.

While the role of the sum co-array geometry in parameter identifiability is well established \cite{roberts2011sparse,rajamaki2023importance}, its impact on \ac{mse} is less clear. 
This motivates the questions: \emph{How does the %physical and 
sum co-array geometry influence the \ac{crb} when deploying orthogonal waveforms? Specifically, which geometries minimize \ac{crb}}? 
% The de facto standard in active sensing employs orthogonal waveforms with a canonical \ac{mimo} array, where the \ac{tx} and \ac{rx} arrays form a \ac{ula} and dilated \ac{ula}, respectively. Typically, the same number of sensors is allocated to the \ac{tx} and \ac{rx} to maximize the $N_{\tx}N_{\rx}$ spatial \ac{dof}, which benefits target identifiability and spatial resolution~\textcolor{red}{[REF]}. However, the impact of this allocation on the \ac{crb} iremains unclear, leading to the question: \emph{How should a given number of sensors be optimally allocated between the \ac{tx} and \ac{rx} arrays for various array geometries and waveform choices?
%Moreover, most prior studies consider linear arrays. How \ac{crb}-optimal array-waveform pairs extend to \emph{planar} arrays, both for coherent and orthogonal waveforms, remains unknown. This raises the question: \emph{How do optimal waveform / array designs extend to planar arrays, and what new challenges arise?} 
%Answers in case of both linear and planar array geometries are currently unknown.
%Other pertinent open questions regarding both orthogonal and \ac{crb}-optimal (coherent) waveforms also remain. 
Additional open questions regarding 
%the use of 
\ac{crb}-optimal array geometries remain for both coherent and
%how the single-target \ac{crb} depends on the \ac{tx} and \ac{rx} array geometries when transmitting 
orthogonal waveforms. %---the latter being
Firstly, in emerging applications, only a limited number of \ac{rf} chains may be available to divide between the \ac{tx} and \ac{rx} arrays due to their high cost and power consumption. 
The question thus arises: \emph{``How should a given number of \ac{rf} chains or sensors be optimally allocated between the \ac{tx} and \ac{rx} arrays?% for various array geometries and waveform choices?
''} 
%only constrain the \emph{combined} number of \ac{tx} $+$ \ac{rx} sensors.
%In addition to the \ac{tx} and \ac{rx} array geometries, also the allocation of sensors between the two is a key factor impacting performance. This is , 
Recently, \cite{liu2025doaestimation} investigated this question in a monostatic \ac{jsac} setting, aiming to minimize \ac{doa} estimation error under a communications \ac{sinr} constraint. % by optimizing the \ac{tx}/\ac{rx} split given a \emph{total} number of antennas.
%, raising key questions about optimal sensor allocation and its dependence on array geometry. 
Despite extensive numerical results \cite{liu2025doaestimation}, theoretical understanding of optimal allocations for sensing, communications, and \ac{jsac} still remains limited. 
Secondly, in the case of planar arrays\textcolor{\edited}{, which are of great practical relevance due to their ability to discriminate azimuth and elevation,} the notion of \ac{crb}-optimality itself depends on the choice of (scalar) objective function. \textcolor{\edited}{Indeed,} %as 
the Fisher information is matrix-valued even in case of a single target (parametrized by two unknown angles). While \cite{forsythe2005waveform} tackled the case of matrix trace, rigorous optimality conditions for this and other common objectives, such as the log-determinant, are still lacking. 
%Furthermore, in the case of planar arrays---of great practical relevance due to their ability to discriminate both azimuth and elevation angles---how does the choice of scalar objective function 
Finally, \cite{vanderWerf2025Jointly} suggested that \emph{different} array geometries 
%corresponding to sequences with \emph{equal} \ac{sos} 
satisfying certain Diophantine equations 
can achieve the same CRB. %for coherent waveforms. 
However, methods for constructing such ``equi-\ac{crb}" arrays were not given. Developing such methods is of practical importance to array design, since geometries with identical \ac{crb}s may vastly differ in robustness to noise, interference, or mutual coupling.

\textcolor{\edited}{
%Given their increasing relevance in emerging applications where spatio-temporal resource-efficiency is key, 
The open challenges outlined above call for a critical re-examination of the performance limits of \ac{mimo} active sensing systems, including optimal array and waveform design.
% ,
% especially given their increasing relevance in emerging applications where spatio-temporal resource-efficiency is key.
}
%%%%%%%% Added 24.04.2026
\textcolor{\edited}{
This paper focuses on the single-target \ac{crb} to provide sharp answers to these questions in the high \ac{snr} (asymptotic) regime, where the \ac{crb} closely bounds \ac{mse}. We note that alternative bounds---such as the \ac{wwb}~\cite{sun2024optimal}, \ac{zzb}~\cite{khan2013explicit,zhang2023ZivZakaiBound,zhang2024zivzakai2D}, \ac{bzb}~\cite{tabrikian2016cognitive,tabrikian2021cognitive} or Barankin bound \cite{tabrikian2006barankin,wang2023barankin}---more accurately capture low-\ac{snr} (non-asymptotic) performance, but involve intricate expressions that typically hinder their closed-form analysis and require resorting to numerical evaluation instead. Analyzing the \ac{crb} in the multi-target case becomes highly non-trivial for similar reasons \cite{liu2024crb,werf2025on,liu2026joint}. 
Previous studies have nevertheless numerically investigated these bounds in various settings, such as performance analysis and array design in passive sensing via \ac{zzb} \cite{zhang2023ZivZakaiBound,zhang2024zivzakai2D} and Barankin bound \cite{wang2023barankin}, structured waveform optimization via \ac{wwb} \cite{sun2024optimal}, and adaptive receive-only array design via \ac{bzb} \cite{tabrikian2021cognitive}. 
In contrast, the single-target \ac{crb} considered herein provides a simple and tractable baseline guiding optimal array and waveform design, as well as a preliminary step toward a full analytical understanding of the multi-target and low-\ac{snr} regimes.
%In contrast, the single-target \ac{crb} considered herein provides a tractable baseline to guide optimal array and waveform design on the road toward a complete analytical understanding of the multi-target and low-\ac{snr} regimes.
%In contrast, the single-target \ac{crb} examined in this paper serves as a tractable baseline, guiding optimal array and waveform design, on the road towards an analytical understanding multi-target and low-\ac{snr} regimes.
%In summary, while CRB-based design does not fully capture finite-SNR behavior, it remains a tractable and insightful criterion, and our numerical results explicitly illustrate its limitations in the non-asymptotic regime. Gaining analytical insight into alternative (non-asymptotic) array design criteria, such as the ZZB or MLE, is an interesting (and challenging) direction for future work.
}

\subsection{Contributions and outline}
% jointly optimal array-waveform pairs
% Overall, while array geometry and waveform design have been studied extensively in isolation, their joint design in active sensing remains largely unexplored. This gap is due to the intricate non-convex coupling between the spatial and waveform domains, and motivates the present work. 
% Building on the foundation in \cite{forsythe2005waveform,li2008range}, our earlier work~\cite{vanderWerf2025Jointly} identified optimal array–waveform pairs for linear arrays in the single-target scenario.
% In the present paper, we significantly extend these results by providing a more comprehensive analysis that, amongst others, establishes performance limits for \textit{orthogonal} waveforms, a generalization to planar arrays and a generalized condition on beamforming optimality. 
% The main contributions of this work are summarized as follows:

This paper 
% characterizes performance limits of \ac{mimo} active sensing systems %optimal array designs 
% employing \textit{orthogonal} and coherent (\ac{crb}-optimal) waveforms for both linear and planar arrays. 
% Our results 
provides new insights into, and guidelines for, the design of optimal array geometries and waveforms. 
We focus on the single-target case due to its continued relevance in radar and \ac{jsac} \cite{chiriyath2016inner,liu2022integrated}, 
% \textcolor{red}{Ids: possible reference:
% \begin{itemize}
%     \item F. Liu et al. (2022) ``ISAC: toward ..."
%     \item A. R. Chiriyath et al. (2016) ``Inner bounds on ..." (actually considers only one target)
% \end{itemize} ]}, 
as well as the %. The %analytical 
fundamental theoretical
%tractability of the \ac{crb} in this case is key for gaining valuable 
insight 
it is able to yield 
into %fundamental performance gains achievable by 
array and waveform design. 
% \textcolor{red}{\textbf{Or a more elaborate motivation for the single target case here?}}
%, as we will see.
Our main contributions %of our work 
are summarized as follows:
\begin{itemize}
    \item 
    %We establish that the single-target \ac{crb}, in case of orthogonal waveforms, is a function of the spatial variance of both the transmit and receive array. It can also be expressed as the spatial variance of the sum co-array elements weighted by their multiplicities. We show that the corresponding optimal array for orthogonal waveforms is redundant, in contrast to the beamforming case, where the \ac{crb} only depends on the receive array and nonredundant arrays can be optimal~\cite{vanderWerf2025Jointly}.
    \emph{Optimal arrays for orthogonal waveforms:} 
   % We show that the single-target \ac{crb} in case of orthogonal waveforms depends on the sum of the \emph{spatial variances} (moments of inertia) of the Tx and Rx arrays, or alternatively, 
    We show that the single-target \ac{crb} in case of orthogonal waveforms can be expressed as 
    the \emph{weighted spatial variance} (moment of intertia)
    of the sum co-array, where the weights correspond to the multiplicities of the sum co-array elements. This reveals that the \ac{crb}-optimal array geometry has a redundant sum co-array, unlike nonredundant array geometries that are typically employed for maximizing the number of identifiable targets. Our result thus illustrates a fundamental trade-off between \ac{mse} and identifiability in array design.
    \item 
    %    We derive the optimal sensor distributions for both orthogonal and optimal (beamforming) waveforms, and for the canonical \ac{mimo} and optimal array. Surprisingly, and contrary to intuition, the optimal sensor distribution is not always the symmetric $N/2$ allocation.
    \emph{Allocation of sensors between Tx and Rx:} %Optimal Tx-Rx sensor allocations:
    We derive the optimal allocation of a total sensor budget between the Tx and Rx arrays for both the \ac{crb}-optimal array geometry and a canonical (nested nonredundant) \ac{mimo} configuration. %considered in \ac{mimo} radar. 
    We show that distributing the sensors \emph{unequally} between Tx and Rx when employing orthogonal waveforms is optimal, with a preference given towards more Rx than Tx sensors. %---even in case of nonredundant arrays. 
    Surprisingly, %This is surprising, as 
    even for nonredundant arrays, the \ac{crb}-maximizing sensor allocation does not yield the array geometry with the largest sum co-array\textcolor{\edited}{, unlike} conventional designs optimized for identifiability. 
    \item 
    % We extend the results of this work and~\cite{vanderWerf2025Jointly} to planar arrays, where the role of the spatial variance is replaced by the spatial covariance (or 2D moment of inertia). We also establish a generalized condition under which beamforming is optimal for arbitrary non-decreasing objective functions, such as the trace, log-determinant or maximum eigenvalue of the \ac{crbm}.
    \emph{Generalizations to planar arrays:} %Optimal waveforms for
    We extend our results to two dimensional arrays, where the role of the spatial variance is replaced by the spatial \emph{covariance}. We derive a novel \textcolor{\edited}{necessary and} sufficient condition that the Tx and Rx spatial covariances should satisfy for the optimal waveforms to remain physically relevant by steering power in the target direction. Our result provides a rigorous generalization of previous works \cite{forsythe2005waveform}, holding for any convex objective function, such as the trace, log-determinant or maximum eigenvalue of the \acl{crbm}.
    \item 
%    We show that different array geometries with identical \ac{crb} can be generated via sequences with equal sums of squares. This establishes a novel connection between sparse array design and Diophantine equations. Despite identical \ac{crb}, we demonstrate numerically that such arrays may yield different \ac{mle} performance at low \ac{snr} or under mutual coupling.
    \emph{Array geometries with equal \ac{crb}:} 
    We show that \emph{different} array geometries with \emph{identical} \ac{crb} can be generated via sequences with equal \emph{\ac{sos}}. This establishes an intriguing connection between sparse array design and the classical topic of Diophantine equations. We also provide a simple method for generating such equi-\ac{crb} (linear or planar) arrays for an arbitrary number of sensors. 
    \textcolor{\edited}{This method can be used to find alternative array geometries that maintain a desired \ac{crb} while offering, for example, improved estimation performance at low \ac{snr} or reduced susceptibility to mutual coupling.}
    % This can be used to find alternative array geometries---to a given array with a desired \ac{crb}---with, e.g., improved estimation performance at low \ac{snr} or reduced susceptibility to mutual coupling.
   % , for example, improve upon a given array geometry with a desired \emph{crb} in terms of
%    This can be useful for constructing array geometries with equal \ac{crb} but different \ac{mle} performance at low \ac{snr}, or susceptibility to mutual coupling.
    %and demonstrate numerically that they may nevertheless yield different \ac{mle} performance (e.g., at low \ac{snr}), which gives rise to new questions in \ac{mse}-optimal array design.
\end{itemize}
%\textcolor{red}{[Ids: Maybe the following paragraph can be omitted, since (most of) this is already covered in the introduction?]}
%Our work builds upon \cite{forsythe2005waveform,li2008range}, and our earlier work \cite{vanderWerf2025Jointly} that focused on optimal array–waveform pairs for linear arrays, but left open questions regarding orthogonal waveforms that are highly important for initial target search---in contrast to \ac{crb}-minimizing beamforming waveforms which require prior knowledge of the unknown direction of interest. 
% We focus on the single-target case due to its continued relevance in radar \textcolor{red}{[REF]} and ISAC \textcolor{red}{[REF]}. The analytical tractability of the single-target \ac{crb} is key for gaining valuable insight into fundamental performance gains achievable by array and waveform design, as we will see.

% Importantly, understanding the single-target optimal designs provides a valuable reference point for the multi-target scenario. Practical array design methods for multiple targets often rely on heuristics due to the increased complexity of the problem, and thus may yield array configurations that are not globally optimal. The single-target analysis clarifies the fundamental trade-offs, aids in interpreting the outcome of (suboptimal) multi-target designs, and establishes performance limits that such methods can realistically achieve.

%\subsection{Outline}\label{subsec:outline}
The remainder of this paper is organized as follows. \cref{sec:background} briefly reviews %introduces 
the %necessary %mathematical 
background and %summarizes 
prior results on \ac{crb}-optimal array and waveform designs. %pairs for linear arrays. 
\cref{sec:orthogonalWaveforms,sec:optimalAllocation} present new results %addresses the case of 
regarding orthogonal waveforms and optimal sensor allocation, respectively, while 
\cref{sec:planarArrays} considers extensions to planar arrays. 
%We extend these results to planar arrays in \cref{sec:planarArrays}. %In 
\cref{sec:arraysWithEqualCRB} examines distinct arrays with identical \ac{crb}s. Finally, \cref{sec:simulations} validates our findings numerically and \cref{sec:conclusions} concludes the paper.
%presents numerical examples illustrating the estimation performance.

% \subsection{not useful anymore?}
% on submodularity for sensor selection:
% \begin{itemize}
%     \item It is worth mentioning that \cite{tohidi2019sparse} demonstrated the supermodularity of D-optimality, although for a somewhat distinct situation.
%     \item \cite{shamaiah2010greedy} also proved submodularity of a $\log\det$ function for sensor selection.
%     \item \cite{ranieri2014nearoptimal} linear inverse problems, does not consider determinant of fisher?
%     \item \cite{hashemi2021randomized} 
%     \item \cite{liu2023greedy}
%     \item \cite{tzoumas2016sensor} submodularity of logdet of MSE of linear estimator; applied to control/Kalman filtering
% \end{itemize}

% Our initial research regarding jointly optimal transmit waveforms and array geometries was presented at the 2025 International Conference on Acoustics, Speech and Signal Processing (ICASSP)~\cite{vanderWerf2025Jointly}. That study concentrated exclusively on optimal waveforms (the sum-beam) for linear arrays. In this paper, we extend these results and offer a more comprehensive and detailed analysis, incorporating \textit{orthogonal} waveforms and planar arrays.

\section{Background}\label{sec:background}
This section introduces the signal model %, \ac{crb}, 
and briefly reviews prior %some relevant existing 
results on \ac{crb}-optimal waveforms and arrays, %which
%Detailed derivations 
% Details 
% can be found in 
% \cite{stoica1989music,forsythe2005waveform,li2008range,vanderWerf2025Jointly}. 
forming the basis for our new contributions starting in \cref{sec:orthogonalWaveforms}.

\subsection{Signal model and Cramér-Rao bound}
We consider a narrowband monostatic MIMO system with $N_{\tx}$ \ac{tx} and $N_{\rx}$ \ac{rx} sensors in a linear array configuration illuminating a far-field target in 2D space at the angle $\omega =  %\frac{2\pi}{\lambda}
\pi\sin(\theta)$ with reflection coefficient $\gamma\in\mathbb{C}$.
The spatio-temporal received (backscatter) signal vector can be written as \cite{bekkerman2006target,li2007mimoradar}
\begin{equation}
\begin{split}
    \mathbf{y} 
    & =(\mathbf{S}\otimes \mathbf{I}_{N_{\rx}})(\mathbf{a}_{\tx}(\omega)\otimes \mathbf{a}_{\rx}(\omega) )\gamma + \mathbf{n}, %\quad \in\mathbb{C}^{N_{\rx}},
    \label{eq:signal_model}
\end{split}
\end{equation}
\textcolor{\icassp}{where $\otimes$ denotes the Kronecker product, and $\mathbf{S}\in\mathbb{C}^{T\times N_{\tx}}$ is a spatio-temporal \ac{tx} waveform matrix} of (sample) length $T\geq 1$ and unit power, i.e., $\|\mathbf{S}\|_{\F}^2=\trace(\mathbf{S}^{\HT}\mathbf{S})=1$. 

The entries of the \ac{tx} and \ac{rx} array response vectors are given by
$\left[\mathbf{a}_{\tx}(\omega)\right]_n = e^{j d_{\tx}[n] \omega}$ and $\left[\mathbf{a}_{\rx}(\omega)\right]_m = e^{j d_{\rx}[m] \omega}$\textcolor{\edited}{, where $d_{\tx}[n]$ and $d_{\rx}[m]$ denote the individual sensor positions in half-wavelength units. The complete sets containing these respective positions are}
\begin{align}
    \mathcal{D}_{\tx} 
    & = \{d_{\tx}[n]\}_{n=1}^{N_{\tx}}, \quad
    %\\
    \mathcal{D}_{\rx} 
    %&
    = \{d_{\rx}[m]\}_{m=1}^{N_{\rx}}.\label{eq:sensorpositions}
\end{align}
Additionally, $\mathbf{n}\!\in\!\mathbb{C}^{N_{\rx}T}$ denotes zero-mean circularly symmetric complex Gaussian noise %vector 
with covariance $\mathbb{E}(\mathbf{n}\mathbf{n}^{\HT})\!=\!\sigma^2 \mathbf{I}$.

% The goal is to identify the array geometry ($\mathcal{D}_{\tx}$ and $\mathcal{D}_{\rx}$) and the transmit waveform ($\mathbf{S}$) that jointly optimize the estimation of the unknown target angle $\omega$.

Assuming $\omega$, $\gamma$, and $\sigma$ are deterministic unknowns ($\gamma$ and $\sigma$ are nuisance parameters), the \ac{crb} of $\omega$, given waveform $\mathbf{S}$, can be written, adapting \cite[Theorem~4.1]{stoica1989music}, as\textcolor{\edited}{\footnote{\textcolor{\edited}{By the parameter transformation property \cite[Sec. 3.6]{kay1993fundamentals}, $\txt{CRB}_{\theta}(\mathbf{S}) = \frac{1}{(\pi \cos(\theta))^2} \txt{CRB}_{\omega}(\mathbf{S})$. Because this scalar transformation is independent of the design variables, any geometry or waveform minimizing $\txt{CRB}_{\omega}$ is also optimal for $\txt{CRB}_{\theta}$.}}}
\begin{align}
    \txt{CRB}_{\omega}(\mathbf{S})
    % & =
    % \frac{\sigma^2}{2|\gamma|^2}
    % \left[
    % \txt{Re}
    % \left\{ 
    % \mathbf{\dot{a}}^{\HT}_\textrm{tr}(\mathbf{S}^{\HT}\otimes \mathbf{I})
    % \mathbf{P}_{(\mathbf{S}\otimes \mathbf{I})\mathbf{a}_\textrm{tr}}^\perp(\mathbf{S}\otimes \mathbf{I})\mathbf{\dot{a}}_\textrm{tr}
    % \right\}
    % \right]^{-1} \notag\\
    & = 
    \frac{\sigma^2}{2|\gamma|^2}
    \|
    \mathbf{P}_{(\mathbf{S}\otimes \mathbf{I})\mathbf{a}_\textrm{tr}}^\perp(\mathbf{S}\otimes \mathbf{I})\mathbf{\dot{a}}_\textrm{tr}
    \|_2^{-2}.
    \label{eq:singleTargetCRB}
\end{align}
Here, the effective \ac{tx}-\ac{rx} steering vector is denoted by 
%\begin{align}
    $
    \mathbf{a}_\textrm{tr}=\mathbf{a}_\textrm{tr}(\omega) \triangleq \mathbf{a}_\textrm{t}(\omega) \otimes \mathbf{a}_\textrm{r}(\omega)
    $ 
    %\label{eq:atr}
%\end{align}
and its derivative with respect to $\omega$ by $\mathbf{\dot{a}}_\textrm{tr}=\mathbf{\dot{a}}_\textrm{tr}(\omega) \triangleq \frac{\partial}{\partial\omega}\mathbf{a}_\textrm{tr}(\omega)$. Moreover, $\mathbf{P}_{\mathbf{X}}^\perp$ is the projection matrix onto the orthogonal complement of the range %space 
of $\mathbf{X}$.

\subsection{Prior results on CRB-optimal waveforms and arrays}
\label{subsec:existingCRBoptWavArrayDesigns}
\subsubsection{Optimal waveform designs}
%We consider a waveform optimal if it is a solution to,
CRB-minimizing unit-power waveforms $\mathbf{S}^\star$ are solutions to the optimization problem
\begin{align}
   \min_{\mathbf{S}} 
 %   \mathbf{S}^\star=\argmin_{\mathbf{S}\in\mathbb{C}^{T\times N_{\tx}}}
    \ \txt{CRB}_{\omega}(\mathbf{S}) \ \txt{s.t.} \ \trace(\mathbf{S}^{\HT}\mathbf{S}) \leq 1.
    \label{eq:minimizeCRBwrtWaveform}
\end{align}
\textcolor{\edited}{While $\txt{CRB}_{\omega}(\mathbf{S})$ is nonconvex in $\mathbf{S}$, it is a convex function of the waveform covariance matrix $\mathbf{R}_{\rm s} = \mathbf{S}^{\HT}\mathbf{S}$.}
Expressing the \ac{crb} \eqref{eq:singleTargetCRB} in terms of the waveform covariance matrix $\mathbf{R}_{\rm s}$, \textcolor{\edited}{the problem \eqref{eq:minimizeCRBwrtWaveform} becomes a convex \ac{sdp}, and} the optimal $\mathbf{R}_{\rm s}$ in \eqref{eq:minimizeCRBwrtWaveform} 
has the following form (for some $\lambda_1,\lambda_2 \geq 0$) \cite{forsythe2005waveform,li2008range}:
%is a linear combination of the so-called ``sum beam'' ($\mathbf{a}_{\tx}$) and ``difference beam'' ($\mathbf{P}^{\perp}_{\mathbf{a}_{\tx}}\mathbf{\dot{a}}_{\tx}$), that is,
%
\begin{align}
    \mathbf{R}_{\rm s} = 
    \lambda_1  
    \frac{
        \mathbf{{a}}_{\tx}\mathbf{{a}}^{\HT}_{\tx}
    }{
        \|\mathbf{{a}}_{\tx}\|_2^2
    }
    + \lambda_2 
    \frac{
        \mathbf{P}^{\perp}_{\mathbf{a}_{\tx}}\mathbf{\dot{a}}_{\tx}\mathbf{\dot{a}}^{\HT}_{\tx}\mathbf{P}^{\perp}_{\mathbf{a}_{\tx}}
    }{
        \|\mathbf{P}^{\perp}_{\mathbf{a}_{\tx}}\mathbf{\dot{a}}_{\tx}\|_2^2
    }. \label{eq:TransmitWaveformCov}
\end{align} 
Intuitively, allocating \ac{tx} power beyond the two-dimensional subspace in \eqref{eq:TransmitWaveformCov} is 
wasteful (necessarily increases the CRB). 
Using \labelcref{eq:TransmitWaveformCov}, \eqref{eq:minimizeCRBwrtWaveform} can be shown to reduce to \cite{forsythe2005waveform,li2008range,vanderWerf2025Jointly}
\begin{align}
\begin{aligned}
    \min_{\lambda_1,\lambda_2} & \ \frac{\sigma^2}{2|\gamma|^2}\frac{1}{N_{\tx}N_{\rx}} \left[\lambda_1 \chi(\mathcal{D}_{\rx})+ \lambda_2 \chi(\mathcal{D}_{\tx})\right]^{-1} \label{eq:CRBequivObjective}\\
    \txt{s.t.} & \ \lambda_1 + \lambda_2 = 1,  \ \lambda_1 > 0, \textcolor{\edited}{\ \lambda_2 \geq 0,} %\\
\end{aligned}
\end{align}
\textcolor{\edited}{where the strict inequality $\lambda_1 > 0$ ensures that the target is illuminated (providing a finite CRB), whereas 
%the difference beam weight 
$\lambda_2$ is allowed to be zero. Moreover,} $\chi(\cdot)$ denotes the so-called ``spatial variance'', which for a linear array $\mathcal{D}$, with ``spatial mean'' $\mu(\mathcal{D})$, is defined as
%For a linear array described by the set $\mathcal{D}$, we define the spatial mean and spatial variance as
\begin{align}
    \chi(\mathcal{D}) \triangleq\frac{1}{|\mathcal{D}|}
            \sum_{d\in\mathcal{D}}(d-\mu(\mathcal{D}))^2,\quad
             \mu(\mathcal{D})
     \triangleq \frac{1}{|\mathcal{D}|}\sum_{d\in\mathcal{D}}d.
    %\label{eq:SpatialMean}\\
            \label{eq:SpatialVariance}
\end{align}
%, weighting the sum and difference beams, $\mathbf{a}_{\tx}$ and $\mathbf{P}^{\perp}_{\mathbf{a}_{\tx}}\dot{\mathbf{a}}_{\tx}$ respectively, in~\eqref{eq:TransmitWaveformCov}. 
Hence, optimal transmit waveforms are functions of the \ac{tx} and \ac{rx} spatial variances $\chi(\mathcal{D}_{\tx}), \chi(\mathcal{D}_{\rx})$ 
per \eqref{eq:CRBequivObjective}. 
\textcolor{\edited}{In physical terms, $\chi(\mathcal{D})$ represents the \textit{moment of inertia} of the array geometry \cite{baysal2003onthegeometry}. A higher spatial variance indicates that the array elements are distributed further from the geometric centroid (e.g., concentrated toward the aperture edges), which is distinct from the concept of array sparsity.}
If 
\begin{align}
    \chi(\mathcal{D}_{\tx}) \!<\! \chi(\mathcal{D}_{\rx}),
    \label{eq:linearArrayDesignRule}
\end{align}
then optimal waveforms correspond to \emph{coherent beamforming} in the target direction $\omega$, i.e., the \emph{``sum beam''} \cite{forsythe2005waveform,li2008range,vanderWerf2025Jointly}:
\begin{align}
    \mathbf{S}^\star = \frac{\mathbf{u}\mathbf{a}_{\tx}^{\HT}(\omega)}{\sqrt{N_{\tx}}},\quad \|\mathbf{u}\|_2=1.
    \label{eq:optTxWaveform_sum}
\end{align}
%Cases with $\chi(\mathcal{D}_{\tx}) \ge \chi(\mathcal{D}_{\rx})$ are less practical and thus omitted here for brevity.
% If $\chi_{\rx}=\chi_{\tx}$, then optimal waveforms beyond \eqref{eq:optTxWaveform_sum} also exist \cite{li2008range}. If $\chi(\mathcal{D}_{\tx}) \ge \chi(\mathcal{D}_{\rx})$, 
% \textcolor{\icassp}{
% then optimal waveforms correspond to transmitting infinitesimal energy in the target direction---see \cite{forsythe2005waveform,li2008range} for details. As this solution has limited practical relevance, we will henceforth focus on \eqref{eq:optTxWaveform_sum} which is optimal given \eqref{eq:sb_cond}.
% }
Optimal waveforms beyond \eqref{eq:optTxWaveform_sum} exist when $\chi(\mathcal{D}_{\tx})\!=\!\chi(\mathcal{D}_{\rx})$ \cite{li2008range}. When $\chi(\mathcal{D}_{\tx})\!>\!\chi(\mathcal{D}_{\rx})$, these correspond to transmitting infinitesimal energy in the target direction \cite{forsythe2005waveform,li2008range}, which has limited practical relevance. Focusing on \eqref{eq:optTxWaveform_sum}, for brevity, we note that $\mathbf{S}^\star$ depends on the \ac{tx} array geometry via $\mathbf{a}_{\tx}(\omega)$. Moreover, 
 %. This reveals the explicit dependence of 
%links the array geometry to the estimation performance, as given \eqref{eq:optTxWaveform_sum}, 
the single-target \ac{crb} is a function of %in 
the 
\ac{rx} 
array spatial variance $\chi(\mathcal{D}_{\rx})$, as \eqref{eq:singleTargetCRB} given \eqref{eq:optTxWaveform_sum} reduces to \cite{vanderWerf2025Jointly}
\begin{align}
    \text{CRB}_{\omega}(\mathbf{S}^\star) =
    \frac{\sigma^2}{2|\gamma|^2}\frac{1}{N_{\tx}N_{\rx}}\chi^{-1}(\mathcal{D}_{\rx}).
    \label{eq:CRB_sb}
\end{align}
Results \eqref{eq:optTxWaveform_sum} and \eqref{eq:CRB_sb} enable characterizing the set of waveforms and array geometries that jointly minimize the CRB.

\subsubsection{Jointly optimal array-waveform designs}
The \ac{rx} array geometry with maximal spatial variance $\chi(\mathcal{D}_{\rx})$ minimizes the CRB in \eqref{eq:CRB_sb}. 
\textcolor{\edited}{We constrain the aperture to the set of non-negative integers, i.e., $L\in\mathbb{N}_+$, 
%to represent discrete multiples of a fundamental inter-element spacing (e.g., $\lambda/2$); 
% ----> since we assume lambda/2 earlier
representing integer multiples of half wavelength spacings;
this ensures a minimum physical separation between elements and prevents sensor overlap \cite{werf2024transmit, amin2024sparsearrays, tohidi2019sparse}.} 
% Given an array aperture $L\in\mathbb{N}_+$, t
The spatial variance is maximized by the so-called ``clustered array" placing sensors at the ends of the aperture \cite[Lemma 1]{vanderWerf2025Jointly}:
\begin{align}
    \mathcal{K}_N^L
    \!\triangleq\!
    \begin{cases}
    \mathcal{U}_{N/2}\cup(L-\mathcal{U}_{N/2})
    , & \mathrm{if}\ N\ \mathrm{even},\\
    \mathcal{U}_{(N-1)/2}\cup(L-\mathcal{U}_{(N+1)/2}),&\mathrm{if}\ N\ \mathrm{odd},
    \end{cases}
    \label{eq:clustered}
\end{align}
where $\mathcal{U}_N\!\triangleq\!\{0,1,\ldots,N\!-\!1\}$ is the $N$-sensor \ac{ula}.
\cref{fig:clusteredRxContiguousSumCoArray} illustrates a clustered \ac{rx} array $\mathcal{D}_{\rx}\!=\!\mathcal{K}_6^{20}$. 
%For illustration, \cref{fig:ClusteredArrayEvenOdd} shows the clustered array geometry for an even and odd number of sensors. Note that in the latter case, there exist two optimal arrays. 
Hence, given $N_{\tx}$, $N_{\rx}$, and $L$, the jointly optimal array-waveform tuple is given by $\mathbf{S}^\star$ in \eqref{eq:optTxWaveform_sum}, $\mathcal{D}_{\rx}^\star=\mathcal{K}_{N_{\rx}}^L$ in \eqref{eq:clustered}, and \emph{any} $\mathcal{D}_{\tx}^\star$ satisfying $\chi(\mathcal{D}_{\rx}^\star)>\chi(\mathcal{D}_{\tx}^\star)$~\cite[Theorem 1]{vanderWerf2025Jointly}, assuming this is feasible given the parameters $N_{\tx}$, $N_{\rx}$, and $L$. 
% \textcolor{red}{[Robin: unclear what ``assuming this is possible'' refers to.]}
The existence of multiple optimal \ac{tx} arrays 
%$\mathbf{S}^\star$ follows \eqref{eq:optTxWaveform_sum} and therefore depends on $\mathcal{D}_{\tx}^\star$, which can be chosen freely provided that $\chi(\mathcal{D}_{\rx}^\star)>\chi(\mathcal{D}_{\tx}^\star)$.
% \begin{figure}
%     \centering
%     \newcommand{\dataLoc}{Data/ClusteredArraysEvenOdd}
%     \newcommand{\xmax}{14}
%     \newcommand{\xmin}{0}
%     \newcommand{\fwidth}{6.5}
%     \newcommand{\msize}{1.5}
% 	\pgfplotsset{every x tick label/.append style={font=\tiny},every y tick label/.append style={font=\scriptsize}}
% 	\centering
% 		\begin{tikzpicture}
% 			\begin{axis}[width= \fwidth cm ,height=2.5 cm,xmin=\xmin-0.2,xmax=\xmax+0.2,xtick={\xmin,\xmin+2,...,\xmax},ytick={-1,0,1},yticklabels={$\mathcal{K}_7^{14}$,$\mathcal{K}_7^{14}$,$\mathcal{K}_6^{14}$},ytick style={draw=none},ymin=-1.5,ymax=1.5,xticklabel shift = 0 pt,xtick pos=bottom,axis line style={draw=none}]
% 				\addplot[blue,only marks,mark=square*,mark size=\msize,y filter/.code={\pgfmathparse{\pgfmathresult*0+1}}] table[x=D,y=D]{\dataLoc/DrEven.dat};
%                 \addplot[blue,only marks,mark=square*,mark size=\msize,y filter/.code={\pgfmathparse{\pgfmathresult*0}}] table[x=D,y=D]{\dataLoc/DrOdd1.dat};
%                 \addplot[blue,only marks,mark=square*,mark size=\msize,y filter/.code={\pgfmathparse{\pgfmathresult*0-1}}] table[x=D,y=D]{\dataLoc/DrOdd2.dat};
% 			\end{axis}
% 		\end{tikzpicture}
%     \caption{Clustered arrays, $\mathcal{K}_N^L$, for $N = 6$ (top) or $N = 7$ (middle and bottom) elements and physical aperture $L = 14$.}
%     \label{fig:ClusteredArrayEvenOdd}
% \end{figure}
can be leveraged to imbue the joint Tx-\ac{rx} array geometry with desirable properties beyond a low CRB. For example, \cref{fig:clusteredRxContiguousSumCoArray} illustrates an array configuration that both minimizes the single-target CRB when transmitting optimal (coherent) waveforms and guarantees the (maximal) identifiability of $N_{\tx}N_{\rx}$ targets when transmitting orthogonal waveforms \cite[Corollary 1]{vanderWerf2025Jointly}. 
{\color{\edited}
Prior works have also observed similar array structures, where Rx elements are placed at the edges and Tx elements in the center, to be desirable in, e.g., ISAC \cite{liu2024crb, liu2025doaestimation}. \ac{crb}-optimal arrays and waveforms (for sensing) provide a theoretical explanation for this, showing that for coherent transmission, the CRB depends only on Rx positions, thereby offering insight into various array geometries obtained via numerical optimization.}
%by ensuring a contiguous, nonredundant sum co-array~ 
%\textcolor{red}{[Premature, co-array not defined.]}
%An example of such a\cite[Corollary 1]{vanderWerf2025Jointly}.
%
\begin{figure}
% \vspace{-.2cm}
    \centering
    \newcommand{\dataLoc}{Data/ClusteredSpecialCase_v2}
    \newcommand{\xmax}{36}
    \newcommand{\xmin}{0}
    \newcommand{\fwidth}{9}
    \newcommand{\msize}{1.5}
	\pgfplotsset{every x tick label/.append style={font=\tiny},every y tick label/.append style={font=\scriptsize}}
	\centering
		\begin{tikzpicture}
			\begin{axis}[width= \fwidth cm ,height=2.5 cm,xmin=\xmin-0.2,xmax=\xmax+0.2,xtick={\xmin,\xmin+2,...,\xmax},ytick={-1,0,1},yticklabels={$\mathcal{D}_\Sigma$,$\mathcal{D}_{\rx}$,$\mathcal{D}_{\tx}$},ytick style={draw=none},ymin=-1.5,ymax=1.5,xticklabel shift = 0 pt,xtick pos=bottom,axis line style={draw=none}]
				\addplot[orange,only marks,line width = 0.7,mark=star,mark size=2.2,y filter/.code={\pgfmathparse{\pgfmathresult*0+1}}] table[x=D,y=D]{\dataLoc/Dt.dat};
				\addplot[blue,only marks,mark=square*,mark size=\msize,y filter/.code={\pgfmathparse{\pgfmathresult*0}}] table[x=D,y=D]{\dataLoc/Dr.dat};
				\addplot[only marks,mark=*,mark size=\msize,y filter/.code={\pgfmathparse{\pgfmathresult*0-1}},x filter/.code={\pgfmathparse{\pgfmathresult}}] table[x=D,y=D]{\dataLoc/D_Sigma.dat};
			\end{axis}
		\end{tikzpicture}%\label{fig:clustTxRx}
    \vspace{-0.3cm}
    \caption{Array geometry with $N_{\tx}\!=\!6$ \ac{tx} and $N_{\rx}\!=\!6$ \ac{rx} sensors. The clustered \ac{rx} array $\mathcal{D}_{\rx} = \mathcal{K}^L_{N_{\rx}}$ maximizes the spatial variance $\chi(\mathcal{D}_{\rx})$ given physical aperture constraint $L\!=\!20$. The sum co-array is contiguous and nonredundant.
    %, the multiplicity of the sum co-array elements is given by $\mathbf{v}=\mathbf{1}^{\T}_{N_{\tx}N_{\rx}}$.
    }
    \vspace{-.3cm}
    \label{fig:clusteredRxContiguousSumCoArray}
\end{figure}
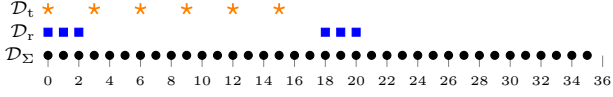

While prior works \cite{forsythe2005waveform,li2008range,vanderWerf2025Jointly} %CRB-minimizing array designs in active sensing 
focused on coherent beamforming\textcolor{\edited}{,} mainly suited for tracking\textcolor{\edited}{,} \ac{mimo} systems often utilize orthogonal waveforms to illuminate the entire field of view without prior knowledge of $\omega$. 
This raises open questions about 
%whether alternative array geometries %or sensor allocation strategies can further reduce the CRB 
%\ac{crb}-optimal array geometries 
%how 
% the single-target \ac{crb} differs from \eqref{eq:CRB_sb}
% when launching orthogonal waveforms and what implications this has for optimal array design. 
\ac{crb}-optimal array designs in the case of orthogonal waveforms\textcolor{\edited}{,} a topic addressed next. 
\section{
CRB for orthogonal waveforms: Role of sum co-array and optimal array redundancy
}\label{sec:orthogonalWaveforms}
% As discussed in the previous section, estimating the unknown angle $\omega$ in an optimal way, requires a transmit waveform that depends on that same $\omega$. This dependency creates a contradiction that makes it impractical to deploy the optimal waveform for angle estimation.
% Instead, in a practical scenario, an initial \emph{orthogonal} waveform could be used to broadcast power in all directions. When a first estimate of the angle is gathered, the focus could narrow down to target a specific area, ultimately approaching a configuration similar to a sum-beam for precise angle estimation.
% Because of its practical relevance, in this section, we focus on the optimal array geometry given orthogonal waveforms. 
Let $\Sorth$ denote an orthogonal waveform matrix, such that
\begin{align}
   \Sorth\in\{\mathbf{S}\in\mathbb{C}^{T\times N_{\tx}}\ | \ \mathbf{S}^{\HT}\mathbf{S} = \tfrac{1}{N_{\tx}}\mathbf{I}_{N_{\tx}}\}.
    \label{eq:orthWav}
\end{align}
%The CRB \eqref{eq:singleTargetCRB} reveals an interesting structure in this case \cite{boyer2011performance}.
In this case, the single-target \ac{crb} \eqref{eq:singleTargetCRB} %be shown to 
depends on the spatial variances of both the \ac{tx} and \ac{rx} arrays, %\cite{boyer2011performance}, 
rather than solely the \ac{rx} array as for optimal (coherent) waveforms \cite{boyer2011performance}:
%\textcolor{red}{[Seems to be known result. Cite \cite{boyer2011performance} and reduce to lemma.]}
%\begin{lemma}[Single-target \ac{crb} for orthogonal waveforms \cite{boyer2011performance}]\label{thm:crb_orth}
    % Given orthogonal waveforms \eqref{eq:orthWav}, %the single-target 
    % \ac{crb} \eqref{eq:singleTargetCRB} can be written as
%
\begin{align}
    \textrm{\normalfont CRB}_{\omega} (\Sorth)
    & = \frac{\sigma^2}{2|\gamma|^2}\frac{1}{N_{\rx}}
    \left[ 
     \chi(\mathcal{D}_{\tx})
    + \chi(\mathcal{D}_{\rx})
    \right]^{-1}.
    \label{eq:CRBorth}
\end{align}
Expressions \eqref{eq:CRBorth} and \eqref{eq:CRB_sb} highlight the distinct roles of the \emph{physical} (\ac{tx}/\ac{rx}) array geometry when using coherent and orthogonal waveforms, respectively.\footnote{The ratio between the \ac{crb}s for coherent \eqref{eq:CRB_sb} and orthogonal \eqref{eq:CRBorth} waveforms is $N_{\tx}^{-1}\frac{\chi(\mathcal{D}_{\tx})+\chi(\mathcal{D}_{\rx})}{\chi(\mathcal{D}_{\rx})}\!\propto\!N_{\tx}^{-1}$ (recall $0\!\leq\!\chi(\mathcal{D}_{\tx})\!<\! \chi(\mathcal{D}_{\rx})$), as beamforming concentrates a factor of $N_{\tx}$ more power in the target direction.} 
However, the \emph{virtual} so-called \emph{sum co-array} is also known to be a key factor influencing identifiability and resolution in active sensing \ac{mimo} systems \cite{li2007mimoradar,rajamaki2023importance}. %---a fact not readily apparent from \eqref{eq:CRBorth}. 
% Yet, it is not immediately clear %We now elucidate 
% how the sum co-array and the repetition of its virtual elements (known as array redundancy) impacts \eqref{eq:CRBorth}, and thereby the fundamentally achievable \ac{mse} in active sensing. We now elucidate this connection.
Yet, the precise impact of the sum co-array and the repetition of its virtual elements (known as array redundancy) on \eqref{eq:CRBorth}, and thus on the fundamentally achievable \ac{mse}, remains unclear. In the following, we clarify this relationship.

\subsection{
Role of sum co-array and virtual element multiplicities
}
% \textcolor{red}{[Relevance of sum co-array needs to be introduced properly (also using signal model).]}

%The estimation performance in active sensing depends on the interaction between the \ac{tx} and \ac{rx} arrays via the so-called 
Formally, the sum co-array is defined as \cite{hoctor1990theunifying}
\begin{align}
    \mathcal{D}_\Sigma\triangleq\mathcal{D}_{\tx}+\mathcal{D}_{\rx}=\{d_{\tx}+d_{\rx}\ |\ d_{\tx}\in\mathcal{D}_{\tx};d_{\rx}\in\mathcal{D}_{\rx}\}.
    \label{def:sumcoarray1D}
\end{align}
Using \eqref{def:sumcoarray1D}, measurement model \eqref{eq:signal_model} can be expressed in terms of the $N_\Sigma\!\triangleq\!|\mathcal{D}_\Sigma|\!\leq\!N_{\tx}N_{\rx}$ unique virtual elements in $\mathcal{D}_\Sigma\!=\!\{d_\Sigma[\ell]\}_{\ell=1}^{N_\Sigma}$ by rewriting the \ac{tx}–\ac{rx} steering vector %$\mathbf{a}_\textrm{tr}(\omega)$ 
as \cite{rajamaki2020hybrid}
\begin{align*}
    %\mathbf{a}_\textrm{tr}(\omega) 
    %& = 
    \mathbf{a}_\textrm{t}(\omega) \otimes \mathbf{a}_\textrm{r}(\omega) 
    % \notag \\
    % & 
    = \bm{\Upsilon}\mathbf{a}_\Sigma(\omega).
\end{align*}
Here, $[\mathbf{a}_\Sigma(\omega)]_\ell = e^{jd_\Sigma[\ell]\omega}$ is the sum co-array response vector and $\bm{\Upsilon}\in\{0,1\}^{N_{\tx}N_{\rx}\times N_{\Sigma}}$ is the redundancy pattern matrix 
%consolidating identical virtual elements and 
accounting for 
%their 
the fact that
%\emph{multiplicities} of 
% possible repetition of the
% virtual elements, 
% as 
\emph{several pairs of physical elements $(d_{\tx},d_{\rx})$ can give rise to the same sum co-array element} $d_\Sigma\in\mathcal{D}_\Sigma$. 
% Hence, the cardinality of the sum co-array %$N_\Sigma \triangleq |\mathcal{D}_\Sigma|$ 
% satisfies $N_\Sigma \leq N_{\tx}N_{\tx}$, where an array is said to be \emph{redundant} if $N_\Sigma <N_{\tx}N_{\rx}$ and \emph{nonredundant} if $N_\Sigma=N_{\tx}N_{\rx}$. 
% mapping the virtual elements arising from the Kronecker product to the sum co-array elements, taking into account the multiplicity.
This enables expressing the CRB in \eqref{eq:CRBorth} equivalently in terms of the sum co-array and the \emph{multiplicities} of its elements, $\upsilon(d_\Sigma)\triangleq \sum_{m,n}\mathbbm{1}(d_{\tx}[m]+d_{\rx}[n]=d_\Sigma)$. 
% We define the \emph{multiplicity-weighted} spatial variance and mean of the sum co-array as:
% \begin{align}
%     % \tilde{\mu}(\mathcal{D}_\Sigma) 
%     % & \triangleq \frac{1}{N_{\tx} N_{\rx}}\sum_{d\in\mathcal{D}_\Sigma} d \upsilon_\Sigma(d), \\
%     % \tilde{\chi}(\mathcal{D}_\Sigma) 
%     % & \triangleq \frac{1}{N_{\tx} N_{\rx}} \sum_{d\in\mathcal{D}_\Sigma} (d-\tilde{\mu}(\mathcal{D}_\Sigma))^2 \upsilon_\Sigma(d).
%     \tilde{\chi}(\mathcal{D}_\Sigma) 
%     & \triangleq \frac{1}{N_{\tx} N_{\rx}} \sum_{d\in\mathcal{D}_\Sigma} (d-\tilde{\mu}(\mathcal{D}_\Sigma))^2 \upsilon_\Sigma(d),\\
%         \tilde{\mu}(\mathcal{D}_\Sigma) 
%     & \triangleq \frac{1}{N_{\tx} N_{\rx}}\sum_{d\in\mathcal{D}_\Sigma} d \upsilon_\Sigma(d),\nonumber
% \end{align}
% where 
% %The spatial mean and variance are now weighted by the \emph{multiplicity} 
% $\upsilon_\Sigma(d_{\Sigma})$ 
% is the multiplicity 
% of virtual element $d_\Sigma\in\mathcal{D}_\Sigma$.%We denote the \textit{multiplicity} of $d_\Sigma\in\mathcal{D}_\Sigma$ by $\upsilon_\Sigma (d_\Sigma)$.
%As we will show, with orthogonal waveforms the \ac{crb} reduces to the spatial variance of the sum co-array, weighted by the element multiplicities.
\begin{proposition}\label{thm:CRBorthWavSumcoArraydependency}
    % \cref{eq:CRBorth} can be rewritten %in terms of the multiplicity-weighted spatial variance of the sum co-array 
    % as 
    Given orthogonal transmit waveforms \eqref{eq:orthWav}, the single-target \ac{crb} \eqref{eq:singleTargetCRB} can be written as
    %Given orthogonal transmit waveforms \eqref{eq:orthWav}, the \ac{crb} is fully described by the multiplicity-weighted spatial variance of the sum co-array,%\footnote{A similar dependence on (only) the sum co-array has also been reported in \cite{werf2024transmit}, where the focus was to estimate the reflection coefficients of multiple targets.},
\begin{align}
    \textrm{\normalfont CRB}_{\omega}(\Sorth) 
    % & = \frac{\sigma^2}{2|\gamma|^2}\frac{1}{N_{\rx}}
    % \left[ 
    %  \chi(\mathcal{D}_{\tx})
    % + \chi(\mathcal{D}_{\rx})
    % \right]^{-1}  
    % \notag\\
    & 
    = \frac{\sigma^2}{2|\gamma|^2} \frac{1}{N_{\rx}}
    \tilde{\chi}^{-1}(\mathcal{D}_{\Sigma}),\label{eq:CRBlinearArrayOrth}
\end{align}
%Moreover, the optimal array geometry has a redundant sum co-array if both $N_{\tx}\geq 3$ and $N_{\rx}\geq 3$.
where $\tilde{\chi}(\mathcal{D}_\Sigma) $ is the multiplicity-weighted spatial variance, %of the sum co-array, defined as:
\begin{align}
    % \tilde{\mu}(\mathcal{D}_\Sigma) 
    % & \triangleq \frac{1}{N_{\tx} N_{\rx}}\sum_{d\in\mathcal{D}_\Sigma} d \upsilon_\Sigma(d), \\
    % \tilde{\chi}(\mathcal{D}_\Sigma) 
    % & \triangleq \frac{1}{N_{\tx} N_{\rx}} \sum_{d\in\mathcal{D}_\Sigma} (d-\tilde{\mu}(\mathcal{D}_\Sigma))^2 \upsilon_\Sigma(d).
    \tilde{\chi}(\mathcal{D}_\Sigma) 
    & \triangleq \frac{1}{N_{\tx} N_{\rx}} \sum_{d_\Sigma\in\mathcal{D}_\Sigma} (d_\Sigma-\tilde{\mu}(\mathcal{D}_\Sigma))^2 \upsilon(d_\Sigma),\label{eq:spatial_variance_sumcoarray}
\end{align}
with $\tilde{\mu}(\mathcal{D}_\Sigma) \triangleq \frac{1}{N_{\tx} N_{\rx}}\sum_{d_\Sigma\in\mathcal{D}_\Sigma} d_\Sigma \upsilon(d_\Sigma)$ the weighted spatial mean, and
%The spatial mean and variance are now weighted by the \emph{multiplicity} 
$\upsilon(d_{\Sigma})$ 
the multiplicity 
of virtual element $d_\Sigma\in\mathcal{D}_\Sigma$.
\end{proposition}
\begin{proof}\let\proof\relax\let\endproof\relax
    See \cref{app:spatialVarianceSumCoArray}.
    % We can show that (see Appendix \ref{app:spatialVarianceSumCoArray} for the derivation),
    % \begin{align}
    %     \chi(\mathcal{D}_{\tx}) + \chi(\mathcal{D}_{\rx}) = \tilde{\chi}(\mathcal{D}_{\Sigma}).
    % \end{align}
    % Additionally, using Lemma~\ref{thm:optArrayOrthWav}, the sum co-array of two clustered arrays is redundant, when both $N_{\tx}\geq 3$ and $N_{\rx}\geq 3$.
\end{proof}

% \begin{remark}
    \cref{thm:CRBorthWavSumcoArraydependency} reveals 
    % the role of the sum co-array in the (single-target) CRB given orthogonal waveforms: 
    % array geometries with identical sum co-arrays and virtual sensor multiplicities achieve the same CRB.
    that array geometries with \emph{identical sum co-arrays and virtual sensor multiplicities} achieve the same (single-target) \ac{crb} given orthogonal waveforms. For instance, this is the case for the arrays in \cref{fig:clusteredRxContiguousSumCoArray,subfig:NestedArray}. % therefore attain identical \ac{crb}s.
    % \textcolor{red}{[Robin: This point is made more clearly by modifying \cref{fig:clusteredRxContiguousSumCoArray} to have same co-array as \cref{subfig:NestedArray}.]} 
    Note that for optimal (coherent) waveforms, the \ac{crb} \eqref{eq:CRB_sb} does not depend on the sum co-array but only the \ac{rx} array.
% \end{remark}

%The cardinality of the sum co-array %$N_\Sigma \triangleq |\mathcal{D}_\Sigma|$ satisfies $N_\Sigma \leq N_{\tx}N_{\tx}$, where 
An array is said to be \emph{redundant} if $N_\Sigma\!<\!N_{\tx}N_{\rx}$ and nonredundant if $N_\Sigma\!=\!N_{\tx}N_{\rx}$. 
Nonredundant arrays are conventionally employed when launching orthogonal waveforms to maximize the size of the sum co-array, and thereby the number of spatial \ac{dof}. %(cf. \cref{sec:motivation}). %, for a given number of physical \ac{tx} and \ac{rx} sensors. 
Interestingly, however, %\cref{thm:CRBorthWavSumcoArraydependency} implies that 
nonredundant arrays need not minimize CRB, as we now show.% which runs counter to conventional wisdom, as nonredundant arrays 

\subsection{Optimal array geometry for orthogonal waveforms}
% Our goal is to determine the array geometry that results in the minimum single target \ac{crb}, taking into account physically relevant constraints on both the number of antenna elements and the physical aperture. Constraining the maximum \ac{tx} and \ac{rx} apertures to $L_{\tx}$ and $L_{\rx}$, and the number of \ac{tx} and \ac{rx} sensors to $N_{\tx}$ and $N_{\rx}$, respectively, optimal array geometries for orthogonal waveforms are then solutions to 
The array geometry minimizing the single-target \ac{crb}, given orthogonal waveforms and constraints on both the physical aperture and number of antennas, is the solution to 
%Constraining the maximum \ac{tx} and \ac{rx} apertures to $L_{\tx}$ and $L_{\rx}$, and the number of \ac{tx} and \ac{rx} sensors to $N_{\tx}$ and $N_{\rx}$, respectively, optimal array geometries for orthogonal waveforms are then solutions to 
\begin{align}
\underset{
    %\substack{
    \mathcal{D}_{\tx},\mathcal{D}_{\rx}\subset \mathbb{N} 
    %\\ \mathbf{S}\in\mathbb{C}^{T\times N_{\tx}}
    % }
}{\txt{min}}\ \textrm{CRB}_{\omega}(\Sorth) \
\txt{s.t.}\ 
    \begin{cases}
        |\mathcal{D}_{\tx}|=N_{\tx}, |\mathcal{D}_{\rx}|=N_{\rx},\\
        %\|\mathbf{S}\|_{\F}^2\leq 1, \\ 
        \max \mathcal{D}_{\tx}\!-\!\min \mathcal{D}_{\tx}\leq L_{\tx},\\
        \max \mathcal{D}_{\rx}\!-\!\min \mathcal{D}_{\rx}\leq L_{\rx}.\\
        % \chi(\mathcal{D}_{\rx})\geq\chi(\mathcal{D}_{\tx}).
        %\mathbf{S}^{\HT}\mathbf{S} = \frac{1}{N_{\tx}}\mathbf{I}_{N_{\tx}}.
    \end{cases}\label{eq:minCRBorth}
\end{align}
% \begin{subequations}
% \begin{align}
%     \min_{\mathcal{D}_{\tx},\mathcal{D}_{\rx}\subset\mathbb{N}} & \quad %\frac{\sigma^2}{2|\gamma|^2}
%     \frac{1}{|\mathcal{D}_{\rx}|}
%     \left[ 
%      \chi(\mathcal{D}_{\tx})
%     + \chi(\mathcal{D}_{\rx})
%     \right]^{-1} \label{eq:minCRBorth}\\
%     \textrm{s.t.} \ 
%     & \quad 
%     \max \mathcal{D}_{\tx}\!-\!\min \mathcal{D}_{\tx}\leq L_{\tx},
%     \label{eq:DtxApertureConstraint}
%     \\
%     & \quad 
%     \max \mathcal{D}_{\rx}\!-\!\min \mathcal{D}_{\rx}\leq L_{\rx}, \label{eq:DrxApertureConstraint}
%     \\
%     & \quad 
%     |\mathcal{D}_{\tx}| \leq N_{\tx}, \: |\mathcal{D}_{\rx}| \leq N_{\rx}. \label{eq:TotalNumberOfSensors}
% \end{align}
% \end{subequations}
\textcolor{\edited}{
The optimization problems throughout this work primarily involve discrete sensor locations $\mathcal{D}_{\tx}$, $\mathcal{D}_{\rx}$ and antenna counts $N_{\tx}$, $N_{\rx}$. These problems are inherently nonconvex due to their discrete domains. However, by reformulating the objective functions in terms of the arrays' spatial variances, we identify optimal configurations analytically.
% In general, optimization problems where optimization variables include the sets $\mathcal{D}_{\tx}$ and $\mathcal{D}_{\rx}$ or $N_{\tx}$ and $N_{\rx}$ are nonconvex due to the fact that these variables are integer valued. However, by expressing the \ac{crb} in terms of the spatial variances of the \ac{tx} and/or \ac{rx} array, it becomes possible to identify optimal configurations. 
% Although the problem in \eqref{eq:minCRBorth} is nonconvex, we establish its solution in the following corollary.
}
\begin{corollary}%[Optimal array geometry for orthogonal waveforms]
\label{thm:optArrayOrthWav}
    %Suppose that orthogonal waveforms are transmitted (i.e. $\mathbf{S}^{\HT}\mathbf{S}=\frac{1}{N_{\tx}}\mathbf{I}$), that the physical aperture is constrained by $L_{\tx}$ and $L_{\rx}$ and that the number of sensors are constrained to $N_{\tx}$ and $N_{\rx}$ for the \ac{tx} and \ac{rx} array, respectively. 
  %  Suppose that \eqref{eq:minCRBorth} is feasible for given $N_{\tx},N_{\rx},L_{\tx},L_{\tx}\in\mathbb{N}_+$, where $N_{\tx}\leq L_{\tx}+1$ and $N_{\rx}\leq L_{\rx}+1$. Then the solution $(\mathcal{D}^\star_{\tx},\mathcal{D}^\star_{\rx})$ to \eqref{eq:minCRBorth} is 
     Given $N_{\tx},N_{\rx},L_{\tx},L_{\rx}\in\mathbb{N}_+$, where $N_{\tx}\leq L_{\tx}+1$ and $N_{\rx}\leq L_{\rx}+1$, the solution $(\mathcal{D}^\star_{\tx},\mathcal{D}^\star_{\rx})$ to \eqref{eq:minCRBorth} is 
    \begin{align}
        \mathcal{D}^\star_{\tx} = \mathcal{K}_{N_{\tx}}^{L_{\tx}}, \quad \mathcal{D}^\star_{\rx} = \mathcal{K}_{N_{\rx}}^{L_{\rx}}. %\quad \mathbf{S} = \frac{1}{\sqrt{N_{\tx}}}\mathbf{U},
        \label{eq:optArrayGeometryOrthogonal}
    \end{align}
    Moreover, the sum co-array of \eqref{eq:optArrayGeometryOrthogonal} is redundant if $N_{\tx},N_{\rx}\geq 3$.
%where $\mathbf{U}\in\mathbb{C}^{T\times N_{\tx}}$ is an arbitrary matrix with orthonormal columns.
\end{corollary}
\begin{proof}
% Note that minimizing the \ac{crb} for orthogonal waveforms, \eqref{eq:CRBlinearArrayOrth}, for a fixed number of transmit and receive sensors is equivalent to separately maximizing the transmit and receive variance. That is, \eqref{eq:minCRBorth} is equal to
%Note that the feasible set of \eqref{eq:minCRBorth} is nonempty, since $(\mathcal{D}_{\tx},\mathcal{D}_{\rx})=(\mathcal{U}_{N_{\tx}},\mathcal{U}_{N_{\rx}})$ is a feasible point when $N_{\tx}\leq L_{\tx}+1$ and $N_{\rx}\leq L_{\rx}+1$. 
By \eqref{eq:CRBorth}, %\cref{thm:crb_orth}, 
\eqref{eq:minCRBorth} is equivalent to separately maximizing the \ac{tx} and \ac{rx} spatial variances:
\begin{align*}
% \underset{
%     \mathcal{D}_{\tx}\subset \mathbb{N} 
% }{\txt{max}}\ \chi(\mathcal{D}_{\tx}) \
% \txt{s.t.}\ 
%         |\mathcal{D}_{\tx}|=N_{\tx}, 
%         \max \mathcal{D}_{\tx}\!-\!\min \mathcal{D}_{\tx}\leq L_{\tx}, \\
% \underset{
%     \mathcal{D}_{\rx}\subset \mathbb{N} 
% }{\txt{max}}\ \chi(\mathcal{D}_{\rx}) \
% \txt{s.t.}\ 
%         |\mathcal{D}_{\rx}|=N_{\rx}, 
%         \max \mathcal{D}_{\rx}\!-\!\min \mathcal{D}_{\rx}\leq L_{\rx}.
\mathcal{D}_{\tx}^\star=\argmax_{
    \mathcal{D}_{\tx}\subset \mathbb{N} 
}\chi(\mathcal{D}_{\tx})\ \txt{s.t.}\ 
        |\mathcal{D}_{\tx}|=N_{\tx}, 
        \max \mathcal{D}_{\tx}\!-\!\min \mathcal{D}_{\tx}\leq L_{\tx}, \\
\mathcal{D}_{\rx}^\star=\argmax_{
    \mathcal{D}_{\rx}\subset \mathbb{N} 
}\chi(\mathcal{D}_{\rx})\ \txt{s.t.}\ 
        |\mathcal{D}_{\rx}|=N_{\rx}, 
        \max \mathcal{D}_{\rx}\!-\!\min \mathcal{D}_{\rx}\leq L_{\rx}.
\end{align*}
By \cite[Lemma 1]{vanderWerf2025Jointly}, the optimal solutions are given by \eqref{eq:clustered}, which directly yields \eqref{eq:optArrayGeometryOrthogonal} under the stated conditions.
%\footnote{Note that the characteristic of a clustered array maximizing spatial variance for an (linear) array does not directly extend to the spatial variance of the sum co-array, as the sum co-array can contain repeated instances of certain elements.} 
% minimized by maximizing the spatial variance of the sum co-array, or equivalently, by maximizing the spatial variance of the \ac{tx} and \ac{rx} array jointly. Since the clustered array maximizes the spatial covariance for a given number of sensors (\cite[Lemma 1]{vanderWerf2025Jointly}), it follows that the \ac{crb} is minimized if both the \ac{tx} and \ac{rx} arrays are clustered arrays.
% \textcolor{red}{[Robin: proof of redundancy]} 
When both $N_{\tx}\geq 3$ and $N_{\rx}\geq 3$, we have $\{0,1\}\!\subseteq\!\mathcal{D}_{\tx}$ and $\{0,1\}\!\subseteq\!\mathcal{D}_{\rx}\!\implies\!\upsilon(1\!=\!0\!+\!1\!=\!1\!+\!0)\!\geq\!2$, i.e., the array is redundant. 
%that $\{d_{\tx},d_{\rx}\}= \{0,1\}$ and $\{d_{\tx},d_{\rx}\}= \{1,0\}$ lead to the same co-array element and hence the array is redundant.
\end{proof}
Note that the (double clustered) array in \cref{thm:optArrayOrthWav} is an optimal array, not only for orthogonal waveforms, but also for coherent waveforms since the \ac{rx} array is a clustered array, as discussed in \cref{subsec:existingCRBoptWavArrayDesigns}.

% \cref{subfig:clusteredTxRx} illustrates the optimal array %geometry in 
% \eqref{eq:optArrayGeometryOrthogonal}. 
% %Interestingly, despite its widespread use, the canonical \ac{mimo} array does not minimize the (single-target) \ac{crb}, even when transmitting orthogonal waveforms. In fact, \cref{thm:optArrayOrthWav} and \cref{thm:CRBorthWavSumcoArraydependency} reveal that the optimal array geometry for orthogonal waveforms is inherently redundant.
% % \begin{remark}
% %     The optimal array geometry has a redundant sum co-array if both $N_{\tx}\geq 3$ and $N_{\rx}\geq 3$.
% % \end{remark}
% % \begin{proof}
% %     Using Lemma~\ref{thm:optArrayOrthWav}, the sum co-array of two clustered arrays is redundant, when both $N_{\tx}\geq 3$ and $N_{\rx}\geq 3$.
% % \end{proof}
% For comparison, \cref{subfig:NestedArray} shows a widely-considered %canonical 
% nonredundant array \cite{bliss2007mimo,li2007mimoradar,chen2008mimoradarspace}:
\cref{subfig:clusteredTxRx} illustrates the optimal array %geometry in 
\eqref{eq:optArrayGeometryOrthogonal}, and \cref{subfig:NestedArray} a widely-considered %canonical 
nonredundant array \cite{bliss2007mimo,li2007mimoradar,chen2008mimoradarspace}:
    \begin{align}
        \mathcal{D}_{\tx}\!=\!\mathcal{U}_{N_{\tx}},\ 
        \mathcal{D}_{\rx}\!=\!N_{\tx}\:\mathcal{U}_{N_{\rx}},\label{eq:canonicalmimo}
    \end{align} 
which has a nested structure consisting of a standard \ac{ula} (\ac{tx}) and dilated \ac{ula} (\ac{rx}). 
% \textcolor{red}{[Robin: Also showing the multiplicities of the co-array elements in \cref{fig:ClustTxRxvsNested} would connecting this subsection better to \cref{thm:CRBorthWavSumcoArraydependency}].} 
Interestingly, the %widely-considered \cite{bliss2007mimo,li2007mimoradar,chen2008mimoradarspace} 
``canonical MIMO array'' in \eqref{eq:canonicalmimo} does \emph{not} minimize the \ac{crb} given 
% Note that \eqref{eq:optArrayGeometryOrthogonal} is also \emph{an} optimal array geometry when launching optimal coherent waveforms \eqref{eq:optTxWaveform_sum} and $N_{\tx}\leq N_{\rx}$. %An example of such an array is shown in \cref{subfig:clusteredTxRx}.
% This is in contrast to the beamforming case, where the \ac{tx} array can be freely designed to maximize identifiability (cf. \cref{fig:clusteredRxContiguousSumCoArray}). Thus, for 
orthogonal waveforms. %there exists a 
This reveals a 
fundamental trade-off between \ac{crb} and 
%identifiability. 
the number of spatial \ac{dof}. 
The next subsection examines this trade-off and its implications in the multi-target scenario.
% \textcolor{red}{[Robin: This subsection could use a strong conclusion (perhaps also an example if space allows). A key reason for (??) is rewriting the CRB in terms of the redundancy of the sum co-array. Indeed, the CRB minimizing array is actually redundant, not nonredundant, as we later show. This may be surprising to the reader and should therefore be highlighted.]}
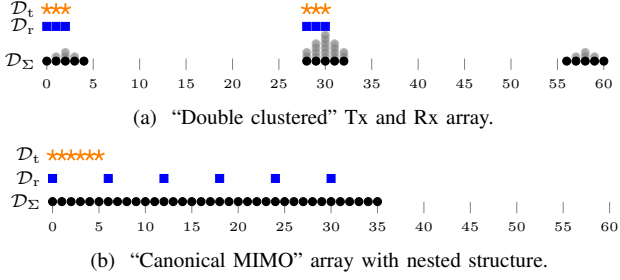
\begin{figure}
    \newcommand{\dataLoc}{Data/ClusteredTxRx}
    \newcommand{\xmax}{60}
    \newcommand{\xmin}{0}
    \newcommand{\fwidth}{9}
    \newcommand{\msize}{1.5}
	\pgfplotsset{every x tick label/.append style={font=\tiny},every y tick label/.append style={font=\scriptsize}}
    \subfloat[
    ``Double clustered'' \ac{tx} and \ac{rx} array.
    %the optimal array geometry for orthogonal waveforms with a redundant sum co-array. Its multiplicities are given by $\mathbf{v} = {[1,2,3,2,1,2,4,6,4,2,1,2,3,2,1]}$.
    ]{
    \centering
		\begin{tikzpicture}
			\begin{axis}[width= \fwidth cm ,height=2.5 cm,xmin=\xmin-0.2,xmax=\xmax+0.2,xtick={\xmin,\xmin+5,...,\xmax},ytick={-2,0,1},yticklabels={$\mathcal{D}_\Sigma$,$\mathcal{D}_{\rx}$,$\mathcal{D}_{\tx}$},ytick style={draw=none},ymin=-2.5,ymax=1.5,xticklabel shift = 0 pt,xtick pos=bottom,axis line style={draw=none}]
				\addplot[orange,only marks,line width = 0.7,mark=star,mark size=2.2,y filter/.code={\pgfmathparse{\pgfmathresult*0+1}}] table[x=D,y=D]{\dataLoc/Dt.dat};
				\addplot[blue,only marks,mark=square*,mark size=\msize,y filter/.code={\pgfmathparse{\pgfmathresult*0}}] table[x=D,y=D]{\dataLoc/Dr.dat};
                \addplot[gray,only marks,mark=*,mark size=\msize, 
                fill opacity=0.7,
                draw opacity=0.5,
                y filter/.code={\pgfmathparse{\pgfmathresult}},x filter/.code={\pgfmathparse{\pgfmathresult}}] coordinates {
                (1, -1.75) 
                (2, -1.75) (2, -1.5)
                (3,-1.75) 
                (28, -1.75) (28, -1.5)
                (29, -1.75) (29, -1.5) (29, -1.25) (29, -1) 
                (30, -1.75) (30, -1.5) (30, -1.25) (30, -1) (30, -0.75) (30, -0.5)
                (31, -1.75) (31, -1.5) (31, -1.25) (31, -1) 
                (32, -1.75) (32, -1.5)
                (57, -1.75) 
                (58, -1.75) (58, -1.5)
                (59,-1.75) 
                };
				\addplot[only marks,mark=*,mark size=\msize,y filter/.code={\pgfmathparse{\pgfmathresult*0-2}},x filter/.code={\pgfmathparse{\pgfmathresult}}] table[x=D,y=D]{\dataLoc/D_Sigma.dat};
                % (16, -1) (17, -1) (18, -1) (30, -1) (31, 1) (32, 1) (33, 1) (45, 1) (46, 1) (47, 1) (48, 1) };
			\end{axis}
		\end{tikzpicture}
        \label{subfig:clusteredTxRx}
        }%
    \vspace{-0.25cm}
    \renewcommand{\dataLoc}{Data/NestedNtx6Nrx6}
    \\
	\centering
    \subfloat[
    %Canonical \ac{mimo} array with a nonredundant contiguous sum co-array, i.e., $\mathbf{v} = \mathbf{1}^{\T}_{N_{\tx}N_{\rx}}$.
   ``Canonical \ac{mimo}'' array with nested structure.
    ]{
    \begin{tikzpicture}
			\begin{axis}[width= \fwidth cm ,height=2.5 cm,xmin=\xmin-0.2,xmax=\xmax+0.2,xtick={\xmin,\xmin+5,...,\xmax},ytick={-1,0,1},yticklabels={$\mathcal{D}_\Sigma$,$\mathcal{D}_{\rx}$,$\mathcal{D}_{\tx}$},ytick style={draw=none},ymin=-1.5,ymax=1.5,xticklabel shift = 0 pt,xtick pos=bottom,axis line style={draw=none}]
				\addplot[orange,only marks,line width = 0.7,mark=star,mark size=2.2,y filter/.code={\pgfmathparse{\pgfmathresult*0+1}}] table[x=D,y=D]{\dataLoc/Dt.dat};
				\addplot[blue,only marks,mark=square*,mark size=\msize,y filter/.code={\pgfmathparse{\pgfmathresult*0}}] table[x=D,y=D]{\dataLoc/Dr.dat};
				\addplot[only marks,mark=*,mark size=\msize,y filter/.code={\pgfmathparse{\pgfmathresult*0-1}},x filter/.code={\pgfmathparse{\pgfmathresult}}] table[x=D,y=D]{\dataLoc/D_Sigma.dat};
			\end{axis}
		\end{tikzpicture}
        \label{subfig:NestedArray}
    }%
    % \vspace{-0.1cm}
    \caption{
    % Given a physical aperture constraint, $L = 30$, $N_{\tx}=6$ \ac{tx} and $N_{\rx}=6$ \ac{rx} sensors, (a) the optimal array geometry for orthogonal waveforms, i.e., the minimizer of \eqref{eq:minCRBorth}, (b) the nested array.
    Key array geometries illustrated for $N_{\tx}=N_{\rx}=6$ \ac{tx}/\ac{rx} sensors, and physical aperture $L = 30$. The optimal array geometry for orthogonal waveforms \protect\subref{subfig:clusteredTxRx} %, i.e., the minimizer of \eqref{eq:minCRBorth}, and 
    has a redundant sum co-array (multiplicities shown in grey). 
    %Its multiplicities are given by $\mathbf{v} = {[1,2,3,2,1,2,4,6,4,2,1,2,3,2,1]}$.
    The conventional \ac{mimo} array \protect\subref{subfig:NestedArray} has a nonredundant contiguous sum co-array.%, i.e., $\mathbf{v} = \mathbf{1}^{\T}_{N_{\tx}N_{\rx}}$.
    }
    \vspace{-0.5cm}
    \label{fig:ClustTxRxvsNested}
\end{figure}

\subsection{Implications in multi-target scenario}
The inherent redundancy of the optimal array geometry for orthogonal waveforms improves \ac{crb} at the expense of fewer distinct and contiguous sum co-array elements $N_\Sigma$, which reduces the maximum number of identifiable targets \cite{bliss2003mimo,li2007onparameter}.%\textcolor{red}{[REF]}.

\begin{proposition}[Sum co-array of clustered array]\label{thm:clusteredsca}
    Let $N_{\tx}$ and $N_{\rx}$ be even. 
    If 
    $L_{\tx}=L_{\rx}=L\geq \max(N_{\tx},N_{\rx})+\min(N_{\rx},N_{\tx})/2-2$ 
   % $L_{\tx}=L_{\rx}>N_{\tx}+N_{\rx}$, 
    then the sum co-array of the optimal array %clustered %\ac{tx} and \ac{rx} arrays in 
    \eqref{eq:optArrayGeometryOrthogonal} has cardinality
    \begin{align}
        N_{\Sigma} = N_{\tx}+N_{\rx} + \max(N_{\tx},N_{\rx})-3.\label{eq:sca_size}
    \end{align} 
    %distinct elements, implying that at most $N_{\Sigma}/2$ targets are uniquely identifiable.
    Moreover, if 
%    $L_{\tx} +1 \leq N_{\tx} + N_{\rx}/2$ and $L_{\rx} +1 \leq N_{\tx}/2 + N_{\rx}$, then 
    $L\leq \max(N_{\tx},N_{\rx})+\min(N_{\rx},N_{\tx})/2-2$ then 
    $\mathcal{D}_\Sigma=\mathcal{U}_{2L}$, i.e., 
    the sum co-array is %contains 
   % $N_\Sigma$ 
    %virtual elements are 
    contiguous.
    %elements.
\end{proposition}
 \begin{proof}\let\proof\relax\let\endproof\relax
   See \cref{app:ProofOfclusteredsca}.
\end{proof}

Hence, the optimal array \eqref{eq:optArrayGeometryOrthogonal} has $N_\Sigma = \mathcal{O}(\max(N_{\tx},N_{\rx}))$ sum co-array elements, compared to \(N_{\tx}N_{\rx}\) for a nonredundant array, such as the canonical \ac{mimo} array \eqref{eq:canonicalmimo}. 
The reduction in $N_\Sigma$, and thereby identifiability, is the price paid for a lower \ac{crb}. 
Scenarios prioritizing low estimation error may 
{\color{\edited} therefore} 
%prefer
consider designs informed by the clustered array, while those requiring high identifiability would benefit from nonredundant designs. 
%We note that analytically characterizing CRB-optimal arrays and waveforms in the multi-target case is still largely unexplored, and thereby a pertinent direction for future work.% \textcolor{red}{[Add a two source MSE simulation to \cref{sec:simulations}?]}
%We note that practical estimation performance (e.g., of the \ac{mle}) at moderate or low \ac{snr} arrays is still largely unexplored, and thereby a pertinent direction for future work.
The practical estimation performance of these array geometries is compared numerically in \cref{sec:simulations}. 
 %{\color{\edited}
%We note that deriving array geometries that minimize \ac{crb} (or \ac{mse}, more generally) in the multi-target scenario remains an important and challenging open problem for future work.}

% However, despite its widespread use, the canonical \ac{mimo} array does not minimize the (single-target) CRB, even when transmitting orthogonal waveforms. In fact, the optimal array is inherently redundant. This redundancy reduces identifiability by limiting the number of distinct sum co-array elements. %This trade-off is expected: minimizing the \ac{crb} prioritizes local sensitivity rather than global resolution.

%\subsection{Implications for non-asymptotic regime}
% \textcolor{red}{[\textbf{To do.}]}
% \begin{itemize}
%     \item base line for asymptotic (high SNR)
%     \item 
% \end{itemize}

Thus far our discussion has assumed a given number of \ac{tx} and \ac{rx} elements. However, the number of available \ac{rf} chains might only constrain the \emph{combined} number of \ac{tx} $+$ \ac{rx} sensors \cite{liu2025doaestimation}. This raises two questions addressed in the next section: for a \emph{total sensor budget}, 
\begin{enumerate*}[label=(\roman*)]
    \item what is the optimal allocation between \ac{tx} and \ac{rx} elements, and
    \item how much can the \ac{crb} be improved via optimal allocation?
\end{enumerate*}

\section{Optimal sensor allocation between Tx and Rx%Optimal (unequal) allocation of \ac{tx} and \ac{rx} sensors%: Advantages of unequal allocation% can provide lower CRB
%Optimal allocation of sensors between \ac{tx} and \ac{rx} arrays
}\label{sec:optimalAllocation}

This section investigates 
optimal transmit and receive sensor allocations minimizing the \ac{crb} 
%how to optimally allocate sensors between the transmit and receive arrays 
in case of both orthogonal and coherent (optimal) transmit waveforms. Specifically,
given sensor budget $N$ and class of array geometries $\mathcal{C}$, we solve:
\begin{align}
\underset{
\substack{
    N_{\tx},N_{\rx}\in\mathbb{N}_+\\
    \mathcal{D}_{\tx},\mathcal{D}_{\rx}\subset\mathbb{N}}
}{\txt{min.}}\  
\text{CRB}_{\omega}(\mathbf{S})\ 
\txt{s.t.} 
\begin{cases}
    N_{\tx}+N_{\rx} \leq N\\
     (\mathcal{D}_{\tx},\mathcal{D}_{\rx})\!\in\!\mathcal{C}.
\end{cases}
\label{eq:minCRBDivideNgen}
\end{align}
% \textcolor{\edited}{The problem in \eqref{eq:minCRBDivideNgen} is nonconvex. However, we can find its solution by expressing the \ac{crb} in terms of the spatial variance of the \ac{tx} and \ac{rx} arrays.}
The following two classes of array geometries are considered: 
\begin{enumerate}
    \item Canonical MIMO array:
    \begin{align}
        \mathcal{C}%(N_{\tx},N_{\rx})
        \!=\!%\mathcal{C}_{\mathrm{m}}(N_{\tx},N_{\rx})\triangleq
        \{(\mathcal{D}_{\tx},\mathcal{D}_{\rx}):  \mathcal{D}_{\tx}\!=\!\mathcal{U}_{N_{\tx}},\ 
        \mathcal{D}_{\rx}\!=\!N_{\tx}\:\mathcal{U}_{N_{\rx}}\}.\label{eq:class_canonicalmimo}
    \end{align} 
    \item Double-clustered (optimal) array in \eqref{eq:optArrayGeometryOrthogonal} of aperture $L$:
    \begin{align}
        \mathcal{C}%(N_{\tx},N_{\rx})
        =\{(\mathcal{D}_{\tx},\mathcal{D}_{\rx}) :  \mathcal{D}_{\tx} = \mathcal{K}_{N_{\tx}}^{L}, \mathcal{D}_{\rx} = \mathcal{K}_{N_{\rx}}^{L}\}. \label{eq:class_doubleclustered}
    \end{align}
\end{enumerate}
% The canonical MIMO array \eqref{eq:class_canonicalmimo} serves as an important benchmark due to its wide-spread use in MIMO radar \cite{bliss2007mimo,li2007mimoradar,chen2008mimoradarspace}. The configuration consists of a standard \ac{ula} (\ac{tx}) and dilated \ac{ula} (\ac{rx}), as illustrated in 
% \cref{subfig:NestedArray} for $N_{\tx} = 5$ and $N_{\rx} = 7$. 

For the canonical MIMO array, $N_{\tx}$ and $N_{\rx}$ are conventionally chosen to be approximately equal ($N_{\tx}=\lfloor N/2\rfloor$, $N_{\rx}=N-N_{\tx}$) to maximize the size of the sum co-array (which is nonredundant and contiguous). %given a total sensor budget. 
This maximizes the number of spatial degrees of freedom available when transmitting independent\textcolor{\edited}{,} typically orthogonal\textcolor{\edited}{,} waveforms. Surprisingly, however, this widely employed sensor allocation strategy yields \emph{suboptimal} \ac{crb} both in case of \labelcref{eq:class_canonicalmimo,eq:class_doubleclustered}, %the canonical MIMO and double-clustered array, 
as we now show.% Similarly, we also show that equal sensor allocation is suboptimal in case of the double-clustered (optimal) array.

\subsection{Orthogonal waveforms: Equal allocation is suboptimal}\label{sec:alloc_orth}
Given orthogonal waveforms, $\mathbf{S}^{\HT}\mathbf{S}\!=\!\frac{1}{N_{\tx}}\mathbf{I}$, \eqref{eq:minCRBDivideNgen} becomes
\begin{equation}
\underset{
   \substack{
    N_{\tx},N_{\rx}\in\mathbb{N}\\
    \mathcal{D}_{\tx},\mathcal{D}_{\rx}\subset\mathbb{N}}
}{\txt{min}} 
%\frac{\sigma^2}{2|\gamma|^2}
\frac{1}{N_{\rx}}
    \left[ 
     \chi(\mathcal{D}_{\tx})\!+\!\chi(\mathcal{D}_{\rx})
    \right]^{-1}
\txt{s.t.}
\begin{cases}
    N_{\tx}+N_{\rx} \leq N\\
     (\mathcal{D}_{\tx},\!\mathcal{D}_{\rx})\!\in\!\mathcal{C}.%(N_{\tx},\!N_{\rx})
\end{cases}
\label{eq:minCRBDivideNorth}
\end{equation}
% \textcolor{\edited}{While the problem in \eqref{eq:minCRBDivideNorth} is (still) nonconvex, the optimal sensor allocation can be identified, as shown in the next theorem.}
%\subsubsection{Canonical MIMO array}
\begin{theorem}[Canonical MIMO array] \label{thm:divideSensorsOrthMimo}
Let $\mathbf{S}^{\HT}\mathbf{S}\!=\!\frac{1}{N_{\tx}}\mathbf{I}$ and $\mathcal{C}$ be given by \eqref{eq:class_canonicalmimo}\textcolor{\edited}{, with $N$ divisible by $5$ and even.} As $N\rightarrow\infty$, \eqref{eq:minCRBDivideNorth} has the following solution: %for the Canonical MIMO array \eqref{eq:class_canonicalmimo}:
    \begin{align}
        N_{\tx}^\star=\frac{2}{5}N, \ N_{\rx}^\star=\frac{3}{5}N.\label{eq:optalloc_mimo}
    \end{align}
\end{theorem}
\begin{proof}
Given \eqref{eq:class_canonicalmimo}, problem \eqref{eq:minCRBDivideNorth} 
%Assuming the array geometry is fixed to the canonical MIMO array, i.e., \eqref{eq:canonicalMIMOarrayTxRx}, then \eqref{eq:minCRBorthDivideNOrth} 
can be rewritten as
\begin{equation}
    \min_{N_{\tx},N_{\rx}\in\mathbb{N}_+}
    \frac{1}{N_{\rx}}
    \left[
    \chi\left(\mathcal{U}_{N_{\tx}}\right)\!+\!\chi\left(N_{\tx}\mathcal{U}_{N_{\rx}}\right)
    \right]^{-1}\ 
    \textrm{s.t.}\
    N_{\tx}\!+\!N_{\rx}\!\leq\!N.
\label{eq:divideNoverMIMOarray}
\end{equation}
First, note that by \eqref{eq:SpatialVariance}, for any $a\geq 0$ and finite set $\mathcal{D}$:
\begin{align}
    \chi\left(a\mathcal{D}\right) = a^2\chi\left(\mathcal{D}\right).\label{eq:chi_const}
\end{align}
Furthermore, the $n$-sensor \ac{ula} can be verified to satisfy:%in case of a ç with $N$ sensors:
\begin{align}
    \chi\left(\mathcal{U}_n\right) = \frac{1}{12}(n^2-1).
    \label{eq:spatialVarianceULA}
\end{align}
Using \labelcref{eq:chi_const,eq:spatialVarianceULA}, and setting $N_{\tx} = \theta N$ and $N_{\rx}=N-N_{\tx} = (1-\theta)N$, for $\theta\in[0,1]$, the objective function in \eqref{eq:divideNoverMIMOarray} asymptotically approaches
\begin{align}
    \lim_{N\rightarrow \infty}\frac{1}{N_{\rx}}
    \frac{1}{
     \chi\left(\mathcal{U}_{N_{\tx}}\right) + N_{\tx}^2\chi\left(\mathcal{U}_{N_{\rx}}\right)
    }
    & = \lim_{N\rightarrow \infty} \frac{1}{N_{\rx}\frac{1}{12}(N_{\tx}^2 N^2_{\rx} -1)} \notag \\
    & = \lim_{N\rightarrow \infty} \frac{12}{\theta^2(1-\theta)^3}N^{-5} \notag \\
    & \geq \lim_{N\rightarrow \infty} 10\cdot\frac{625}{18}N^{-5}
    \label{eq:rateCRBMIMOarrayOrthWav}
\end{align}
Here equality is reached when $\theta = \frac{2}{5}$.
\end{proof}
\textcolor{\edited}{
The assumption in \cref{thm:divideSensorsOrthMimo} that $N$ is even and divisible by $5$ is made to simplify the exposition and ensure that $N_{\tx}^\star=\frac{2}{5}N$, and $N_{\rx}^\star=\frac{3}{5}N$ are integers. Extensions to other values are possible %but include the use of floor/ceiling operators which clutter the asymptotic analysis without changing the physical result.
but do not substantially change the result.
}

\cref{thm:divideSensorsOrthMimo} shows that equal sensor allocation, %that maximizes the aperture of the sum co-array of the canonical MIMO array, 
i.e., $N_{\tx} =N_{\rx} = \frac{1}{2}N$ (when $N$ is even), does \emph{not} minimize the \ac{crb} when orthogonal waveforms are used. Indeed, by \eqref{eq:optalloc_mimo}, using more \ac{rx} than \ac{tx} sensors is optimal. This can be attributed to the factor $N_{\rx}$ in the denominator of the CRB expression \eqref{eq:CRBorth}. %\eqref{eq:CRBlinearArrayOrth}. 
%\subsubsection{Double-clustered (optimal) array}
\textcolor{\edited}{This result aligns with a similar observation made by \cite{wang2019further} concerning sensor allocation in nested %and co-prime 
arrays in the context of 
passive sensing.%, illuminating the impact of the waveform on sensor allocation in active sensing.
}

\begin{remark}\label{rem:allocation_crb_improvement}
Given orthogonal waveforms, the canonical \ac{mimo} array can improve the \ac{crb} by $10\%$ by allocating \ac{tx}/\ac{rx} sensors optimally, i.e., unequally instead of equally, since
% For orthogonal waveforms, the canonical \ac{mimo} array can reduce the \ac{crb} by about $10\%$ through optimal (i.e., unequal) allocation of \ac{tx}/\ac{rx} sensors:
\begin{align*}
\lim_{N\to\infty}
\frac{\mathrm{CRB}_{\omega}(\Sorth)\big|_{N_{\tx} = \frac{1}{2}N}^{N_{\rx} = \frac{1}{2}N}}{\mathrm{CRB}_{\omega}(\Sorth)\big|_{N_{\tx} = \frac{2}{5}N}^{N_{\rx} = \frac{3}{5}N}}
=
    \frac{(\frac{1}{2})^{-2} (\frac{1}{2})^{-3}}{(\frac{2}{5})^{-2} (\frac{3}{5})^{-3}} = \frac{355}{321} \approx 1.106.
\end{align*}
%
% At the same time, optimal allocation for the canonical \ac{mimo} array does \emph{not} maximize the number of sum co-array elements 
However, this optimal allocation does \emph{not} maximize the number of sum co-array elements 
$N_\Sigma(N_{\tx},N_{\rx})\!=\!N_{\rx}N_{\tx}$, as $N_\Sigma(\frac{2}{5} N,\frac{3}{5}N)/N_\Sigma(\frac{1}{2} N,\frac{1}{2}N)
\!=\!24/25\!\approx\!0.96$.
% , i.e., optimal allocation yields $\approx\!4\%$ fewer spatial \ac{dof} (corresponding to a  smaller virtual array aperture). 
% Hence, a gain of $10\%$ in \ac{crb} is traded off to a loss of $4\%$ in \ac{dof}.
Hence, a $10\%$ improvement in \ac{crb} comes at the cost of roughly $4\%$ fewer spatial \ac{dof} (i.e., a smaller virtual aperture).
%
% Moreover, the \ac{crb} of the canonical \ac{mimo} array with equi-allocation can be improved $\approx 5.5\times$ via optimal array design alone (with equal sensor allocation) and $\approx 11\times$ via jointly optimal array design and (unequal) sensor allocation.\footnote{Here, $\tfrac{12/(1/2)^5}{625/9}\!\approx\!5.5$ and $\tfrac{12/(1/2)^5}{625/18}\!\approx\!11$ follow from setting $\theta\!=\!\tfrac{1}{2}$ in \eqref{eq:rateCRBMIMOarrayOrthWav}, together with $\theta\!=\!\tfrac{1}{2}$ and $\theta\!=\!0$ in \eqref{eq:rateCRBClusteredOrthWav}, respectively.}
\end{remark}

A similar result holds for the double-clustered array \eqref{eq:class_doubleclustered}, when the aperture $L$ is set %of the double-clustered array 
to match that of the canonical MIMO array \eqref{eq:class_canonicalmimo} with optimal allocation \eqref{eq:optalloc_mimo}.%, for a fair comparison.
\begin{theorem}[Double-clustered array] \label{thm:divideSensorsOrthOpt}
Let $\mathbf{S}^{\HT}\mathbf{S}\!=\!\frac{1}{N_{\tx}}\mathbf{I}$, $\mathcal{C}$ be given by \labelcref{eq:class_doubleclustered}, and $L = \frac{2}{5}N(\frac{3}{5}N-1)$, with $N$ divisible by $5$ and even. As $N\rightarrow\infty$, \eqref{eq:minCRBDivideNorth} has the following solution:
%for the Double-clustered (optimal) array in \labelcref{eq:class_doubleclustered} with $L = \frac{2}{5}N(\frac{3}{5}N-1)$:
    \begin{align}
        N_{\tx}^\star = 2, \ N_{\rx}^\star = N-2.
    \end{align}
\end{theorem}
\begin{proof}
%Using \cite[Lemma 1]{vanderWerf2025Jointly}, \eqref{eq:minCRBorthDivideNOrth} simplifies to
Given \eqref{eq:class_doubleclustered}, problem \eqref{eq:minCRBDivideNorth} simplifies to
\begin{equation}
%\begin{split}
    \min_{N_{\tx},N_{\rx}\in\mathbb{N}_+} %& \quad %\frac{\sigma^2}{2|\gamma|^2}
    \frac{1}{N_{\rx}}
    \left[ 
     \chi\left(\mathcal{K}_{N_{\tx}}^{L}\right)\!+\!\chi\left(\mathcal{K}_{N_{\rx}}^{L}\right)
    \right]^{-1}
    \textrm{s.t.}\ 
    % & \quad 
    % \max \mathcal{D}_{\tx}\!-\!\min \mathcal{D}_{\tx}\leq L_{\tx},
    % %\label{eq:DtxApertureConstraint}
    % \\
    % & \quad 
    % \max \mathcal{D}_{\rx}\!-\!\min \mathcal{D}_{\rx}\leq L_{\rx}, %\label{eq:DrxApertureConstraint}
    % \\
    %& \quad 
    N_{\tx}\!+\!N_{\rx}\!\leq\!N. %\label{eq:TotalNumberOfSensors}
%\end{split}
%\label{eq:divideNoverclusteredarraysOrthWav}
\nonumber
\end{equation}
%Again, we fix the available aperture to the aperture used by the canonical MIMO array when the sensors are optimally allocated, i.e., $L = \frac{2}{5}N(\frac{3}{5}N-1)$.
%The spatial variance of a clustered array $\mathcal{K}_N^L$ depends on the number of sensors $N$ and its aperture $L$. 
For an even $m > 2$, a simple calculation shows that
\begin{align}
    \chi(\mathcal{K}_m^L)\!=\!
    \frac{1}{4}L^2\!+\!\frac{1}{12}m^2\!-\!\frac{1}{4}mL\!+\!\frac{1}{2}L\!-\!\frac{1}{4}m\!+\!\frac{1}{6}.
    \label{eq:spatialVarianceClusteredArray}
\end{align}
Using \eqref{eq:spatialVarianceClusteredArray} and $L = \frac{2}{5}N(\frac{3}{5}N-1)$, we can write that $\chi(\mathcal{K}^L_{N_{\rx}}) = \frac{9}{625}N^{4}+\mathcal{O}(N^3)$. Setting $N_{\tx} = \theta N$ and $N_{\rx}=N-N_{\tx} = (1-\theta)N$, for $\theta\in[0,1]$, then yields
\begin{align}
    \lim_{N\rightarrow\infty} \frac{1}{N_{\rx}} \left[\chi(\mathcal{K}^L_{N_{\tx}}) + \chi(\mathcal{K}^L_{N_{\rx}})\right]^{-1}
    & = \lim_{N\rightarrow\infty}   \frac{1}{(1-\theta)}\frac{625}{18}N^{-5}\notag \\
    & \geq \lim_{N\rightarrow\infty} \frac{625}{18}N^{-5}.
    \label{eq:rateCRBClusteredOrthWav}
\end{align}
Here equality is reached when $\theta\rightarrow 0$. Hence, since %we require 
$N_{\tx},N_{\rx}>1$ by assumption, $N_{\tx} = 2$ and $N_{\rx} = N-2$ will minimize the \ac{crb} asymptotically for increasing $N$.
\end{proof}
% To ensure a fair comparison between the double-clustered array and the canonical MIMO array, we set the aperture $L$ %of the double-clustered array 
% in \cref{thm:divideSensorsOrthOpt} to match that of \eqref{eq:class_canonicalmimo} with $N_{\tx},N_{\rx}$ given by \eqref{eq:optalloc_mimo}. 
\textcolor{\edited}{
Similar to \cref{thm:divideSensorsOrthMimo}, the assumption in \cref{thm:divideSensorsOrthOpt} that $N$ is even and divisible by $5$ is made to simplify the exposition and ensure the aperture $L = \frac{2}{5}N(\frac{3}{5}N-1)$ is an integer.
}

\begin{remark}%\label{rem:allocation_crb_improvement}
The \ac{crb} of the canonical \ac{mimo} array with equi-allocation 
%(using orthogonal waveforms) 
can be improved 
$%(12/(1/2)^5)/(625/9)
\approx 5.5\times$ via optimal array design alone (%double-clustered array
with equal sensor allocation) and 
$%(12/(1/2)^5)/(625/18)
\approx 11\times$ via jointly optimal array design and (unequal) sensor allocation.\footnote{Here, $\tfrac{12/(1/2)^5}{625/9}%(12/(1/2)^5)/(625/9)
\!\approx\!5.5$ and $\tfrac{12/(1/2)^5}{625/18}%(12/(1/2)^5)/(625/18)
\!\approx\!11$ follow from setting $\theta\!=\!\tfrac{1}{2}$ in \eqref{eq:rateCRBMIMOarrayOrthWav}, together with $\theta\!=\!\tfrac{1}{2}$ and $\theta\!=\!0$ in \eqref{eq:rateCRBClusteredOrthWav}, respectively.}% and(double-clustered array with unequal allocation).
This gain in \ac{crb} comes at the cost of a factor $\min(N_{\tx},N_{\rx})$ fewer spatial \ac{dof}, since $N_\Sigma=\mathcal{O}(\max(N_{\tx},N_{\rx}))$ instead of $N_\Sigma=N_{\tx}N_{\rx}$.
\end{remark}

\subsection{Coherent (optimal) waveforms: Equal allocation is optimal}
Assuming coherent waveforms \eqref{eq:optTxWaveform_sum} are optimal, i.e., beamforming in the target direction when \eqref{eq:linearArrayDesignRule} holds, \eqref{eq:minCRBDivideNgen} becomes:
\begin{equation}
\underset{
    \substack{
    N_{\tx},N_{\rx}\in\mathbb{N}_+\\
    \mathcal{D}_{\tx},\mathcal{D}_{\rx}\subset \mathbb{N}}
}{\txt{min.}}\ 
%\frac{\sigma^2}{2|\gamma|^2}
\frac{1}{N_{\tx}N_{\rx}}\chi^{-1}(\mathcal{D}_{\rx})\
\txt{s.t.}\ 
\begin{cases}
    N_{\tx}+N_{\rx} \leq N\\
     (\mathcal{D}_{\tx},\mathcal{D}_{\rx})\!\in\!\mathcal{C}.%(N_{\tx},N_{\rx})
\end{cases}
\label{eq:minCRBorthDivideN}
\end{equation}
Equal sensor allocation turns out to be optimal in this case for both the canonical MIMO array \eqref{eq:class_canonicalmimo} and double-clustered array \eqref{eq:class_doubleclustered}; thus striking a balance between \ac{tx} beamforming gain and the number of independent measurements. This is in stark contrast to the orthogonal waveform case (\cref{sec:alloc_orth}), where \ac{rx} sensors are preferred over \ac{tx} sensors.%Hence, equal ensor allocation between \ac{tx} and \ac{rx} strikes a balance between \ac{tx} beamforming gain and independent measurements, respectively.
\begin{proposition}[Coherent waveforms]\label{thm:divideSensorsOpt}
    Suppose $\mathbf{S}$ is given by \eqref{eq:optTxWaveform_sum} and $N$ is even. As $N\rightarrow\infty$, \eqref{eq:minCRBorthDivideN} has the following solution in case of the canonical MIMO array \eqref{eq:class_canonicalmimo}: 
    \begin{align}
        N_{\tx}^\star = N_{\rx}^\star = \frac{1}{2}N.
        \label{eq:equalSensorDistr}
    \end{align}
   % Restricting the available aperture to the aperture used by the optimal canonical MIMO array, i.e., $L = \frac{1}{2}N(\frac{1}{2}N-1)$, the sensor distribution in \eqref{eq:equalSensorDistr} is also the solution to \eqref{eq:minCRBorthDivideN} in case of the double-clustered array \eqref{eq:class_doubleclustered}.
    This is also the solution %to \eqref{eq:minCRBorthDivideN} 
    in case of the double-clustered array \eqref{eq:class_doubleclustered} with aperture $L = \frac{1}{2}N(\frac{1}{2}N-1)$, i.e., the same aperture as the canonical MIMO array using sensor allocation \eqref{eq:equalSensorDistr}.
    % Suppose $\mathbf{S}$ is given by \eqref{eq:optTxWaveform_sum}, and $L = \frac{1}{2}N(\frac{1}{2}N-1)$ with $N=2n$ even. As $n\rightarrow\infty$, \eqref{eq:minCRBorthDivideN} has the following solution both in case of the 
    % \begin{enumerate*}
    %     \item Canonical MIMO array \eqref{eq:class_canonicalmimo}, and
    %     \item Double-clustered array \eqref{eq:class_doubleclustered}:
    % \end{enumerate*}
    %     \begin{align}
    %     N_{\tx}^\star = N_{\rx}^\star = \frac{1}{2}N.
    % \end{align}
\end{proposition}
\begin{proof}
\begin{enumerate}
    \item Canonical MIMO array:
    By \labelcref{eq:chi_const,eq:class_canonicalmimo}, the problem in \eqref{eq:minCRBorthDivideN} simplifies to
\begin{equation}
\begin{split}
    \min_{N_{\tx},N_{\rx}\in\mathbb{N}_+} & \quad
    \frac{1}{N^3_{\tx}N_{\rx}}
    \chi^{-1}\left(\mathcal{U}_{N_{\rx}}\right)\
    \textrm{s.t.}\ 
    N_{\tx} + N_{\rx} \leq N.
\end{split}
\label{eq:divideNoverMIMOarrayOptimalWav}
\end{equation}
Using \eqref{eq:spatialVarianceULA} and setting $N_{\tx}=\theta N$ and $N_{\rx} = (1-\theta)N$, where $\theta \in [0,1]$, %the objective function of \eqref{eq:divideNoverMIMOarrayOptimalWav}
we can write
    \begin{align}
        \lim_{N\rightarrow\infty} \frac{1}{N^3_{\tx}N_{\rx}}
    \chi^{-1}\left(\mathcal{U}_{N_{\rx}}\right)
        & = \lim_{N\rightarrow\infty} \frac{12}{N^3_{\tx}N_{\rx} (N^2_{\rx}-1)} \notag\\
        & = 
        \lim_{N\rightarrow\infty} \frac{12}{(\theta (1-\theta))^3} N^{-6} \notag \\
        & \geq  
        \lim_{N\rightarrow\infty} 3\cdot 256 N^{-6}.
        \label{eq:rateCRBMIMOarrayOptimalWav}
    \end{align}
Here equality is reached when $\theta = \frac{1}{2}$. 
\item Double-clustered array:
In this case, 
%Since the clustered array maximizes the spatial variance for any fixed number of sensors and aperture \cite[Lemma 1]{vanderWerf2025Jointly}, the problem in 
\eqref{eq:minCRBorthDivideN} simplifies to 
%optimally dividing the available sensors between any \ac{tx} array and a clustered \ac{rx} array:
%
\begin{equation}
\begin{split}
    \min_{N_{\tx},N_{\rx}\in\mathbb{N}_+}  %\frac{\sigma^2}{2|\gamma|^2}
    \frac{1}{N_{\tx}N_{\rx}} 
        \chi^{-1}\left(\mathcal{K}_{N_{\rx}}^{L}\right)
        \ 
        % \\
    \textrm{s.t.} 
    % \ 
    % & 
    \ 
    N_{\tx} + N_{\rx} \leq N. %\label{eq:TotalNumberOfSensors}
\end{split}
\label{eq:divideNoverclusteredarrays}
\end{equation}
% Note that the problems above are not completely symmetric as the objective functions only contain $1/N_{\rx}$ (and the $1/N_{\tx}$ is absent).
%
%We fix the available aperture to the aperture used by the canonical \ac{mimo} array when the sensors are optimally allocated, i.e. $L = \frac{1}{2}N(\frac{1}{2}N-1)$.
%We can then establish the optimal allocation between \ac{tx} and \ac{rx} sensors for the optimal array.
By \eqref{eq:spatialVarianceClusteredArray} and $L = \frac{1}{2}N(\frac{1}{2}N-1)$, % = \frac{1}{4}N^2 - \frac{1}{2}N$
    we have $\chi(\mathcal{K}^L_{N_{\rx}}) = \frac{1}{64}N^{4}+\mathcal{O}(N^3)$. Setting $N_{\tx} = \theta N$ and $N_{\tx}=N-N_{\rx} = (1-\theta)N$, for $\theta\in[0,1]$, then yields 
\begin{align}
    \lim_{N\rightarrow\infty} \frac{1}{N_{\tx}N_{\rx}\chi(\mathcal{K}^L_{N_{\rx}})} 
    & = \lim_{N\rightarrow\infty} \frac{64}{\theta (1-\theta)}N^{-6} \notag \\
    & \geq \lim_{N\rightarrow\infty} 256 N^{-6}.\label{eq:rateCRBClusteredOptimalWav}
\end{align}
Equality is achieved in the final inequality when $\theta = \frac{1}{2}$.
\end{enumerate}
\end{proof}
\textcolor{\edited}{
Again, the assumption that $N$ is even in \cref{thm:divideSensorsOpt} is made to simplify the exposition and ensure that $N_{\tx}^\star = N_{\rx}^\star = N/2$ is an integer.
}
%Hence, the array that maximizes the aperture of the sum co-array (i.e., $N_{\tx}N_{\rx}$) minimizes the \ac{crb} when an optimal transmit waveform is used.

\cref{tab:sensor_division} summarizes the optimal sensor allocations derived in \cref{thm:divideSensorsOrthMimo,thm:divideSensorsOrthOpt,thm:divideSensorsOpt}. 
% \cref{tab:rate} shows the asymptotic decay rate of the \ac{crb} for the optimal sensor allocations in \cref{tab:sensor_division} (cf. \labelcref{eq:rateCRBMIMOarrayOrthWav,eq:rateCRBClusteredOrthWav,eq:rateCRBMIMOarrayOptimalWav,eq:rateCRBClusteredOptimalWav}). %Note that the \ac{crb} attained by the clustered array is lower by a factor of $3$ and $10$ compared to the canonical \ac{mimo} array, for optimal and orthogonal waveforms, respectively.
%These results are asymptotical (total sensor budget $N\to\infty$). In \cref{sec:simulations} we will illustrate how these results translate to finite $N$.
%
%\textcolor{red}{[\cref{tab:rate} needs brief discussion before transitioning into the next subsection.]}
\begin{table}[!t]
\vspace{-.3cm}
\centering
\caption{
%Optimal sensor allocation for different array configurations and transmit waveforms. The total number of sensors is $N$. The arrays have the same aperture for coherent and orthogonal waveforms respectively.
Optimal sensor allocations for different array configurations and transmit waveforms (equal array aperture for a given waveform). The total number of sensors is $N$.
}
\label{tab:sensor_division}
\begin{tabular}{|c|c|c|c|}
\hline
\multicolumn{2}{|c|}{} & \multicolumn{2}{c|}{\textbf{Transmit waveform}} \\ 
\cline{3-4} 
\multicolumn{2}{|c|}{} & \textbf{Coherent (optimal)} & \textbf{Orthogonal} \\ 
\hline
\multirow{2}{*}{\rotatebox[origin=c]{90}{\textbf{Array}}} &{\textbf{Clustered (optimal)}} & $N_{\tx} = N_{\rx} = \frac{1}{2}N$ & $N_{\tx}\!=\!2$, $N_{\rx}\!=\!N\!-\!2$ \\ 
 [1ex]
\cline{2-4}
& \textbf{Canonical MIMO} & $N_{\tx}= N_{\rx} = \frac{1}{2}N$ & $N_{\tx}\!=\!\frac{2}{5}N$, $N_{\rx}\!=\!\frac{3}{5}N$ \\  [1ex]
% \cline{2-4}
% & \textbf{Co-prime array} & $N_{\tx} = N_{\rx} = \frac{1}{2}N$ & $N_{\tx} = \frac{2}{5}N$, $N_{\rx} = \frac{3}{5}N$ \\  [1ex]
\hline
\end{tabular}
\end{table}

\subsection{Impact of optimal sensor allocation on CRB for finite N%$N$% given total sensor budget
}
%The optimal allocation of sensors between the \ac{tx} and \ac{rx} arrays, which was asymptotically characterized in \cref{sec:optimalAllocation} (see \cref{tab:sensor_division}), is now numerically evaluated for finite $N$. 
To conclude this section, we numerically evaluate the optimal allocation of sensors between the \ac{tx} and \ac{rx} arrays for finite $N$. 
\cref{fig:sensorDivision} shows the \ac{crb} as a function of the number of transmit sensors $N_{\tx}$ for $N=40$ (with $N_{\rx}=N-N_{\tx}$) and $\tfrac{\sigma^2}{2|\gamma|}=1$. The aperture for the clustered array is set to $L = \frac{2}{5}N(\frac{3}{5}N-1)$ and $L = \frac{1}{2}N(\frac{1}{2}N-1)$ for optimal and orthogonal waveforms, respectively.
When orthogonal waveforms are used, the \ac{crb} is minimized at $N_{\tx}=2$ for the double-clustered array, and at $N_{\tx}=\tfrac{2}{5}N=16$ for the canonical \ac{mimo} array (dashed lines). For the optimal (beamforming) waveform, the minimum occurs at $N_{\tx}=N_{\rx}=\tfrac{1}{2}N=20$, for both clustered and canonical arrays (solid lines).
These results validate that the derived asymptotical results (\cref{tab:sensor_division}) also characterize the optimal sensor allocation well for finite $N$.
\begin{figure}
\vspace{-0.4cm}
    \centering
    \includegraphics[width=0.9\linewidth]{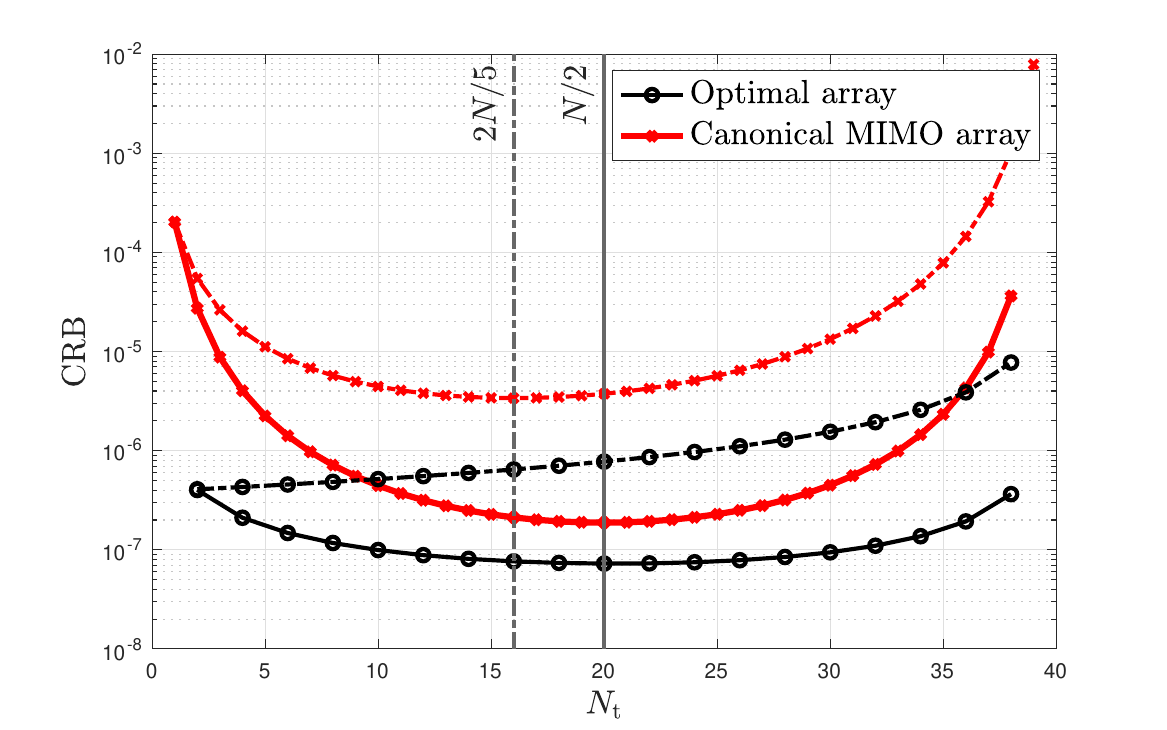}
    \vspace{-0.5cm}
    \caption{Single target \ac{crb} %(up to scaling factor) 
    versus number of transmit sensors $N_{\tx}$ %for various array geometries, where 
    for a total number of $N_{\tx}+N_{\rx}=N = 40$ sensors. 
    Solid lines correspond to optimal waveforms and dashed lines to orthogonal waveforms. 
    % \textcolor{red}{[Robin: Take-home message here]} 
   The (finite sensor) minimizers agree well with the derived asymptotical results (cf. \cref{tab:sensor_division}). 
%    The results confirm that the asymptotically optimal sensor allocations in \cref{tab:sensor_division} also minimize the \ac{crb} for finite number of sensors.
    }
    % \vspace{-0.5cm}
    \label{fig:sensorDivision}
\end{figure}

\cref{tab:rate} shows the asymptotic decay rate of the \ac{crb} for the optimal sensor allocations in \cref{tab:sensor_division} (cf. \labelcref{eq:rateCRBMIMOarrayOrthWav,eq:rateCRBClusteredOrthWav,eq:rateCRBMIMOarrayOptimalWav,eq:rateCRBClusteredOptimalWav}). 
Given the asymptotically optimal sensor allocations, the \ac{crb} attained by the clustered array \eqref{eq:clustered} is lower by a factor of $3$ and $10$ compared to the canonical \ac{mimo} array \eqref{eq:canonicalmimo}, for optimal and orthogonal waveforms, respectively (as $N \to \infty$). 
%(cf. \cref{tab:rate}). 
This also holds for moderate values of $N$, as shown in \cref{fig:CRBratio}, which plots the ratio of the \ac{crb}s of the canonical \ac{mimo} array and clustered array %for different waveforms 
%$\textrm{CRB}_{\omega}(\Sorth)/\textrm{CRB}_{\omega}(\mathbf{S}^\star)$, 
%i.e., $\frac{\textrm{CRB}_{\textrm{MIMO}}}{\textrm{CRB}_{\textrm{optimal}}}$, 
as a function of $N$. % Even for moderate values of $N$, the clustered array achieves a substantial performance gain.
The decay rate in \cref{tab:rate} increases from $N^{-5}$ to $N^{-6}$ when using optimal instead of orthogonal waveforms. This gain arises because beamforming concentrates a factor $N_{\tx}=\mathcal{O}(N)$ more power in the target direction.
\begin{table}[!t]
% \vspace{-.3cm}
\centering
\caption{Asymptotic \ac{crb} as $N\rightarrow\infty$ (up to common scaling factors) %$\frac{\sigma^2}{2|\gamma|}$
when sensors are optimally allocated between \ac{tx} and Rx.}
\label{tab:rate}
\begin{tabular}{|c|c|c|c|}
\hline
\multicolumn{2}{|c|}{} & \multicolumn{2}{c|}{\textbf{Transmit waveform}} \\ 
\cline{3-4} 
\multicolumn{2}{|c|}{} & \textbf{Optimal} & \textbf{Orthogonal} \\ 
\hline
\multirow{2}{*}{\rotatebox[origin=c]{90}{\textbf{Array}}} &{\textbf{Optimal}} & $256 N^{-6}$ & $\frac{625}{18}N^{-5}$ \\ 
 [1ex]
\cline{2-4}
& \textbf{Canonical MIMO} &$3\cdot 256 N^{-6}$ & $10\cdot\frac{625}{18} N^{-5}$ \\  [1ex]
\hline
\end{tabular}
% \begin{tablenotes}
%       \small
%       \item * $a = 256$, ** $b = \frac{628}{18}$
% \end{tablenotes}
\end{table}
\begin{figure}
\vspace{-.3cm}
    \centering
    \includegraphics[width=0.9\linewidth]{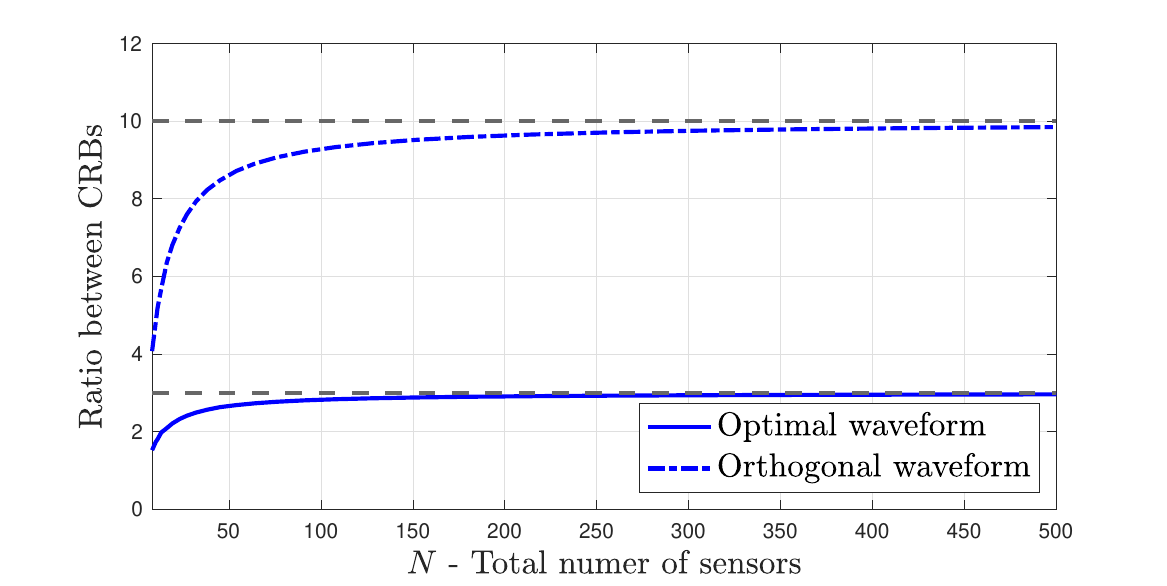}
    \vspace{-0.3cm}
    \caption{Ratio of \ac{crb}s (canonical \ac{mimo} over clustered array) %(up to scaling factor) 
    versus total number of available sensors when these are optimally allocated between the \ac{tx} and \ac{rx} arrays. %The solid lines are used for optimal waveforms, dashed lines are used for orthogonal waveforms. 
    % \textcolor{red}{[Robin: Take-home message here]}
    The ratio converges to the ratios of the asymptotic values listed in \cref{tab:rate} as $N\to\infty$, while the clustered array already provides a significant \ac{crb} advantage for moderate $N$.
    }
    \vspace{-0.3cm}
    \label{fig:CRBratio}
\end{figure}

\section{Planar arrays and spatial covariance% in minimizing the CRB
}\label{sec:planarArrays}
% In the previous sections, we developed a general understanding of the relationship between the transmit waveform and array geometry in the \textit{one-dimensional} case. However, in practical applications, such as automotive sensing~\textcolor{red}{[REF]}, \textit{two-dimensional} (planar) arrays are commonly employed. In this section, we extend the model and the corresponding conclusions to accommodate planar array geometries.
While the previous sections considered one-dimensional array geometries, in practice, two-dimensional arrays are of great interest in a wide range of radar and wireless communications applications, due to their ability to discriminate both azimuth and elevation angles. This section extends the results of \cref{sec:orthogonalWaveforms} to planar arrays, highlighting distinct challenges of the two-dimensional case arising from the fact that the CRB now depends on the spatial \emph{covariance matrix} of the array geometry, rather than simply a scalar-valued variance.%no longer depends on the (scalar) spatial variance, but

\subsection{Signal model and CRBM for planar arrays}
A planar array can estimate both azimuth $\theta$ and elevation $\phi$, 
%. The target angle can then be 
as parametrized by the two-dimensional spatial frequency vector $\bm{\omega} =  %\frac{2\pi}{\lambda} 
\pi[\sin(\theta)\cos(\phi), \ \cos(\theta)\cos(\phi)]^{\T}$. %, where $\theta$ and $\phi$ denote the azimuth and elevation angles, respectively.
Analogously to \eqref{eq:sensorpositions}, the \ac{tx} and \ac{rx} sensor positions (in half wavelengths) are
\begin{align*}
    % \mathcal{D}_{\tx} 
    % & = \{\mathbf{d}_{\tx}[n]\}_{n=1}^{N_{\tx}} = \left\{
    % \begin{bmatrix}
    %     d^{(1)}_{\tx}[n],
    %     d^{(2)}_{\tx}[n]
    % \end{bmatrix}^{\T}
    % \right\}_{n=1}^{N_{\tx}}, \\
    % \mathcal{D}_{\rx} 
    % & = \{\mathbf{d}_{\rx}[m]\}_{n=1}^{N_{\rx}} = \left\{
    % \begin{bmatrix}
    %     d^{(1)}_{\rx}[m],
    %     d^{(2)}_{\rx}[m]
    % \end{bmatrix}^{\T}
    % \right\}_{m=1}^{N_{\rx}}.
        \mathcal{D}_{\tx} 
    = \{\mathbf{d}_{\tx}[n]\}_{n=1}^{N_{\tx}},\quad
    \mathcal{D}_{\rx} 
    = \{\mathbf{d}_{\rx}[m]\}_{n=1}^{N_{\rx}},
\end{align*}
where $\mathbf{d}_{\tx}[n]\!=\![d^{(1)}_{\tx}[n],  d^{(2)}_{\tx}[n]]^{\T}$, $\mathbf{d}_{\rx}[m]\!=\![d^{(1)}_{\rx}[m],  d^{(2)}_{\rx}[m]]^{\T}$, and the respective array response vectors are $\left[\mathbf{a}_{\tx}(\bm{\omega})\right]_n = e^{j \mathbf{d}^{\T}_{\tx}[n] \bm{\omega}} = e^{j (d^{(1)}_{\tx}[1] \omega_1 + d^{(2)}_{\tx}[1] \omega_2)}$ and $\left[\mathbf{a}_{\rx}(\bm{\omega})\right]_n = e^{j \mathbf{d}^{\T}_{\rx}[n] \bm{\omega}} = e^{j (d^{(1)}_{\rx}[1] \omega_1 + d^{(2)}_{\rx}[1] \omega_2)}$.
The received signal model then becomes 
$
% \begin{equation}
% \begin{split}
    \mathbf{y} 
    =(\mathbf{S}\otimes \mathbf{I})(\mathbf{a}_{\tx}(\bm{\omega})\otimes \mathbf{a}_{\rx}(\bm{\omega}) )\gamma + \mathbf{n},%\quad \in\mathbb{C}^{N_{\rx}}.
%    \nonumber%\label{eq:signal_model_planar}
% \end{split}
% \end{equation}
$ 
which contrary to \eqref{eq:signal_model} depends on both azimuth and elevation.%Moreover, the model in \eqref{eq:signal_model_planar} coincides with \eqref{eq:signal_model} when $\phi = 0$. \textcolor{red}{Is that so?}

Since the target angle $\bm{\omega}$ is a two-dimensional vector, the %\ac{crb} (given waveform $\mathbf{S}$) is a 
$2\times 2$ \acf{crbm}, given waveform $\mathbf{S}$, becomes:
\begin{align}
    \txt{CRBM}_{\bm{\omega}}(\mathbf{S})\!=\!
    \frac{\sigma^2}{2|\gamma|^2}\!
    \big[
    \txt{Re}
    \big\{ 
    \mathbf{\dot{A}}^{\HT}_\textrm{tr}(\mathbf{S}^{\HT}\otimes \mathbf{I})
    \mathbf{P}_{(\mathbf{S}\otimes \mathbf{I})\mathbf{a}_\textrm{tr}}^\perp(\mathbf{S}\otimes \mathbf{I})\mathbf{\dot{A}}_\textrm{tr}
    \big\}
    \big]^{-1}.
    \nonumber%\label{eq:CRB_planar}
\end{align}
Here, the derivative of the array manifold is given by $\mathbf{\dot{A}}_\textrm{tr} = \frac{\partial}{\partial\bm{\omega}^{\T}} (\mathbf{a}_{\tx}\otimes\mathbf{a}_{\rx})$, which has dimensions $N_{\tx}N_{\rx} \times 2$. % Consequently, the \ac{crbm} given waveform $\mathbf{S}$ is a $2\times 2$ matrix
%has dimensions $2\times 2$.
%Note that ``minimizing" the \ac{crbm} is not uniquely defined. Typically, a scoring function is used that evaluates the matrix, i.e. $f: \mathbb{S}^2_+ \rightarrow \mathbb{R}_+$. Common choices include the A-optimality, D-optimality and E-optimality, that minimize the trace, maximum eigenvalue and the determinant of the \ac{crbm}, respectively (see e.g. \cite{boyd2004convexoptimization}).
% \textcolor{red}{[or do we require strictly decreasing?]}
%\subsection{Spatial covariance of planar array geometry}
% The concept of spatial variance in linear arrays is broadened to spatial \textit{co}variance in the context of planar arrays.
 % and $\mu^{(i)}(\mathcal{D})$ denotes the spatial mean in the $i$'th dimension.
%
% \begin{figure*}[htb]
% \hrule
% \begin{align}
%     \bm{\Sigma}(\mathcal{D}) 
%     & \triangleq 
%     \frac{1}{|\mathcal{D}|}
%         \sum_{[d^{(1)},d^{(2)}]^{\T}\in\mathcal{D}} 
%         \begin{bmatrix}
%             d^{(1)}-\mu^{(1)}(\mathcal{D}) \\
%             d^{(2)}-\mu^{(2)}(\mathcal{D})
%         \end{bmatrix}
%         \begin{bmatrix}
%             d^{(1)}-\mu^{(1)}(\mathcal{D}) &
%             d^{(2)}-\mu^{(2)}(\mathcal{D})
%         \end{bmatrix} 
%         = \frac{1}{|\mathcal{D}|}
%         \sum_{\mathbf{d}\in\mathcal{D}} 
%         (\mathbf{d}-\bm{\mu}(\mathcal{D}))(\mathbf{d}-\bm{\mu}(\mathcal{D}))^{\T}
%     \label{eq:SpatialCovariance}
% \end{align}
% \hrule
% \end{figure*}

\subsection{Optimality of sum beam: 
%Novel 
\textcolor{\edited}{Necessary and} 
sufficient condition
}
To optimize the \ac{crb} in the two-dimensional case, we define a scalar-valued objective function $f: \mathbb{S}^2_{++} \rightarrow \mathbb{R}_+$, where we assume $f(\mathbf{X}^{-1})$ to be convex on the positive semi-definite cone. Typical choices of $f$ include the trace, maximum eigenvalue and the determinant of the \ac{crbm} (known as A, D, and E-optimality, respectively \cite{boyd2004convexoptimization}). 
Given $f$, optimal waveforms are solutions to the optimization problem\textcolor{\edited}{\footnote{\textcolor{\edited}{The \ac{crbm} for the azimuth $\theta$ and elevation $\phi$ can be obtained via the standard Jacobian transformation \cite[Sec. 3.6]{kay1993fundamentals}. Depending on the choice of the objective function $f(\cdot)$, the optimal design parameters may become angle-dependent, and thus diverge from those that minimize $f(\text{CRBM}_{\bm{\omega}})$.}}}
\begin{align}
   \min_{\mathbf{S}} 
 %   \mathbf{S}^\star=\argmin_{\mathbf{S}\in\mathbb{C}^{T\times N_{\tx}}}
    \ f(\txt{CRBM}_{\bm{\omega}}(\mathbf{S})) \ \txt{s.t.} \ \trace(\mathbf{S}^{\HT}\mathbf{S}) \leq 1.
    \label{eq:minimizeCRBwrtWaveform_plan}
\end{align}

Similarly to \eqref{eq:TransmitWaveformCov},
%the \ac{crb} for linear arrays, 
the waveform covariance $\mathbf{R}_\mathrm{s}=\mathbf{S}\mathbf{S}^{\HT}$ of any solution to \eqref{eq:minimizeCRBwrtWaveform_plan} 
%the \ac{crbm} %for planar arrays 
%is solely a function of the transmit waveform covariance matrix $\mathbf{S}^{\HT}\mathbf{S}$, which 
takes the form~\cite{forsythe2005waveform,li2008range} (for $\lambda_1\geq 0$, $\bm{\Lambda}_2 \succeq \mathbf{0}$):
\begin{align}
%    \mathbf{S}^{\HT}\mathbf{S} 
    \mathbf{R}_\mathrm{s}\!=\! 
    \lambda_1 \frac{\mathbf{a}_{\tx}\mathbf{a}^{\HT}_{\tx}}{\|\mathbf{a}_{\tx}\|_2^2} 
    \!+\! 
    \frac{
        \mathbf{P}^{\perp}_{\mathbf{a}_{\tx}}\mathbf{\dot{A}}_{\tx}
        % \begin{bmatrix}
        %     \lambda_{2} & 0\\
        %     0 & \lambda_{3}
        % \end{bmatrix}
        {\bm{\Sigma}(\mathcal{D}_\textrm{t})}^{-\frac{1}{2}} 
        \bm{\Lambda}_{2}
        {\bm{\Sigma}(\mathcal{D}_\textrm{t})}^{-\frac{1}{2}} 
        \mathbf{\dot{A}}^{\HT}_{\tx}\mathbf{P}^{\perp}_{\mathbf{a}_{\tx}}
    }{
        %N_{\tx}
        \|\mathbf{a}_{\tx}\|_2^2
        % \|\mathbf{P}^{\perp}_{\mathbf{a}_{\tx}}\mathbf{\dot{A}}_{\tx}\|_2^2
        % \|\mathbf{P}^{\perp}_{\mathbf{a}_{\tx}}\mathbf{\dot{A}}_{\tx}\|_2^2
    }. \label{eq:optWavPlanar}
\end{align}
%where $\lambda_1\geq 0$ and $\bm{\Lambda}_2 \succeq \mathbf{0}$. 
Here, for a planar array $\mathcal{D}=\{\mathbf{d}_n\}_{n=1}^{N}$, $\bm{\Sigma}(\mathcal{D})$ naturally extends the notion of spatial variance to spatial \textit{co}variance %\footnote{The diagonal of $\bm{\Sigma}(\mathcal{D})$ contains the spatial variances of the individual dimensions.} 
as 
\begin{align}
    \bm{\Sigma}(\mathcal{D}) 
    \triangleq 
      \frac{1}{|\mathcal{D}|}
        \sum_{\mathbf{d}\in\mathcal{D}} 
        (\mathbf{d}-\bm{\mu}(\mathcal{D}))(\mathbf{d}-\bm{\mu}(\mathcal{D}))^{\T},
    \label{eq:SpatialCovariance}
\end{align}
%in \eqref{eq:SpatialCovariance}, 
where $\bm{\mu}(\mathcal{D})=\frac{1}{|\mathcal{D}|}\sum_{\mathbf{d}\in\mathcal{D}}\mathbf{d}\in\mathbb{R}^2$ is the spatial mean. 
Similarly to \eqref{eq:CRBequivObjective}, 
minimizing $f(\txt{CRBM}(\bm{\omega}))$ with respect to $\mathbf{S}$ 
can be shown to reduce  
%therefore boils down 
to solving the following problem\footnote{After some straightforward but tedious algebra, omitted here for brevity.}

% \textcolor{red}{[Ids: I have a derivation for this, but that is going to take up quite some space in the appendix...what should we do?]}
% \textcolor{red}{[Robin: Can you add it so that we can check it? We can then say that it is omitted due to being long and tedious.]}
\begin{align}
\begin{aligned}
    \hspace{-.3cm}\min_{\lambda_{1},\bm{\Lambda_{2}}} 
    & f
    \Big(
    \frac{\sigma^2}{2|\gamma|^2 N_{\tx}N_{\rx}}
        \big[
            \lambda_{1} \bm{\Sigma}(\mathcal{D}_{\rx})
            \!+\!
            % {\bm{\Sigma}_\textrm{t}}^{\frac{1}{2}}
            \bm{\Sigma}^{\frac{1}{2}}(\mathcal{D}_{\tx}) 
            % \textrm{diag}(\lambda_{2},\lambda_{3})
            \bm{\Lambda}_{2}
            % \begin{bmatrix}
            %     \lambda_{2} & 0\\
            %     0 & \lambda_{3}
            % \end{bmatrix}
            % {\bm{\Sigma}_\textrm{t}}^{\frac{1}{2}}
            \bm{\Sigma}^{\frac{1}{2}}(\mathcal{D}_{\tx}) 
        \big]^{-1}
    \Big)\\% \label{eq:objectiveFunctionCRBM}\\
    \textrm{s.t.} & \ \lambda_{1} + 
    % \lambda_{2} + \lambda_{3} 
    \textrm{trace}(\bm{\Lambda}_{2})
    = 1, 
    % \\
    % & 
    \ \lambda_{1} > 0, \ \bm{\Lambda}_{2}\succeq \mathbf{0}.% \nonumber
    % \\
    % & \ \mathbf{S}^{\HT}\mathbf{S} = 
    % \lambda_1 \frac{\mathbf{a}_{\tx}\mathbf{a}^{\HT}_{\tx}}{\|\mathbf{a}_{\tx}\|_2^2} 
    % + 
    % \frac{
    %     \mathbf{P}^{\perp}_{\mathbf{a}_{\tx}}\mathbf{\dot{A}}_{\tx}
    %     % \begin{bmatrix}
    %     %     \lambda_{2} & 0\\
    %     %     0 & \lambda_{3}
    %     % \end{bmatrix}
    %     {\bm{\Sigma}_\textrm{t}}^{-\frac{1}{2}}
    %     \bm{\Lambda}_{2}
    %     {\bm{\Sigma}_\textrm{t}}^{-\frac{1}{2}} 
    %     \mathbf{\dot{A}}^{\HT}_{\tx}\mathbf{P}^{\perp}_{\mathbf{a}_{\tx}}
    % }{
    %     %N_{\tx}
    %     \|\mathbf{a}_{\tx}\|_2^2
    %     % \|\mathbf{P}^{\perp}_{\mathbf{a}_{\tx}}\mathbf{\dot{A}}_{\tx}\|_2^2
    % }. \label{eq:optWavPlanar}
    \end{aligned}\label{eq:objectiveFunctionCRBM}
\end{align}
%
%The \ac{crbm} relies on the spatial \emph{co}variance of both the transmitter and receiver, $\bm{\Sigma}(\mathcal{D}_{\tx})$ and $\bm{\Sigma}(\mathcal{D}_{\rx})$, respectively. 
The %interaction between the 
%these two factors once more 
\ac{tx} and \ac{rx} spatial covariances ($\bm{\Sigma}(\mathcal{D}_{\tx})$, $\bm{\Sigma}(\mathcal{D}_{\rx})$)
determine the optimal transmit waveform, which is a linear combination of the \textcolor{\edited}{so-called ``sum'' and ``difference'' beams in \eqref{eq:optWavPlanar}.} 
%sum and difference beams. 
Unlike the one-dimensional case, multiple difference beams are possible in two-dimensions by linearly combining the columns of $\mathbf{P}^{\perp}_{\mathbf{a}_{\tx}}\mathbf{\dot{A}}_{\tx}$. 
Analogously to \eqref{eq:optTxWaveform_sum}, sum beam waveforms correspond to fully coherent transmission in the target direction $\bm{\omega}$:
\begin{align}
    \mathbf{S}^\star = \frac{\mathbf{u}\mathbf{a}_{\tx}^{\HT}(\bm{\omega})}{\sqrt{N_{\tx}}}, \quad \|\mathbf{u}\|_2 = 1.
    \label{eq:optTxWaveform_sumPlanar}
\end{align}
Given \eqref{eq:optTxWaveform_sumPlanar},
the single-target \ac{crbm} simplifies to (cf. \eqref{eq:CRB_sb})
\begin{align}
    \txt{CRBM}_{\bm{\omega}}(\mathbf{S}^\star) = \frac{\sigma^2}{2|\gamma|^2}\frac{1}{N_{\tx}N_{\rx}}\bm{\Sigma}^{-1}(\mathcal{D}_{\rx}).
    \label{eq:CRBMoptimalTxWaveFplanar}
\end{align}
Necessary and sufficient conditions for \eqref{eq:optTxWaveform_sumPlanar} being a solution to \eqref{eq:minimizeCRBwrtWaveform_plan} are currently unknown, \textcolor{\edited}{which is} in stark contrast to the one-dimensional (scalar) case \cite{li2008range}. The following \lcnamecref{thm:sumBeamOptimality} presents a 
novel necessary and sufficient condition for general $f$ and demonstrates the nontriviality of two-dimensional extensions to \eqref{eq:linearArrayDesignRule}. 
We subsequently denote the array dimension by $n$
to unify the treatment of linear ($n=1$) and planar ($n=2$) geometries.

\begin{theorem}(Sum beam optimality)
\label{thm:sumBeamOptimality}
    Let $f(\mathbf{X}^{-1}):\mathbb{S}^n_{++}\!\rightarrow\!\mathbb{R}_+$ be a convex function in $\mathbf{X}$. Then 
    \eqref{eq:optTxWaveform_sumPlanar} is a solution to \eqref{eq:minimizeCRBwrtWaveform_plan} if \textcolor{\edited}{and only if}
    \begin{align}
        \bm{\Sigma}^{\frac{1}{2}}(\mathcal{D}_{\tx})%^{\T}
        \mathbf{G}\bm{\Sigma}^{\frac{1}{2}}(\mathcal{D}_{\tx}) \preceq \textrm{\normalfont trace}\left(\mathbf{G}\bm{\Sigma}(\mathcal{D}_{\rx})\right)\mathbf{I},
        \label{eq:condition}
    \end{align}  
    where 
    \begin{align}
        \mathbf{G} 
        & = - \left. \nabla_{\mathbf{X}} f(\mathbf{X}^{-1})\right\rvert_{\mathbf{X} = \bm{\Sigma}(\mathcal{D}_{\rx})}  \notag \\
        & = \bm{\Sigma}^{-1}(\mathcal{D}_{\rx})\left(\left. \nabla_{\mathbf{Y}}f(\mathbf{Y})\right\rvert_{\mathbf{Y}=\bm{\Sigma}^{-1}(\mathcal{D}_{\rx})}\right)\bm{\Sigma}^{-1}(\mathcal{D}_{\rx}).
    \end{align}  
\end{theorem}
\begin{proof}\let\proof\relax\let\endproof\relax
    See \cref{app:sumBeamOptimality}.
\end{proof}

\textcolor{\edited}{Note that, if $f(\mathbf{X}^{-1})$ is \textit{strictly} convex in $\mathbf{X}$, the condition in \cref{thm:sumBeamOptimality} identifies \eqref{eq:optTxWaveform_sumPlanar} as the \textit{unique} global solution to \eqref{eq:minimizeCRBwrtWaveform_plan}. This property often holds for standard sensing objectives, such as the trace and log-determinant (but not for the maximum eigenvalue) of the \ac{crbm}.}
\textcolor{\edited}{Physically, the condition \eqref{eq:condition} represents a balance between the transmit and receive spatial characteristics. The matrix $\mathbf{G}$ acts as a sensitivity metric, weighting the transmit spatial covariance $\bm{\Sigma}(\mathcal{D}_{\tx})$ by the gradient of the objective function evaluated at the receive spatial covariance $\bm{\Sigma}(\mathcal{D}_{\rx})$. Intuitively, this condition ensures that the receive array's spatial resolution is sufficiently high relative to the transmit array; if this condition is violated, a difference-beam (which creates narrower peaks) may outperform the sum-beam. 
%While the general condition in \eqref{eq:condition} is nontrivial, it simplifies significantly for standard estimation objectives, as shown in the following corollary.
The general condition in \eqref{eq:condition} simplifies significantly for standard estimation objectives, as shown in the following corollary.
}
\begin{corollary}\label{thm:sum_beam_planar_specialized}
If $f(\cdot) = \log\det(\cdot)$, then \eqref{eq:condition} reduces to
    \begin{align}
        \bm{\Sigma}(\mathcal{D}_{\rx}) \succeq \tfrac{1}{n} \bm{\Sigma}(\mathcal{D}_{\tx}).
        \label{eq:planarArrayDesignRuleLogDet}
    \end{align}
If $f(\cdot) = \mathrm{trace}(\cdot)$, then \eqref{eq:condition} reduces to 
    \begin{align}
        \bm{\Sigma}^{2}(\mathcal{D}_{\rx}) \succeq \frac{1}{\mathrm{trace}\!\left(\bm{\Sigma}^{-1}(\mathcal{D}_{\rx})\right)} \, \bm{\Sigma}(\mathcal{D}_{\tx}).
        \label{eq:planarArrayDesignRuleTrace}
    \end{align}
\end{corollary}
\begin{proof}
The proof follows directly from substituting the specific forms of $f$ into the expression for $\mathbf{G}$ in \eqref{eq:condition}.  
For $f(\mathbf{Y})=\log\det(\mathbf{Y})$, the gradient is $\nabla_{\mathbf{Y}} f(\mathbf{Y})=\mathbf{Y}^{-1}$, which gives $\mathbf{G}=\bm{\Sigma}^{-1}(\mathcal{D}_{\rx})$. Substituting this into \eqref{eq:condition} and using $\mathrm{trace}(\mathbf{I}_n)=n$ yields \eqref{eq:planarArrayDesignRuleLogDet}.
Similarly, for $f(\mathbf{Y})=\mathrm{trace}(\mathbf{Y})$, the gradient is $\nabla_{\mathbf{Y}} f(\mathbf{Y})=\mathbf{I}$, which gives $\mathbf{G}=\bm{\Sigma}^{-2}(\mathcal{D}_{\rx})$. Substituting this into \eqref{eq:condition} and rearranging yields \eqref{eq:planarArrayDesignRuleTrace}. 
\end{proof}

To illustrate \eqref{eq:planarArrayDesignRuleLogDet}, suppose that the 
%\ac{tx} and \ac{rx} array share the same geometry but differ in scale, e.g. $\mathcal{D}_{\tx} = \mathcal{D}$ and $\mathcal{D}_{\rx} = a\mathcal{D}$,
\ac{rx} array equals the \ac{tx} array dilated by a scalar factor $a>0$, i.e., $\mathcal{D}_{\rx} = a\mathcal{D}_{\tx}$, 
such that $a^2\bm{\Sigma}(\mathcal{D}_{\tx}) = \bm{\Sigma}(\mathcal{D}_{\rx})$. Condition \eqref{eq:planarArrayDesignRuleLogDet} then yields that $a \geq 1/\sqrt{n}$ guarantees sum-beam optimality. In other words, if the $\log\det$ function is used as objective, then, the \ac{rx} array may be up to a factor $\sqrt{n}$ smaller than the \ac{tx} array (in $n$ dimensions), while still preserving sum-beam optimality.

Regarding \eqref{eq:planarArrayDesignRuleTrace}, we first note that it 
%We note that \eqref{eq:planarArrayDesignRuleTrace} 
differs from the following (\textcolor{\edited}{proposed}
%alleged sum-beam 
optimality) condition in \cite{forsythe2005waveform}:%\footnote{}
\begin{align}
    \textrm{trace}\left(\bm{\Sigma}(\mathcal{D}_{\rx})\bm{\Sigma}^{-1}(\mathcal{D}_{\tx})\right) > 1.
    \label{eq:planarArrayDesignRuleTraceForsythe}
\end{align}
%
% This discrepancy arises because the derivation in \cite{forsythe2005waveform} appears to contain a mistake.\footnote{Specifically, the authors forget to include dual variables related to the semidefinite constraint in the Lagrangian.} 
% As a result, there exist array geometries that satisfy \eqref{eq:planarArrayDesignRuleTraceForsythe} but not our corrected condition \eqref{eq:planarArrayDesignRuleTrace}, for example $\mathcal{D}_{\tx} = \{(0,0),(2,0),(0,3),(2,3)\}$ and $\mathcal{D}_{\rx} = \{(0,0),(2,0),(0,2),(2,2)\}$. One can verify using numerical optimization that this array indeed does not have the sum-beam as optimal transmit waveform.

This discrepancy arises because \eqref{eq:planarArrayDesignRuleTraceForsythe} %\cite{forsythe2005waveform} 
is \textcolor{\edited}{derived from an incomplete stationarity condition. More specifically, the Lagrangian of the problem in \eqref{eq:objectiveFunctionCRBM} is \eqref{eq:trueLagr}.}
\begin{figure*}[t]
\textcolor{\edited}{
\begin{align}
    \mathcal{L} = %\frac{\sigma^2}{2|\gamma|^2 N_{\tx}N_{\rx}}
    % \Big(
    \underbrace{\text{trace}
    \left(
        \left[
            \lambda_{1} \bm{\Sigma}_{\rx}
            \!+\!
            \bm{\Sigma}^{\frac{1}{2}}_{\tx}
            \bm{\Lambda}_{2}
            \bm{\Sigma}^{\frac{1}{2}}_{\tx} 
        \right]^{-1}
        \right)
    + \mu(\lambda_{1} 
        + 
    \textrm{trace}(\bm{\Lambda}_{2}) - 1)}_{\tilde{\mathcal{L}}}
    -
    \nu \lambda_1
    -
    \text{trace}(\mathbf{Z}\bm{\Lambda}_2).
    \label{eq:trueLagr}
\end{align}
\hrule
}
\end{figure*}
\textcolor{\edited}{
However, \cite{forsythe2005waveform} uses $\tilde{\mathcal{L}}$, omitting the terms related to the inequality constraints, which leads to incomplete stationarity conditions $\nabla_{{\lambda}_1}\tilde{\mathcal{L}} = 0$ and $\nabla_{\bm{\Lambda}_2}\tilde{\mathcal{L}} = \mathbf{0}$. In contrast,} \eqref{eq:planarArrayDesignRuleTrace} is derived from the \textcolor{\edited}{complete} \ac{kkt} conditions. 
%\textcolor{red}{[Ids: What do you mean by ``approximate" and ``actual"? \cite{forsythe2005waveform} forget terms in the Lagrangian and then only look at a (hence wrong) stationarity condition.]}
Indeed, there exist array geometries for which \textcolor{\edited}{the coherent waveform} \eqref{eq:optTxWaveform_sumPlanar} is suboptimal although \eqref{eq:planarArrayDesignRuleTraceForsythe} is satisfied: %(and $f(\cdot)=\textrm{\normalfont trace}(\cdot)$): \textcolor{red}{[Ids: can the ``(...):" be removed?]}
An example is $\mathcal{D}_{\tx}\!=\!\{(0,0),(2,0),(0,5),(2,5)\}$ and $\mathcal{D}_{\rx}\!=\!\{(0,0),(2,0),(0,2),(2,2)\}$\textcolor{\edited}{;} this can be verified by numerically solving (convex) optimization problem \eqref{eq:objectiveFunctionCRBM} and comparing the objective value (\SI{5.42e-2}{}) to the trace of \eqref{eq:CRBMoptimalTxWaveFplanar} (\SI{6.25e-2}{}). 
%for which it can be verified---by numerically solving (convex) optimization problem \eqref{eq:objectiveFunctionCRBM}---that \eqref{eq:optTxWaveform_sumPlanar} is suboptimal, despite the fact that \eqref{eq:planarArrayDesignRuleTraceForsythe} is satisfied. 
In contrast, \eqref{eq:condition} is \textcolor{\edited}{a %the correct
necessary and} sufficient condition for \textcolor{\edited}{the coherent waveform} \eqref{eq:optTxWaveform_sumPlanar} to \textcolor{\edited}{be the solution to} \eqref{eq:objectiveFunctionCRBM}. \textcolor{\edited}{It can be verified that the example above does not satisfy \eqref{eq:condition}, or more specifically, \eqref{eq:planarArrayDesignRuleTrace} %, since
(\cite{forsythe2005waveform} uses $f(\cdot)=\text{trace}(\cdot)$).}

% \textcolor{red}{[Robin: The difference in the objectives is very small ($0.32\%$), so a reviewer might justifiably doubt whether the difference is just due to numerical inaccuracy. Perhaps better to omit objective function values if we don't have a more convincing example.]}

% \textcolor{red}{Ids: You can lengthen the $\mathcal{D}_{\tx}$ more (than the current $1.5$ ratio) to increase the difference. For example: $\mathcal{D}_{\tx}\!=\!\{(0,0),(2,0),(0,5),(2,5)\}$ leads to an objective of $\SI{5.42e-2}{}$, and $\mathcal{D}_{\tx}\!=\!\{(0,0),(2,0),(0,50),(2,50)\}$ leads to $\SI{3.37e-2}{}$, of course the \SI{6.25e-2}{} does not change. Back when I was constructing this example I listed this with $3$ as it is the first integer $>2$ for which this difference is visible.}

%The condition in \eqref{eq:planarArrayDesignRuleTrace} is less transparent in general. However, for \textit{symmetric} arrays, the spatial covariance matrix becomes diagonal. 
To illustrate \eqref{eq:planarArrayDesignRuleTrace}, we consider \textit{symmetrical} \ac{tx} and \ac{rx} arrays, whose spatial covariance matrices become diagonal.
We define an array $\mathcal{D}$ to be symmetric if
    %it obeys symmetry with respect to the vertical axis at $\mu^{(1)}(\mathcal{D})$ and the horizontal axis at $\mu^{(2)}(\mathcal{D})$, hence if
\begin{align}
  %  \mathbf{d} \in \mathcal{D}
   % \!\implies\! 
    % \begin{bmatrix}
    %     2\mu^{(1)} - d^{(1)} \\
    %     d^{(2)}
    % \end{bmatrix},
    % \begin{bmatrix}
    %     d^{(1)} \\
    %     2\mu^{(2)} - d^{(2)}
    % \end{bmatrix}
    % \in \mathcal{D},
\begin{bmatrix}
        2\mu^{(1)} - d^{(1)} \\
        d^{(2)}
    \end{bmatrix}
    \!\in\!\mathcal{D}, \forall \mathbf{d}\!\in\!\mathcal{D}
    \text{ or }
    \begin{bmatrix}
        d^{(1)} \\
        2\mu^{(2)} - d^{(2)}
    \end{bmatrix}
    \!\in\!\mathcal{D}, \forall \mathbf{d}\!\in\!\mathcal{D},
\label{eq:symmetricArraySufficientCondition}
\end{align}
i.e., if $\mathcal{D}$ is symmetrical with respect to 
%both 
the vertical axis $\mu^{(1)}(\mathcal{D})$ or horizontal axis $\mu^{(2)}(\mathcal{D})$ (or both). 

% \textcolor{red}{[Ids: we only need one symmetry axis for the covariance to be diagonal, but we can have two (or more).]}

% \textcolor{red}{[Robin: Check revised symmetry condition.]}

%
For a symmetric array $\mathcal{D}$, the off-diagonal terms of $\mathbf{\Sigma}(\mathcal{D})$ vanish\textcolor{\edited}{,} as verified by substituting \eqref{eq:symmetricArraySufficientCondition} into \eqref{eq:SpatialCovariance}: 
%Substituting one of the relations in \eqref{eq:symmetricArraySufficientCondition} into \eqref{eq:SpatialCovariance} shows that, for symmetric arrays, the off-diagonal terms of the spatial covariance matrix vanish:
\begin{align}
    &[\mathbf{\Sigma}(\mathcal{D})]_{1,2}
    = \frac{1}{|\mathcal{D}|}\sum_{\begin{subarray}{c}
        \mathbf{d} %= [d^{(1)}, d^{(2)}]^{\T} \\
        \in \mathcal{D}
    \end{subarray}}
    \left(d^{(1)} - \mu^{(1)}(\mathcal{D})\right)
    \left(d^{(2)} - \mu^{(2)}(\mathcal{D})\right) \notag \\
    & = \frac{1}{2|\mathcal{D}|}\sum_{\begin{subarray}{c}
        \mathbf{d} %= [d^{(1)}, d^{(2)}]^{\T} \\
        \in \mathcal{D}
    \end{subarray}}
    \left(d^{(1)} - \mu^{(1)}(\mathcal{D})\right)
    \left(d^{(2)} - \mu^{(2)}(\mathcal{D})\right) \notag \\
    & \ +\frac{1}{2|\mathcal{D}|}
    \sum_{\begin{subarray}{c}
        \mathbf{d} %= [d^{(1)}, d^{(2)}]^{\T} \\
        \in \mathcal{D}
    \end{subarray}}
    \left(2\mu^{(1)} - d^{(1)} - \mu^{(1)}(\mathcal{D})\right)
    \left(d^{(2)} - \mu^{(2)}(\mathcal{D})\right) \notag\\
    & = \frac{1}{|\mathcal{D}|}\sum_{\begin{subarray}{c}
        \mathbf{d} %= [d^{(1)}, d^{(2)}]^{\T} \\
        \in \mathcal{D}
    \end{subarray}} 
    0\left(d^{(2)} - \mu^{(2)}(\mathcal{D})\right) = 0.
    \label{eq:offdiagonalsEqualZero}
\end{align}
% \textcolor{red}{[Robin: The above calculation is not very illuminating, especially since there is undefined notation $\mathcal{D}_1$.]}
% \textcolor{red}{[Ids: I will check if I can find a reference for this property, maybe \cite{baysal2003onthegeometry}.]}
%Here $\mathcal{D}_1 = \{\mathbf{d}\in\mathcal{D} \ | \ [2\mu^{(1)} - d^{(1)},d^{(2)}]^{\T}\notin \mathcal{D}\}$.
Hence, if $\bm{\Sigma}(\mathcal{D}_{\tx})=\textrm{diag}(t_1,t_2)$ and $\bm{\Sigma}(\mathcal{D}_{\rx})=\textrm{diag}(r_1,r_2)$ denote the spatial covariance of a symmetrical Tx and Rx array, respectively, then \eqref{eq:planarArrayDesignRuleTrace} simplifies to the pair of inequalities
\begin{align}
    r_1\left(1 + \frac{r_1}{r_2}\right) \geq t_1 
    \quad \wedge \quad
    r_2\left(1+\frac{r_2}{r_1}\right) \geq t_2.
\end{align}
%This shows that the ratio $r_1/r_2$ (and its inverse) determines how much smaller the \ac{rx} array can be along each axis to preserve sum-beam optimality.
The ratio $r_1/r_2$ thus determines how much smaller the spatial variance of the \ac{rx} array can be along each axis (compared to the \ac{tx} array), such that the optimal waveform is given by \eqref{eq:optTxWaveform_sumPlanar}.
In case the spatial covariance matrices are scaled identities, i.e., $\bm{\Sigma}_{\tx} = t\mathbf{I}_n$, $\bm{\Sigma}_{\rx} = r\mathbf{I}_n$, with $r,t>0$, both conditions \eqref{eq:planarArrayDesignRuleLogDet} and \eqref{eq:planarArrayDesignRuleTrace} simplify to (for $n$-dimensional arrays)
\begin{align}
    r n \geq t.
    \label{eq:conditionIsotropic}
\end{align}
Arrays with scaled identity spatial covariance matrices are referred to as \textit{isotropic} arrays~\cite{baysal2003onthegeometry}. For such arrays, the \ac{crbm} with respect to the azimuth and elevation 
\emph{angles} (not just spatial frequency $\bm{\omega}$) is independent of the actual values of $(\theta,\phi)$ \cite{baysal2003onthegeometry}.
For example, consider isotropic \ac{tx} and \ac{rx} circular arrays with radii $\rho_{\tx}$ and $\rho_{\rx}$, respectively. Regardless of the number of sensors, their spatial covariance matrices are scaled identities with $t = \rho_{\tx}^2/2$ and $r = \rho_{\rx}^2/2$.
\footnote{Let $\mathcal{D}=\{[x,y]^{\T} \mid (x,y)=(\cos(2\pi i/N),\sin(2\pi i/N)),\ i=0,\dots,N-1\}$ denote the unit circular array, for which $\bm{\Sigma}(\mathcal{D})=\tfrac{1}{2}\mathbf{I}$. Hence, the scaled array $\mathcal{D}_{(\cdot)}=\rho_{(\cdot)}\mathcal{D}$ satisfies $\bm{\Sigma}(\mathcal{D}_{(\cdot)})=\rho_{(\cdot)}^2 /2\mathbf{I}$.}
% \footnote{Let $\mathcal{D}=\{[x,y]^T\ | \ (x,y)=(\cos(2\pi i/N),\sin(2\pi i/N)),\ i = 1,\dots,N-1\}$ denote the `unit' circular array which has $\bm{\Sigma}(\mathcal{D}) = \frac{1}{2}\mathbf{I}$. Hence $\mathcal{D}_{(.)}=\rho_{(.)}\mathcal{D}$ has $\bm{\Sigma}(\mathcal{D}_{(.)}) = \frac{\rho_{(.)}}{2}\mathbf{I}$.}
Hence,  for $n=2$, condition~\eqref{eq:conditionIsotropic} simplifies to $\rho_{\rx} \geq \rho_{\tx}/\sqrt{2}$. This is consistent with the scaling relation $1/\sqrt{n}$ mentioned earlier.
%This property has been studied in the passive sensing scenario~\cite{baysal2003onthegeometry}.
% Note that for a symmetric the off-diagonal elements of the spatial covariance matrix $\bm{\Sigma}(\mathcal{D})$ in \eqref{eq:SpatialCovariance} become zero as, 

% If the \ac{tx} and \ac{rx} array are symmetric, the condition in \eqref{eq:planarArrayDesignRuleLogDet} simplifies to two inequality constraints:

To conclude, when extending from linear to planar arrays, the condition for sum-beam optimality generally becomes less restrictive: the spatial covariance of the \ac{rx} array need not be strictly larger than that of the \ac{tx} array. 
%The general condition in 
% Condition 
% \eqref{eq:condition} provides a 
% \textcolor{\edited}{general} 
% %unifying 
% design rule 
% %that can be adapted to the number of dimensions and objective function considered.
% \textcolor{\edited}{applicable to a wide range of objective functions (e.g., trace and log-determinant) and array dimensions (linear, planar, volumetric).}
Additionally, while this section focuses on planar arrays, the extension to cubic arrays is straightforward, as the spatial covariance matrix naturally generalizes from $2\times 2$ to $3\times 3$.

\subsection{Orthogonal waveforms and optimal array geometries}\label{sec:planar_orth}
% This naturally leads to the question: which array geometries maximize spatial covariance?

For orthogonal waveforms, \labelcref{eq:CRBorth,eq:CRBlinearArrayOrth} straightforwardly extend to planar arrays (see \cref{app:spatialVarianceSumCoArray} for details):
%relies solely on the covariance of the sum co-array:
\begin{align}
    \textrm{CRBM}_{\bm{\omega}} (\Sorth)
    & = \frac{\sigma^2}{2|\gamma|^2}\frac{1}{N_{\rx}}
    \left[ 
     \bm{\Sigma}(\mathcal{D}_{\tx})
    + \bm{\Sigma}(\mathcal{D}_{\rx})
    \right]^{-1} \notag \\
    & = \frac{\sigma^2}{2|\gamma|^2}\frac{1}{N_{\rx}}
     \tilde{\bm{\Sigma}}^{-1}(\mathcal{D}_{\Sigma}), \label{eq:CRBMorthWaveformPlanar}
\end{align}
where the multiplicity-weighted spatial covariance of the sum co-array 
$\mathcal{D}_\Sigma = \mathcal{D}_{\tx}+\mathcal{D}_{\rx} = \{\mathbf{d}_{\tx}+\mathbf{d}_{\rx} \ | \ \mathbf{d}_{\tx}\in\mathcal{D}_{\tx} , \mathbf{d}_{\rx}\in\mathcal{D}_{\rx} \}$ is defined similarly to \eqref{eq:spatial_variance_sumcoarray} as
% for planar arrays is, similarly to \eqref{def:sumcoarray1D},
% %
% \begin{align*}
%     \mathcal{D}_\Sigma = \mathcal{D}_{\tx}+\mathcal{D}_{\rx} = \{\mathbf{d}_{\tx}+\mathbf{d}_{\rx} \ | \ \mathbf{d}_{\tx}\in\mathcal{D}_{\tx} , \mathbf{d}_{\rx}\in\mathcal{D}_{\rx} \}.
% \end{align*}
%Note that $\mathcal{D}_\Sigma$ can again contain the same element multiple times. Therefore, we define $\upsilon_\Sigma()$ as the multiplicity of each element.
% Therefore, we define $\mathcal{U}_\Sigma$ as the set of $N_\Sigma \leq N_{\tx}N_{\rx}$ unique elements from $\mathcal{D}_\Sigma$, and $\upsilon_\Sigma()$ as the multiplicity of each element.
%
\begin{align*}
    % ,
    %  =\frac{1}{N_{\tx}N_{\rx}}\sum_{d_{\tx}\in\mathcal{D}_{\tx}}
    % \sum_{d_{\rx}\in\mathcal{D}_{\rx}}(d_{\tx}+d_{\rx}) \notag\\
    % & 
    % = \bm{\mu}(\mathcal{D}_{\tx}) + \bm{\mu}(\mathcal{D}_{\rx}), 
    % \\
    % \tilde{\bm{\Sigma}}(\mathcal{D}_{\Sigma})
    % & = 
    % \frac{1}{N_{\tx}N_{\rx}}
    %     \sum_{\mathbf{d}_{\Sigma}\in\mathcal{D}_{\Sigma}} 
    %     (\mathbf{d}-\tilde{\bm{\mu}}(\mathcal{D}_\Sigma))(\mathbf{d}-\tilde{\bm{\mu}}(\mathcal{D}_\Sigma))^{\T}
    %     \upsilon_\Sigma(\mathbf{d}_{\Sigma})
    \tilde{\bm{\Sigma}}(\mathcal{D}_{\Sigma})\!=\! 
    \frac{1}{N_{\tx}N_{\rx}}
        \sum_{\mathbf{d}_{\Sigma}\in\mathcal{D}_{\Sigma}} 
        (\mathbf{d}_\Sigma\!-\!\tilde{\bm{\mu}}(\mathcal{D}_\Sigma))(\mathbf{d}_\Sigma\!-\!\tilde{\bm{\mu}}(\mathcal{D}_\Sigma))^{\T}
        \upsilon(\mathbf{d}_{\Sigma}),
\end{align*}
with multiplicity-weighted spatial mean $\tilde{\bm{\mu}}(\mathcal{D}_{\Sigma}) = \frac{1}{N_{\tx}N_{\rx}}\sum_{\mathbf{d}_{\Sigma}\in\mathcal{D}_{\Sigma}}\mathbf{d}_{\Sigma}\upsilon(\mathbf{d}_{\Sigma})$. 
%Here we have used that $\bm{\Sigma}(\mathcal{D}_{\tx}) + \bm{\Sigma}(\mathcal{D}_{\rx}) = \tilde{\bm{\Sigma}}(\mathcal{D}_{\Sigma})$, which is derived in detail in Appendix \ref{app:spatialVarianceSumCoArray}.

For both orthogonal and coherent waveforms, 
%identifying 
characterizing 
optimal planar arrays 
remains largely an open problem.
%The problem is more challenging than in the linear case, as the answer depends on the choice of scalar objective function $f(\txt{CRBM}_{\bm{\omega}})$. %, since the objective involves minimizing $f(\txt{CRBM}_{\bm{\omega}})$. 
However, by simply extrapolating the intuition behind \eqref{eq:clustered} to two dimensions, 
%``maximizing'' the spatial covariance matrix of the \ac{tx} /\ac{rx} arrays 
a low CRB can be achieved for many choices of objective function $f(\txt{CRBM}_{\bm{\omega}})$ by placing \ac{tx}/\ac{rx} sensors along the 
\emph{boundary} of a given area (e.g., rectangle or circle) containing them. 
%edges is beneficial. 
%, but obtaining closed-form optimal planar arrays remains an open problem.
We note that similar sensor selection problems have previously been addressed using greedy or submodular optimization methods~\cite{shamaiah2010greedy,tohidi2019sparse}. While such approaches can yield upper bounds on the worst-case suboptimality of any found array geometry, it can be shown that $f(\txt{CRBM}_{\bm{\omega}})$ considered here unfortunately lacks the required submodularity property. %for non-decreasing $f$ on the cone of positive definite matrices (such as trace and log determinant). %We hence find optimal planar array geometries by exhaustive search.
%\textcolor{red}{[Ids: strictly speaking we should talk about supermodularity, as we are minimizing f(), instead of maximizing. However, I doubt if the general reader knows this difference, and the term supermodularity might not be familiar?]}
% \textcolor{red}{Robin: Ids, please check that the preceding statement is correct.}
%for 
%the \ac{crbm}s considered 
%herein, due to lack of
%here, the 
%submodularity property does not hold, limiting the direct applicability of such approaches.

\begin{figure}[tb]
\vspace{-.6cm}
    \newcommand{\dataLoc}{Data/OptPlanarArrayCohWavN6}
    \newcommand{\xmax}{3}
    \newcommand{\xmin}{0}
    \newcommand{\ymax}{3}
    \newcommand{\ymin}{0}
    \newcommand{\xwidth}{3}
    \newcommand{\ywidth}{3}
    \newcommand{\msize}{1.5}
	\pgfplotsset{every x tick label/.append style={font=\tiny},every y tick label/.append style={font=\tiny}}
    \centering
    \subfloat[$N_{\rx}=6$.]{
    \begin{tikzpicture}
			\begin{axis}%
            [width= \xwidth cm ,height=\ywidth cm, %
            xmin=\xmin-0.2,xmax=\xmax+0.2,xtick={\xmin,\xmin+1,...,\xmax},
            ymin=\ymin-0.2,ymax=\ymax+0.2,ytick={\ymin,\ymin+1,...,\ymax},
            xticklabel shift = 0 pt,xtick pos=bottom,axis line style={draw=none},
            yticklabel shift = 0 pt,ytick pos=left,axis line style={draw=none},      
            ]
				% Dt
                \addplot[blue,only marks,mark=square*,mark size=\msize,
                y filter/.code={\pgfmathparse{\pgfmathresult}},%
                x filter/.code={\pgfmathparse{\pgfmathresult}}]
                table[x=x,y=y]{\dataLoc/Dr.dat};
			\end{axis}
		\end{tikzpicture}
        }%
    \renewcommand{\dataLoc}{Data/OptPlanarArrayCohWavN8}
    \subfloat[$N_{\rx}=8$.]{
    \begin{tikzpicture}
			\begin{axis}%
            [width= \xwidth cm ,height=\ywidth cm, %
            xmin=\xmin-0.2,xmax=\xmax+0.2,xtick={\xmin,\xmin+1,...,\xmax},
            ymin=\ymin-0.2,ymax=\ymax+0.2,ytick={\ymin,\ymin+1,...,\ymax},
            xticklabel shift = 0 pt,xtick pos=bottom,axis line style={draw=none},
            yticklabel shift = 0 pt,ytick pos=left,axis line style={draw=none},      
            ]
				% Dt
                \addplot[blue,only marks,mark=square*,mark size=\msize,
                y filter/.code={\pgfmathparse{\pgfmathresult}},%
                x filter/.code={\pgfmathparse{\pgfmathresult}}]
                table[x=x,y=y]{\dataLoc/Dr.dat};
			\end{axis}
		\end{tikzpicture}
        }%
    \renewcommand{\dataLoc}{Data/OptPlanarArrayCohWavN10}
    \subfloat[$N_{\rx}=10$.]{
    \begin{tikzpicture}
			\begin{axis}%
            [width= \xwidth cm ,height=\ywidth cm, %
            xmin=\xmin-0.2,xmax=\xmax+0.2,xtick={\xmin,\xmin+1,...,\xmax},
            ymin=\ymin-0.2,ymax=\ymax+0.2,ytick={\ymin,\ymin+1,...,\ymax},
            xticklabel shift = 0 pt,xtick pos=bottom,axis line style={draw=none},
            yticklabel shift = 0 pt,ytick pos=left,axis line style={draw=none},      
            ]
				% Dt
                \addplot[blue,only marks,mark=square*,mark size=\msize,
                y filter/.code={\pgfmathparse{\pgfmathresult}},%
                x filter/.code={\pgfmathparse{\pgfmathresult}}]
                table[x=x,y=y]{\dataLoc/Dr.dat};
			\end{axis}
		\end{tikzpicture}
        }%
    \renewcommand{\dataLoc}{Data/OptPlanarArrayCohWavN12}
    \subfloat[$N_{\rx}=12$.]{
    \begin{tikzpicture}
			\begin{axis}%
            [width= \xwidth cm ,height=\ywidth cm, %
            xmin=\xmin-0.2,xmax=\xmax+0.2,xtick={\xmin,\xmin+1,...,\xmax},
            ymin=\ymin-0.2,ymax=\ymax+0.2,ytick={\ymin,\ymin+1,...,\ymax},
            xticklabel shift = 0 pt,xtick pos=bottom,axis line style={draw=none},
            yticklabel shift = 0 pt,ytick pos=left,axis line style={draw=none},      
            ]
				% Dt
                \addplot[blue,only marks,mark=square*,mark size=\msize,
                y filter/.code={\pgfmathparse{\pgfmathresult}},%
                x filter/.code={\pgfmathparse{\pgfmathresult}}]
                table[x=x,y=y]{\dataLoc/Dr.dat};
			\end{axis}
		\end{tikzpicture}
        }%
    \caption{\ac{rx} array geometries minimizing the trace of the \ac{crbm} for coherent (optimal) waveforms \eqref{eq:CRBMoptimalTxWaveFplanar} for different number of sensors. Sensors tend to lie on the boundary of the $4\times 4$ square they are constrained to lie in.
    %The aperture is set to $L = 3$.
    %\textcolor{red}{[Take-home message here?]}
    }
    \vspace{-0.3cm}
    \label{fig:OptPlanarArrayOptWav}
\end{figure}
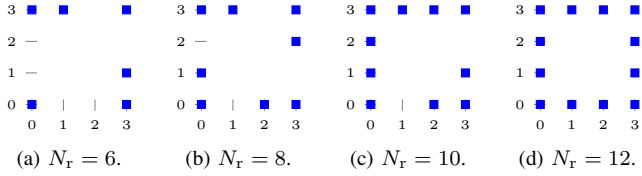
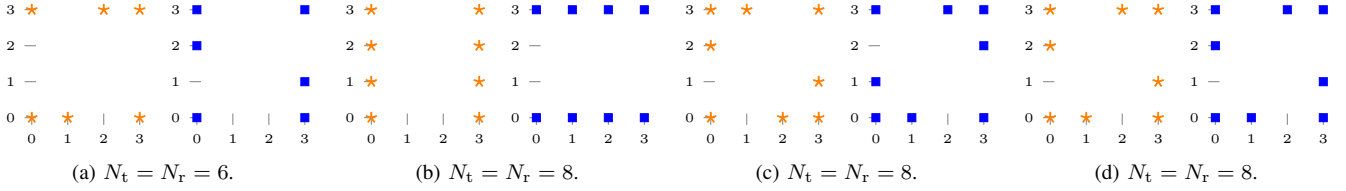
\begin{figure*}[tb]
    \vspace{-0.8cm}
    \newcommand{\dataLoc}{Data/OptPlanarArrayOrthWavN6}
    \newcommand{\xmax}{3}
    \newcommand{\xmin}{0}
    \newcommand{\ymax}{3}
    \newcommand{\ymin}{0}
    \newcommand{\xwidth}{3.2}
    \newcommand{\ywidth}{3.2}
    \newcommand{\msize}{1.5}
	\pgfplotsset{every x tick label/.append style={font=\tiny},every y tick label/.append style={font=\tiny}}
    \centering
    \subfloat[$N_{\tx}=N_{\rx}=6$.]{
    \begin{tikzpicture}
			\begin{axis}%
            [width= \xwidth cm ,height=\ywidth cm, %
            xmin=\xmin-0.2,xmax=\xmax+0.2,xtick={\xmin,\xmin+1,...,\xmax},
            ymin=\ymin-0.2,ymax=\ymax+0.2,ytick={\ymin,\ymin+1,...,\ymax},
            xticklabel shift = 0 pt,xtick pos=bottom,axis line style={draw=none},
            yticklabel shift = 0 pt,ytick pos=left,axis line style={draw=none},      
            ]
				% Dt
                \addplot[orange,only marks,line width = 0.7,mark=star,mark size=2.2,
                y filter/.code={\pgfmathparse{\pgfmathresult}},%
                x filter/.code={\pgfmathparse{\pgfmathresult}}]
                table[x=x,y=y]{\dataLoc/Dt.dat};
			\end{axis}
		\end{tikzpicture}
        % \hspace{0.1cm}
        \begin{tikzpicture}
			\begin{axis}%
            [width= \xwidth cm ,height=\ywidth cm, %
            xmin=\xmin-0.2,xmax=\xmax+0.2,xtick={\xmin,\xmin+1,...,\xmax},
            ymin=\ymin-0.2,ymax=\ymax+0.2,ytick={\ymin,\ymin+1,...,\ymax},
            xticklabel shift = 0 pt,xtick pos=bottom,axis line style={draw=none},
            yticklabel shift = 0 pt,ytick pos=left,axis line style={draw=none},      
            ]
				% Dt
                \addplot[blue,only marks,mark=square*,mark size=\msize,
                y filter/.code={\pgfmathparse{\pgfmathresult}},%
                x filter/.code={\pgfmathparse{\pgfmathresult}}]
                table[x=x,y=y]{\dataLoc/Dr.dat};
			\end{axis}
		\end{tikzpicture}%\label{subfig:Dt2D}
        }%
    \renewcommand{\dataLoc}{Data/OptPlanarArrayOrthWavN8}
    \subfloat[$N_{\tx}=N_{\rx}=8$.]{
    \begin{tikzpicture}
			\begin{axis}%
            [width= \xwidth cm ,height=\ywidth cm, %
            xmin=\xmin-0.2,xmax=\xmax+0.2,xtick={\xmin,\xmin+1,...,\xmax},
            ymin=\ymin-0.2,ymax=\ymax+0.2,ytick={\ymin,\ymin+1,...,\ymax},
            xticklabel shift = 0 pt,xtick pos=bottom,axis line style={draw=none},
            yticklabel shift = 0 pt,ytick pos=left,axis line style={draw=none},      
            ]
				% Dt
                \addplot[orange,only marks,line width = 0.7,mark=star,mark size=2.2,
                y filter/.code={\pgfmathparse{\pgfmathresult}},%
                x filter/.code={\pgfmathparse{\pgfmathresult}}]
                table[x=x,y=y]{\dataLoc/Dt1.dat};
			\end{axis}
		\end{tikzpicture}
        % \hspace{0.1cm}
        \begin{tikzpicture}
			\begin{axis}%
            [width= \xwidth cm ,height=\ywidth cm, %
            xmin=\xmin-0.2,xmax=\xmax+0.2,xtick={\xmin,\xmin+1,...,\xmax},
            ymin=\ymin-0.2,ymax=\ymax+0.2,ytick={\ymin,\ymin+1,...,\ymax},
            xticklabel shift = 0 pt,xtick pos=bottom,axis line style={draw=none},
            yticklabel shift = 0 pt,ytick pos=left,axis line style={draw=none},      
            ]
				% Dt
                \addplot[blue,only marks,mark=square*,mark size=\msize,
                y filter/.code={\pgfmathparse{\pgfmathresult}},%
                x filter/.code={\pgfmathparse{\pgfmathresult}}]
                table[x=x,y=y]{\dataLoc/Dr1.dat};
			\end{axis}
		\end{tikzpicture}%
        \label{subfig:N8option1}
        }%
    \subfloat[$N_{\tx}=N_{\rx}=8$.]{
    \begin{tikzpicture}
			\begin{axis}%
            [width= \xwidth cm ,height=\ywidth cm, %
            xmin=\xmin-0.2,xmax=\xmax+0.2,xtick={\xmin,\xmin+1,...,\xmax},
            ymin=\ymin-0.2,ymax=\ymax+0.2,ytick={\ymin,\ymin+1,...,\ymax},
            xticklabel shift = 0 pt,xtick pos=bottom,axis line style={draw=none},
            yticklabel shift = 0 pt,ytick pos=left,axis line style={draw=none},      
            ]
				% Dt
                \addplot[orange,only marks,line width = 0.7,mark=star,mark size=2.2,
                y filter/.code={\pgfmathparse{\pgfmathresult}},%
                x filter/.code={\pgfmathparse{\pgfmathresult}}]
                table[x=x,y=y]{\dataLoc/Dt2.dat};
			\end{axis}
		\end{tikzpicture}
        % \hspace{0.1cm}
        \begin{tikzpicture}
			\begin{axis}%
            [width= \xwidth cm ,height=\ywidth cm, %
            xmin=\xmin-0.2,xmax=\xmax+0.2,xtick={\xmin,\xmin+1,...,\xmax},
            ymin=\ymin-0.2,ymax=\ymax+0.2,ytick={\ymin,\ymin+1,...,\ymax},
            xticklabel shift = 0 pt,xtick pos=bottom,axis line style={draw=none},
            yticklabel shift = 0 pt,ytick pos=left,axis line style={draw=none},      
            ]
				% Dt
                \addplot[blue,only marks,mark=square*,mark size=\msize,
                y filter/.code={\pgfmathparse{\pgfmathresult}},%
                x filter/.code={\pgfmathparse{\pgfmathresult}}]
                table[x=x,y=y]{\dataLoc/Dr2.dat};
			\end{axis}
		\end{tikzpicture}%
        \label{subfig:N8option2}
        }%
    \subfloat[$N_{\tx}=N_{\rx}=8$.]{
    \begin{tikzpicture}
			\begin{axis}%
            [width= \xwidth cm ,height=\ywidth cm, %
            xmin=\xmin-0.2,xmax=\xmax+0.2,xtick={\xmin,\xmin+1,...,\xmax},
            ymin=\ymin-0.2,ymax=\ymax+0.2,ytick={\ymin,\ymin+1,...,\ymax},
            xticklabel shift = 0 pt,xtick pos=bottom,axis line style={draw=none},
            yticklabel shift = 0 pt,ytick pos=left,axis line style={draw=none},      
            ]
				% Dt
                \addplot[orange,only marks,line width = 0.7,mark=star,mark size=2.2,
                y filter/.code={\pgfmathparse{\pgfmathresult}},%
                x filter/.code={\pgfmathparse{\pgfmathresult}}]
                table[x=x,y=y]{\dataLoc/Dt3.dat};
			\end{axis}
		\end{tikzpicture}
        % \hspace{0.1cm}
        \begin{tikzpicture}
			\begin{axis}%
            [width= \xwidth cm ,height=\ywidth cm, %
            xmin=\xmin-0.2,xmax=\xmax+0.2,xtick={\xmin,\xmin+1,...,\xmax},
            ymin=\ymin-0.2,ymax=\ymax+0.2,ytick={\ymin,\ymin+1,...,\ymax},
            xticklabel shift = 0 pt,xtick pos=bottom,axis line style={draw=none},
            yticklabel shift = 0 pt,ytick pos=left,axis line style={draw=none},      
            ]
				% Dt
                \addplot[blue,only marks,mark=square*,mark size=\msize,
                y filter/.code={\pgfmathparse{\pgfmathresult}},%
                x filter/.code={\pgfmathparse{\pgfmathresult}}]
                table[x=x,y=y]{\dataLoc/Dr3.dat};
			\end{axis}
		\end{tikzpicture}%
        \label{subfig:N8option3}
        }%
    \caption{Array geometries minimizing the \ac{crbm} trace for orthogonal waveforms \eqref{eq:CRBMorthWaveformPlanar}. %for different number of sensors. The aperture is set to $L = 3$. 
    Multiple optimal configurations can exist, as \protect\subref{subfig:N8option1}, \protect\subref{subfig:N8option2} and \protect\subref{subfig:N8option3} demonstrate.
   % Multiple array geometries can lead to the same optimum, as \protect\subref{subfig:N8option1}, \protect\subref{subfig:N8option2} and \protect\subref{subfig:N8option3} demonstrate. 
    %Moreover, rotating both \ac{tx} and \ac{rx} array by $\SI{\pm 90}{\degree}$ leads to the same trace of the \ac{crbm} and are hence also optimal array geometries.
 %   Further optimal array geometries can be produced by rotating both \ac{tx} and \ac{rx} arrays by $\SI{\pm 90}{\degree}$, as this does not change the \ac{crbm} trace.% \textcolor{red}{[Take-home message here]}
    }
    \vspace{-0.5cm}
    \label{fig:OptPlanarArrayOrthWav}
\end{figure*}

Nevertheless, optimal planar array geometries of moderate size can be found via exhaustive search.
\cref{fig:OptPlanarArrayOptWav} shows optimal \ac{rx} array geometries for coherent waveforms with different numbers of sensors under a fixed aperture. These can be paired with any \ac{tx} array 
satisfying 
%, provided the 
condition 
%in 
\eqref{eq:planarArrayDesignRuleTrace}. 
%is satisfied. 
\cref{fig:OptPlanarArrayOrthWav} illustrates the corresponding optimal \ac{tx}/\ac{rx} array geometries for orthogonal waveforms. 
In both cases, multiple distinct geometries yield the same \ac{crbm} trace. Moreover, rotating the \ac{rx} array (coherent case) or both \ac{tx} and \ac{rx} arrays (orthogonal case) by $\SI{\pm 90}{\degree}$ leaves the trace unchanged. 
This observation motivates the study of array geometries with equal \ac{crbm}, which we explore in the next section.

\section{CRB-Equivalent Array Geometries% via Equal Sums of Squares
}\label{sec:arraysWithEqualCRB}
% The \ac{crb} (\ac{crbm}) for both orthogonal and optimal coherent waveforms is a function of the spatial (co)variance of the \ac{tx} and/or \ac{rx} arrays per \labelcref{eq:CRB_sb,eq:CRBorth} (\labelcref{eq:CRBMorthWaveformPlanar,eq:CRBMoptimalTxWaveFplanar}). 
% Consequently, any two arrays with equal spatial covariance yield the same \ac{crbm}, and thereby same estimation performance in the asymptotic, high \ac{snr} region. %is not compromised when replacing one array with another that has identical spatial (co)variance. 
%The \ac{crbm} is a function of the spatial covariance of the physical arrays or multiplicity weighted spatial covariance of the sum co-array per \labelcref{eq:CRBMorthWaveformPlanar,eq:CRBMoptimalTxWaveFplanar} for orthogonal and optimal coherent waveforms. 
The expressions \eqref{eq:CRBMoptimalTxWaveFplanar} and \eqref{eq:CRBMorthWaveformPlanar} suggest that physical \ac{rx} arrays (sum co-arrays) with equal (multiplicity-weighted) spatial covariance yield the same \ac{crbm}, and thereby same estimation performance in the asymptotic, high \ac{snr} region when employing coherent (orthogonal) waveforms. 
This reveals an overlooked degree of freedom in array design that we investigate next. Specifically, we establish a connection to sequences with \emph{equal \acf{sos}} \cite{barnett1942SoS,bradley1998equal}, a particular form of \emph{Diophantine equations}. \textcolor{\edited}{This connection yields a systematic method} for constructing equi-\ac{crb} arrays, \textcolor{\edited}{generalizing our earlier results in} \cite{vanderWerf2025Jointly}.%, i.e., multivariate polynomials with integer-valued solutions.

\subsection{Equal sums of squares (SoS)% and CRB-equivalent array design
}
The equal \ac{sos} problem of interest herein is that of finding two distinct integer sequences $(x_j)_{j=1}^n$ and $(\hat{x}_i)_{i=1}^n$ of size $n$ satisfying $\sum^n_{i=1} x_i^2 =\sum^n_{j=1} \hat{x}_j^2$. Solutions directly yield linear array geometries with equal spatial variance per \labelcref{eq:CRB_sb,eq:CRBorth}, as we will explain later. Examples of equal \ac{sos} sequences are $(4,7)$ and $(1,8)$ for $n =2$, and $(0,1,4)$ and $(2,2,3)$ for $n = 3$. 
Such sequences are of interest when sensors lie on a grid of integer multiples of some unit spacing (typically half a wavelength). While generating integer sequences with equal \ac{sos} %can generally be posed as finding \emph{integer-valued} solutions to certain nonlinear systems of equations \cite{bradley1998equal}, which 
is in general computationally challenging, the problem simplifies considerably for pairs of integers ($n = 2$).\footnote{In absence of integer constraints, e.g., in case of so-called rational arrays \cite{kulkarni2022noninteger}, generating (rational or real) sequences with equal SoS may also be easier.} We henceforth focus on this case to illustrate the basic idea of CRB-equivalent array design via equal SoS sequences. 

Given integers $x_1,x_2$ and coprime integers $k_1,k_2$, we can generate another pair of numbers %(not necessarily integer) 
with the same \ac{sos} %as exists a closed form solution that allows us to generate pairs of integers with equal \ac{sos}~
\cite[(4)]{bradley1998equal}:
\begin{align}
    &(\hat{x}_1,\hat{x}_2) \notag\\
    &\!=\!\left(\frac{-(k_2^2\!-\!k_1^2)x_1\!+\!2k_1k_2x_2}{k_1^2+k_2^2}\!,\!\frac{2k_1k_2x_1\!+\!(k_2^2\!-\!k_1^2)x_2}{k_1^2+k_2^2}\right).
    \label{eq:generateESOS}
\end{align}
While it is straightforward to verify by substitution that $\hat{x}_1^2+\hat{x}_2^2=x_1^2+x_2^2$, the challenge lies in choosing $k_1,k_2$, given $x_1,x_2$, such that $\hat{x}_1,\hat{x}_2$ are also integers (and distinct from $x_1,x_2$). For simplicity, in the following we consider the case 
\begin{align}
    \{x_1,x_2\}=\{\ell\delta,(\ell+1)\delta\},\label{eq:esos}
\end{align}
%$\{x_1,x_2\}=\{\ell\delta,(\ell+1)\delta\}$, 
$\ell\in\mathbb{N}$, where setting $\delta=k_1^2+k_2^2$ ensures that $\hat{x}_1,\hat{x}_2\in\mathbb{Z}$.

\subsection{Linear arrays with equal spatial variance}
%This subsection shows how distinct linear array geometries with equal spatial variance, and thus identical \ac{crb}, can be systematically constructed using equal \ac{sos} sequences, as illustrated in the following example.  
The next example illustrates how distinct linear array geometries with equal spatial variance, and thus identical \ac{crb}, can be systematically constructed using equal \ac{sos} sequences.

%\cref{eq:generateESOS,eq:esos} enable constructing distinct pairs of linear array geometries with equal \ac{crb}, as  demonstrated next. 
\begin{example} \label{ex:lin_esos}%[Linear arrays with equal spatial variance]%[Constructing the \ac{tx} arrays in \cref{subfig:arrayC,subfig:arrayD}]
Let $\ell=1$ and $(k_1,k_2)=(2,1)$, such that $(x_1,x_2) = (5,10)$ by \eqref{eq:esos}. Then, $(\hat{x}_1,\hat{x}_2) = (11,-2)$ by \eqref{eq:generateESOS}. %We can verify that indeed $x_1^2+x_2^2 = \hat{x}_1^2+\hat{x}_2^2$. 
\cref{subfig:arrayC,subfig:arrayD} show the two example \ac{tx} array geometries of equal spatial variance constructed using these sequences:% as follows
%To construct two array geometries with equal spatial variance, define
\[
\mathcal{D}_{\tx}^{(a)} = \{-x_2,-x_1,x_1,x_2\}+x_2 = \{0,5,15,20\}%\{-10,-5,5,10\}+10, 
\]
\[
\mathcal{D}_{\tx}^{(b)} = \{-\hat{x}_1,\hat{x}_2,-\hat{x}_2,\hat{x}_1\}
+\hat{x}_1= \{0,9,13,22\}.%.\{-11,-2,2,11\}+11.
\]
The \ac{rx} arrays in \cref{subfig:arrayC,subfig:arrayD} are constructed similarly using equal SoS sequences $(4,7)$ and $(1,8)$.
%Note that shifting does not affect the spatial variance.
%This is illustrated in \cref{subfig:arrayC,subfig:arrayD}, which show two array geometries with 
Since their \ac{tx} and \ac{rx} spatial variances are equal, the array geometries in \cref{subfig:arrayC,subfig:arrayD} achieve identical \ac{crb} when employing either orthogonal or optimal coherent transmit waveforms. 
\end{example}
\cref{ex:lin_esos} can be extended to any number of elements straightforwardly by generating several pairs of equal SoS using \labelcref{eq:generateESOS,eq:esos}, or by adding the same elements to both arrays, as 
$\bm{\Sigma}(\mathcal{D}^{(a)}_{\tx}\cup \mathcal{D}) = \bm{\Sigma}(\mathcal{D}^{(b)}_{\tx}\cup \mathcal{D})$
% $\sum_{i=1}^nx_i^2\!+\!c^2\!=\!\sum_{j=1}^n\hat{x}_j^2\!+\!c^2$ for any $c\!\in\!\mathbb{R}$
for any $\mathcal{D}$ such that $\mathcal{D}^{(a)}_{\tx}\cap \mathcal{D}=\mathcal{D}^{(b)}_{\tx}\cap \mathcal{D}=\emptyset$ and $\bm{\Sigma}(\mathcal{D}_{\tx}^{(a)}) = \bm{\Sigma}(\mathcal{D}_{\tx}^{(b)})$.
% This is illustrated in \cref{subfig:arrayC,subfig:arrayD}, which show two array geometries with equal \ac{tx} and \ac{rx} spatial variance, and thereby identical \ac{crb} when employing either orthogonal or optimal coherent transmit waveforms. 
Note that the \ac{rx} array in \cref{subfig:arrayD} lacks consecutive sensors (separated by half a wavelength), unlike the one in \cref{subfig:arrayC}. This may be advantageous 
%when considering mutual coupling effects.  
in presence of mutual coupling, which typically increases with decreasing inter-sensor spacing~\cite{liu2016supernested}. 
%Indeed, can be exploited to satisfy additional requirements, such as mitigating mutual coupling, where an array with more widely separated elements may be preferred. 
%We demonstrate constructing two arrays with equal spatial variance in the following example.
%

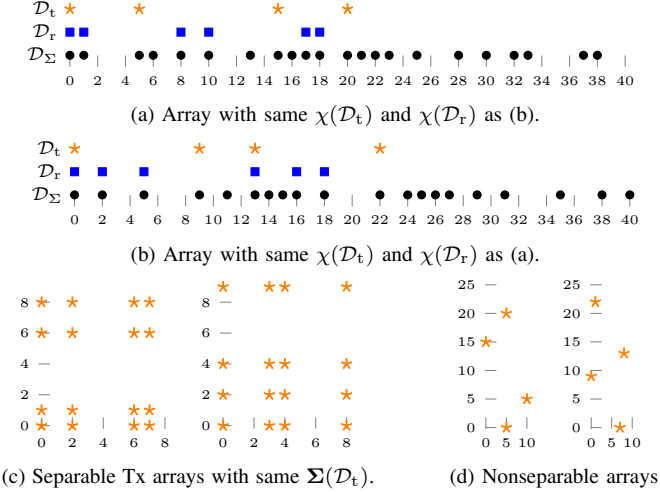
\begin{figure}
    \newcommand{\dataLoc}{Data/ArrayC}
    \newcommand{\xmax}{40}
    \newcommand{\xmin}{0}
    \newcommand{\fwidth}{9}
    \newcommand{\msize}{1.5}
	\pgfplotsset{every x tick label/.append style={font=\tiny},every y tick label/.append style={font=\scriptsize}}
	\centering
	\subfloat[Array with same $\chi(\mathcal{D}_{\tx})$ and $\chi(\mathcal{D}_{\rx})$ as 
    {\protect\subref{subfig:arrayD}}.
    ]{    
		\begin{tikzpicture}
 		     \begin{axis}[  width= \fwidth cm,
                            height=2.5 cm,
                            xmin=\xmin-0.2,
                            xmax=\xmax+0.2,
                            xtick={\xmin,\xmin+2,...,\xmax},
                            ytick={-1,0,1},
                            yticklabels={$\mathcal{D}_\Sigma$,$\mathcal{D}_{\rx}$,$\mathcal{D}_{\tx}$},
                            ytick style={draw=none},
                            ymin=-1.5,
                            ymax=1.5,
                            xticklabel shift = 0 pt,
                            xtick pos=bottom,
                            axis line style={draw=none}]
 			\addplot[orange,only marks,line width = 0.7,mark=star,mark size=2.2,y filter/.code={\pgfmathparse{\pgfmathresult*0+1}}] table[x=D,y=D]{\dataLoc/Dt.dat};
 			\addplot[blue,only marks,mark=square*,mark size=\msize,y filter/.code={\pgfmathparse{\pgfmathresult*0}}] table[x=D,y=D]{\dataLoc/Dr.dat};
			\addplot[only marks,mark=*,mark size=\msize,y filter/.code={\pgfmathparse{\pgfmathresult*0-1}},x filter/.code={\pgfmathparse{\pgfmathresult}}] table[x=D,y=D]{\dataLoc/D_Sigma.dat};
			\end{axis}
		\end{tikzpicture}
        \label{subfig:arrayC}
    % \vspace{-0.3cm}
	 }%
     \vspace{-0.3cm}
 \renewcommand{\dataLoc}{Data/ArrayD}
 \\
	\centering
	\subfloat[Array with same $\chi(\mathcal{D}_{\tx})$ and $\chi(\mathcal{D}_{\rx})$ as {\protect\subref{subfig:arrayC}}.]{
		\begin{tikzpicture}
			\begin{axis}[   width= \fwidth cm ,
                            height=2.5 cm,
                            xmin=\xmin-0.2,
                            xmax=\xmax+0.2,
                            xtick={\xmin,\xmin+2,...,\xmax},
                            ytick={-1,0,1},
                            yticklabels={$\mathcal{D}_\Sigma$,$\mathcal{D}_{\rx}$,$\mathcal{D}_{\tx}$},
                            ytick style={draw=none},
                            ymin=-1.5,
                            ymax=1.5,
                            xticklabel shift = 0 pt,
                            xtick pos=bottom,
                            axis line style={draw=none}]
				\addplot[orange,only marks,line width = 0.7,mark=star,mark size=2.2,y filter/.code={\pgfmathparse{\pgfmathresult*0+1}}] table[x=D,y=D]{\dataLoc/Dt.dat};
				\addplot[blue,only marks,mark=square*,mark size=\msize,y filter/.code={\pgfmathparse{\pgfmathresult*0}}] table[x=D,y=D]{\dataLoc/Dr.dat};
				\addplot[only marks,mark=*,mark size=\msize,y filter/.code={\pgfmathparse{\pgfmathresult*0-1}},x filter/.code={\pgfmathparse{\pgfmathresult}}] table[x=D,y=D]{\dataLoc/D_Sigma.dat};
			\end{axis}
		\end{tikzpicture}\label{subfig:arrayD}
    
	}%
    \vspace{-0.3cm}
    \\
 \renewcommand{\dataLoc}{Data/PlanarEqualTxSoS}
    \renewcommand{\xmax}{9}
    \renewcommand{\xmin}{0}
    \newcommand{\ymax}{9}
    \newcommand{\ymin}{0}
    \newcommand{\xwidth}{3.5}
    \newcommand{\xdelta}{2}
    \newcommand{\ywidth}{3.5}
    \newcommand{\ydelta}{2}
    \renewcommand{\msize}{1.25}
	\pgfplotsset{every x tick label/.append style={font=\tiny},every y tick label/.append style={font=\tiny}}
    \centering
    \subfloat[Separable \ac{tx} arrays with same $\bm{\Sigma}(\mathcal{D}_{\tx})$.% and $\bm{\Sigma}(\mathcal{D}_{\rx})$.
    ]{
    \begin{tikzpicture}
			\begin{axis}%
            [width= \xwidth cm ,height=\ywidth cm, %
            xmin=\xmin-0.2,xmax=\xmax+0.2,xtick={\xmin,\xmin+\xdelta,...,\xmax},
            ymin=\ymin-0.2,ymax=\ymax+0.2,ytick={\ymin,\ymin+\ydelta,...,\ymax},
            xticklabel shift = 0 pt,xtick pos=bottom,axis line style={draw=none},
            yticklabel shift = 0 pt,ytick pos=left,axis line style={draw=none}]
				% Dt
                \addplot[orange,only marks,line width = 0.7,mark=star,mark size=2.2,
                y filter/.code={\pgfmathparse{\pgfmathresult}},%
                x filter/.code={\pgfmathparse{\pgfmathresult}}]
                table[x=x,y=y]{\dataLoc/Dt1.dat};
			\end{axis}
		\end{tikzpicture}
        \begin{tikzpicture}
			\begin{axis}%
            [width= \xwidth cm ,height=\ywidth cm, %
            xmin=\xmin-0.2,xmax=\xmax+0.2,xtick={\xmin,\xmin+\xdelta,...,\xmax},
            ymin=\ymin-0.2,ymax=\ymax+0.2,ytick={\ymin,\ymin+\ydelta,...,\ymax},
            xticklabel shift = 0 pt,xtick pos=bottom,axis line style={draw=none},
            yticklabel shift = 0 pt,ytick pos=left,axis line style={draw=none}]
				% Dt
                \addplot[orange,only marks,line width = 0.7,mark=star,%pentagon,
                mark size=2.2,
                y filter/.code={\pgfmathparse{\pgfmathresult}},%
                x filter/.code={\pgfmathparse{\pgfmathresult}}]
                table[x=x,y=y]{\dataLoc/Dt2.dat};
			\end{axis}
		\end{tikzpicture}\label{subfig:planarESoS1}
        }
        \renewcommand{\dataLoc}{Data/PlanarEqualSoS_TxRx}
        \renewcommand{\xwidth}{2.2}
    \renewcommand{\xmax}{11}
    \renewcommand{\ymax}{25}
    \renewcommand{\xdelta}{5}
    \renewcommand{\ydelta}{5}
	\pgfplotsset{every x tick label/.append style={font=\tiny},every y tick label/.append style={font=\tiny}}
    \centering
    \subfloat[Nonseparable arrays% \ac{tx} arrays with same $\bm{\Sigma}(\mathcal{D}_{\tx})$. %and $\bm{\Sigma}(\mathcal{D}_{\rx})$.
    ]{
    \begin{tikzpicture}
			\begin{axis}%
            [width= \xwidth cm ,height=\ywidth cm, %
            xmin=\xmin-0.2,xmax=\xmax+0.2,xtick={\xmin,\xmin+\xdelta,...,\xmax},
            ymin=\ymin-0.2,ymax=\ymax+0.2,ytick={\ymin,\ymin+\ydelta,...,\ymax},
            xticklabel shift = 0 pt,xtick pos=bottom,axis line style={draw=none},
            yticklabel shift = 0 pt,ytick pos=left,axis line style={draw=none}]
				% Dt
                \addplot[orange,only marks,line width = 0.7,mark=star,mark size=2.2,
                y filter/.code={\pgfmathparse{\pgfmathresult}},%
                x filter/.code={\pgfmathparse{\pgfmathresult}}]
                table[x=x,y=y]{\dataLoc/Dt1.dat};
                
                % \addplot[orange,only marks,line width = 0.7,mark=star,mark size=2.2,
                % y filter/.code={\pgfmathparse{\pgfmathresult}},%
                % x filter/.code={\pgfmathparse{\pgfmathresult}}]
                % table[x=x,y=y]{\dataLoc/Dt1.dat};
                % \addplot[orange,only marks,blue,only marks,mark=square*,mark size=\msize,
                % y filter/.code={\pgfmathparse{\pgfmathresult}},%
                % x filter/.code={\pgfmathparse{\pgfmathresult}}]
                % table[x=x,y=y]{\dataLoc/Dr21.dat};
			\end{axis}
		\end{tikzpicture}
            \begin{tikzpicture}
			\begin{axis}%
            [width= \xwidth cm ,height=\ywidth cm, %
            xmin=\xmin-0.2,xmax=\xmax+0.2,xtick={\xmin,\xmin+\xdelta,...,\xmax},
            ymin=\ymin-0.2,ymax=\ymax+0.2,ytick={\ymin,\ymin+\ydelta,...,\ymax},
            xticklabel shift = 0 pt,xtick pos=bottom,axis line style={draw=none},
            yticklabel shift = 0 pt,ytick pos=left,axis line style={draw=none}]
                \addplot[orange,only marks,line width = 0.7,mark=star, %pentagon, 
                mark size=2.2,
                y filter/.code={\pgfmathparse{\pgfmathresult}},%
                x filter/.code={\pgfmathparse{\pgfmathresult}}]
                table[x=x,y=y]{\dataLoc/Dt2.dat};
                
                % \addplot[orange,only marks,line width = 0.7,mark=star,mark size=2.2,
                % y filter/.code={\pgfmathparse{\pgfmathresult}},%
                % x filter/.code={\pgfmathparse{\pgfmathresult}}]
                % table[x=x,y=y]{\dataLoc/Dt1.dat};
                % \addplot[orange,only marks,blue,only marks,mark=square*,mark size=\msize,
                % y filter/.code={\pgfmathparse{\pgfmathresult}},%
                % x filter/.code={\pgfmathparse{\pgfmathresult}}]
                % table[x=x,y=y]{\dataLoc/Dr21.dat};
			\end{axis}
		\end{tikzpicture}   
        \label{subfig:Planar1}
        }
 \caption{Array geometries with equal spatial (co)variances.
 The array geometries in \protect\subref{subfig:arrayC} and \protect\subref{subfig:arrayD} have identical CRB values for both optimal and orthogonal transmit waveforms. 
Planar arrays with equal spatial covariance, as in \protect\subref{subfig:planarESoS1} and \protect\subref{subfig:Planar1}, can be constructed similarly to the linear case using equal \ac{sos} sequences.
 }
 \label{fig:ArraysWithEqualTxRxVariance}
\vspace{-.5cm}
\end{figure}

\subsection{Planar arrays with equal spatial covariance}
Planar arrays %is more involved, as it 
require identifying pairs of two-dimensional \emph{vector} sequences $(\mathbf{p}_i=[x_i, y_i]^{\T})_{i=1}^n$ and $(\hat{\mathbf{p}}_j=[\hat{x}_i, \hat{y}_i]^{\T})_{j=1}^n$ such that 
%whose covariance matrices are identical. Formally, this amounts to finding sequences satisfying
$\sum_i^n \mathbf{p}_i \mathbf{p}_i^{\T} = \sum_j^n \hat{\mathbf{p}}_j \hat{\mathbf{p}}_j^{\T}$. %, with $\mathbf{p}_i = [x_i, y_i]^{\T}$. 
A straightforward construction %consists of ``outer products" of two linear arrays, as 
is illustrated in \cref{subfig:planarESoS1}, 
where the planar array is the Cartesian product of two linear arrays, i.e.,  
$\mathcal{D}^{(a)} = \{[x,y]^{\T} \ | \ (x,y)\in\{\pm x_1,\pm x_2\}\times\{\pm y_1,\pm y_2\}\}$, and $\mathcal{D}^{(b)} = \{[x,y]^{\T} \ | \ (x,y)\in\{\pm \hat{x}_1,\pm \hat{x}_2\}\times\{\pm \hat{y}_1,\pm \hat{y}_2\}\}$.
In such ``separable'' cases, finding two pairs of linear arrays with equal spatial variance suffices to construct two planar arrays with equal spatial covariance.
Another construction is given by the following \lcnamecref{thm:equalCovariance}.
%For constructing more general pairs of planar arrays with equal spatial covariance, we provide the following proposition.

\begin{proposition}\label{thm:equalCovariance}
Let $(\hat{x}_1, \hat{x}_2)$ and $(\hat{y}_1, \hat{y}_2)$ be two distinct sequences generated using \eqref{eq:generateESOS} and the same coprime pair $k_1, k_2$ applied to $(x_1, x_2)$ and $(y_1, y_2)$, respectively. 
%Suppose that distinct sets $\{x_1, x_2\} \neq \{y_1, y_2\}$ produce to distinct sets $\{\hat{x}_1, \hat{x}_2\} \neq \{\hat{y}_1, \hat{y}_2\}$ through \eqref{eq:generateESOS}. If the same coprime pair $\{k_1, k_2\}$ is used, then, the following property is satisfied:
Then
    \begin{align}
        \sum_{i=1,2}
        \begin{bmatrix}
            x_i \\
            y_i
        \end{bmatrix}
        \begin{bmatrix}
            x_i & y_i
        \end{bmatrix}
        =\sum_{i=1,2}
        \begin{bmatrix}
            \hat{x}_i \\
            \hat{y}_i
        \end{bmatrix}
        \begin{bmatrix}
            \hat{x}_i & \hat{y}_i
        \end{bmatrix}.
        \label{eq:equalCovariance}
    \end{align}
\end{proposition}
\begin{proof}
The diagonal elements of the left and right hand side of \eqref{eq:equalCovariance} are equal, i.e., $x_1^2+x^2_2=\hat{x}^2_1+\hat{x}^2_2$ and $y_1^2+y^2_2=\hat{y}^2_1+\hat{y}^2_2$, since \eqref{eq:generateESOS} generates a pair of integers with equal \ac{sos} by definition. For the off-diagonal elements, it can be verified that substitution of \eqref{eq:generateESOS} into $\hat{x}_1\hat{y}_1 + \hat{x}_2\hat{y}_2 $ yields
%we can write \eqref{eq:Sumi_xiyi}, which concludes the proof.
\begin{align*}
    \hat{x}_1\hat{y}_1 + \hat{x}_2\hat{y}_2
    & = \frac{((k_2^2-k_1^2)^2-4k_1^2k_2^2)(x_1 y_1 + x_2 y_2)}{(k_1^2+k_2^2)^2}\\
    & = \frac{(k_2^2+k_1^2)^2(x_1 y_1 + x_2 y_2)}{(k_1^2+k_2^2)^2}
    % \\
    % &
    = x_1 y_1 + x_2 y_2,
\end{align*}
which completes the proof.
% \begin{figure*}[b]
% \hrule
% \begin{align}
%     %
%     % \frac{-(k_2^2-k_1^2)x_1+2k_1k_2x_2}{k_1^2+k_2^2},
%     % \frac{2k_1k_2x_1+(k_2^2-k_1^2)x_2}{k_1^2+k_2^2}
%     %
%     & \hat{x}_1\hat{y}_1 + \hat{x}_2\hat{y}_2 
%     \notag \\
%     & 
%     = \frac{1}{(k_1^2+k_2^2)^{\textcolor{red}{2}}}
%     \Big(\left[-(k_2^2-k_1^2)x_1 + 2 k_1 k_2 x_2\right]
%     \left[-(k_2^2-k_1^2)y_1 + 2 k_1 k_2 y_2\right] \textcolor{red}{+}
%     % \notag\\
%     % & \hspace{1cm}+
%     \left[2 k_1 k_2 x_1 + (k_2^2-k_1^2)x_2\right]
%     \left[2 k_1 k_2 y_1 + (k_2^2-k_1^2)y_2\right]
%     \Big)\notag\\
%     & = \frac{1}{(k_1^2+k_2^2)^{\textcolor{red}{2}}}
%     \left(
%     (k_2^2-k_1^2)^2 x_1 y_1 + 4 \textcolor{red}{k}^2_1 \textcolor{red}{k}^2_2 x_2 y_2 + 4 \textcolor{red}{k}^2_1 \textcolor{red}{k}^2_2 x_1 y_{\textcolor{red}{1}} + (k_2^2-k_1^2)^2 x_2 y_2
%     \right)\notag\\
%     & = \frac{1}{(k_1^2+k_2^2)^{\textcolor{red}{2}}}
%     (k_2^2+k_1^2)^2(x_1 y_1 + x_2 y_2) = x_1 y_1 + x_2 y_2.
%     \label{eq:Sumi_xiyi}
% \end{align}
% \end{figure*}
\end{proof}
\cref{thm:equalCovariance} 
provides a means to construct pairs of planar arrays (that are not necessarily separable into linear arrays) with equal spatial covariances and thereby \ac{crb}s, i.e., $\mathcal{D}^{(a)} = \{[x,y]^{\T}\ | \ (x,y)\in\{\pm(x_1,y_1),\pm(x_2,y_2)\}\}$ and $\mathcal{D}^{(b)} = \{[x,y]^{\T}\ | \ (x,y)\in\{\pm(\hat{x}_1,\hat{y}_1),\pm(\hat{x}_2,\hat{y}_2)\}\}$.
\cref{subfig:Planar1} shows examples of such arrays generated using \labelcref{eq:esos,eq:generateESOS}, coprime pair $(k_1,k_2)=(2,1)$, and parameter $\ell= -1,1$.

\section{Numerical examples}\label{sec:simulations}
This section %illustrates the validity of the optimal sensor allocations for finite $N$ and 
%compares the \ac{mle} performance and \ac{crb} of different arrays. %against the \ac{crb}. 
investigates the empirical performance of the \ac{mle} using the array geometries examined in \cref{sec:orthogonalWaveforms,sec:optimalAllocation,sec:arraysWithEqualCRB} to complement the preceding \ac{crb}-based analysis. 
%\subsection{MLE performance}
\textcolor{\icassp}{We consider the \ac{mle} of $\omega$ for a \emph{fixed} waveform matrix $\mathbf{S}$} of length $T = N_{\tx}$ samples. \textcolor{\icassp}{Given} the received signal \eqref{eq:signal_model}, \textcolor{\icassp}{the \ac{mle} can be shown to reduce to the following joint \ac{tx}-\ac{rx} beamformer \cite{evans1982application}:} 
\begin{align}
\hat{\omega}
=\argmax_{\bar{\omega}\in[-\pi,\pi)}
\frac
{
|\mathbf{y}^{\HT}(\mathbf{S}\mathbf{a}_{\tx}(\bar{\omega})\otimes \mathbf{a}_{\rx}(\bar{\omega}))|^2}
{\|\mathbf{S}\mathbf{a}_{\tx}(\bar{\omega})\otimes \mathbf{a}_{\rx}(\bar{\omega})\|_2^2}.
\end{align}
\textcolor{\icassp}{The ground truth target angle and reflectivity are set to $\omega=0$ and $\gamma=1$, respectively, } such that the \ac{snr} is $\frac{|\gamma|^2}{\sigma^2} = \frac{1}{\sigma^2}$. \textcolor{\icassp}{The squared error of the \ac{mle} is averaged over $10^4$ Monte Carlo trials} with independent noise realizations.
\textcolor{\edited}{
%In the presentation of the \ac{mle} performance, i.e., \cref{fig:clusteredvsnested,fig:ArrayCvsArrayD,fig:MLEsensorAllocationClustered}, we have used the colors to differentiate between arrays: black denotes the clustered array, red the canonical MIMO array, and blue and green the arrays in \cref{subfig:arrayC,subfig:arrayD}.
In the \ac{mle} performance plots (\cref{fig:clusteredvsnested,fig:ArrayCvsArrayD,fig:MLEsensorAllocationClustered}), colors distinguish arrays as follows: black--clustered, red--canonical MIMO, and blue/green--the arrays in \cref{subfig:arrayC,subfig:arrayD}.
}
%whereas entries of noise vector $\mathbf{n}$ are drawn from an i.i.d. circularly symmetric normal distribution with variance $\sigma^2$.

%\subsection{Clustered array vs. canonical \ac{mimo} array}
\cref{fig:clusteredvsnested} shows the estimation performance of the clustered array (black lines) with the canonical \ac{mimo} array (red lines) for optimal \textcolor{\edited}{(coherent)} and orthogonal transmit waveforms. Here the $N = 20$ sensors are optimally allocated to the transmitter and receiver following \cref{tab:sensor_division}. 
%Note that since the optimal waveform in \eqref{eq:optTxWaveform_sum} \textcolor{\icassp}{is a function of the true angle $\omega$, an initial estimate of $\omega$ would be needed in practice; see \cite{li2008range} for a discussion and examples.}
The clustered array has a better asymptotical performance for both waveforms, as expected. 
The poor performance of the canonical \ac{mimo} array in \cref{subfig:MLEoptimalWav} is caused by ambiguities introduced by the dilated \ac{ula} at the receiver when transmitting optimal (coherent) waveforms. In contrast, the clustered array avoids such ambiguities due to its \ac{rx} array containing \ac{ula} segments.%since it contains consecutive sensors.

\textcolor{\edited}{As seen in \cref{subfig:MLEorthogonalWav} (and \cref{subfig:ArrayCvsArrayDoptWav}) the \ac{mle} performance curves may cross in the non-asymptotic (low \ac{snr}) regime. This behavior is due to the distinct sidelobe levels of each array geometry, which lead to different error characteristics before the performance converges to the asymptotic \ac{crb} limit.}
\begin{figure*}[t]
\vspace{-.8cm}
\centering
    \subfloat[$\mathbf{S} = \frac{\mathbf{u}\mathbf{a}_{\tx}^{\HT}(\omega)}{\sqrt{N_{\tx}}}$]{
    \centering
    \includegraphics[width=0.45\linewidth]{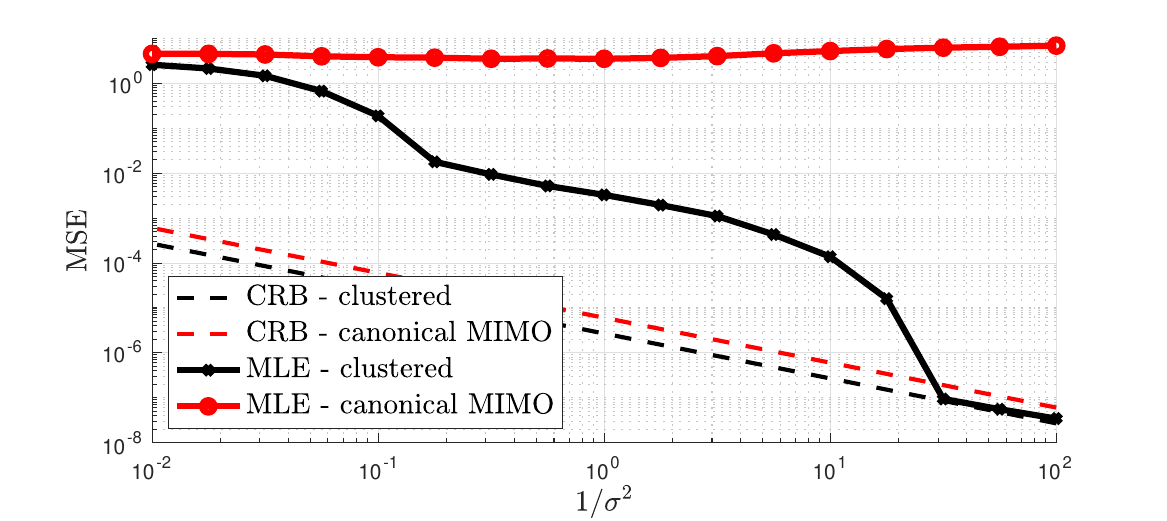}
    \label{subfig:MLEoptimalWav}
    }%
    % \vspace{-0.3cm}
    \subfloat[$\mathbf{S} = \frac{1}{\sqrt{N_{\tx}}}\mathbf{I}$.]{
    \centering
    \includegraphics[width=0.45\linewidth]{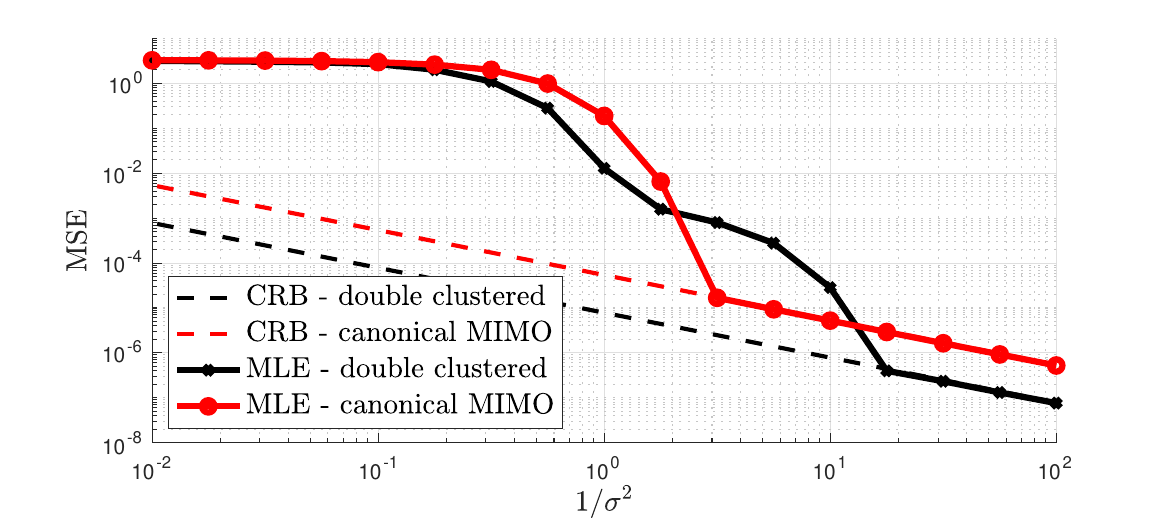}
    \label{subfig:MLEorthogonalWav}
    }%
    \caption{MLE performance %for single-target angle estimation for 
    of 
    the clustered (optimal) and canonical MIMO array, using coherent (optimal) \protect\subref{subfig:MLEoptimalWav} and orthogonal waveforms \protect\subref{subfig:MLEorthogonalWav}. The number of total sensors is $N = 20$, optimally allocated following \cref{tab:sensor_division}. 
    % \textcolor{red}{[Robin: Take-home message here]}
    The clustered array has superior performance for optimal waveforms (cf. \cref{subfig:MLEoptimalWav}), due to its \ac{rx} array geometry, and achieves a lower \ac{crb} in the asymptotic \ac{snr} regime for both optimal and orthogonal waveforms.}
    \vspace{-0.5cm}
    \label{fig:clusteredvsnested}
\end{figure*}

\subsubsection{Impact of optimal sensor allocation% on MLE performance
}
\cref{fig:MLEsensorAllocationClustered} compares the \ac{mle} performance of clustered \ac{tx} and \ac{rx} 
arrays with equal and optimal sensor allocations. The results show that proper allocation yields a substantial performance gain. 
%Interestingly, the relative improvement observed in the \ac{mle} accuracy when moving from equal to optimal allocation is larger than what is suggested by the corresponding \ac{crb}, highlighting the critical role of sensor allocation in practice.
Indeed, the \ac{mle} improves drastically even in the non-asymptotic \ac{snr} region when moving from equal to optimal allocation, highlighting the critical role of sensor allocation in practice.

% \textcolor{red}{[Robin: The precending sentence is unclear. It sounds like the MSE is lower than the CRB.]}
% \begin{figure}
% \vspace{-0.1cm}
%     \centering
%     \includegraphics[width=0.9\linewidth]{Figures/CRBresultsOrthWavClusteredSensorDistr.eps}
%     \vspace{-0.2cm}
%     \caption{MLE performance %for single-target angle estimation for a 
%     of 
%     double clustered %\ac{tx} and \ac{rx} 
%     array using orthogonal waveforms. 
%     % \textcolor{red}{[Robin: Take-home message here]}
%     Optimal sensor allocation can significantly improve \ac{mle} accuracy.} 
%     \vspace{-0.5cm}
%     \label{fig:MLEsensorAllocationClustered}
% \end{figure}

\subsubsection{Arrays with identical CRB}
\cref{fig:ArrayCvsArrayD} shows the estimation performance of the array geometries %with equal spatial variance 
in \cref{subfig:arrayC,subfig:arrayD}. Since the arrays have equal spatial variance, 
the corresponding \ac{mle}s converge to the same \ac{crb}. Although the \ac{crb}s are identical, the \ac{mle} performances deviate in the threshold region. 
While this can largely be attributed to differences in the beampatterns of the two arrays, %A precise characterization of 
precisely characterizing how the array geometry influences the \textcolor{\edited}{(non-asymptotic)} \ac{mle} performance, particularly for arrays with equal spatial variance, remains an open question and 
%important 
\textcolor{\edited}{interesting} 
direction for future work.
\begin{figure*}[t]
\vspace{-.3cm}
\centering
    \subfloat[$\mathbf{S} = \frac{\mathbf{u}\mathbf{a}_{\tx}^{\HT}(\omega)}{\sqrt{N_{\tx}}}$]{
    \centering
    \includegraphics[width=0.45\linewidth]{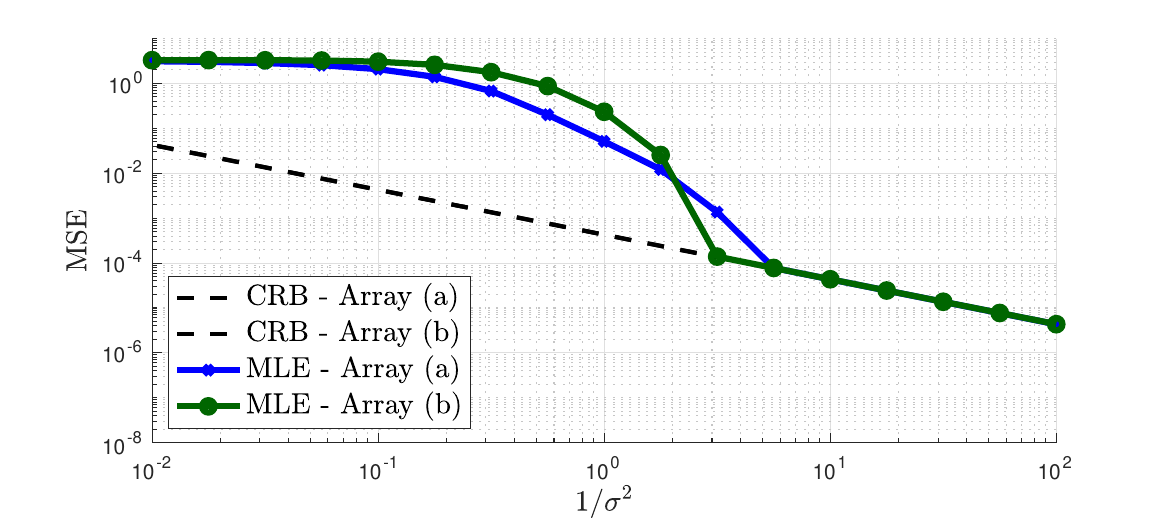}
    \label{subfig:ArrayCvsArrayDoptWav}
    }%
    % \vspace{-0.3cm}
    \subfloat[$\mathbf{S} = \frac{1}{\sqrt{N_{\tx}}}\mathbf{I}$.]{
    \centering
    \includegraphics[width=0.45\linewidth]{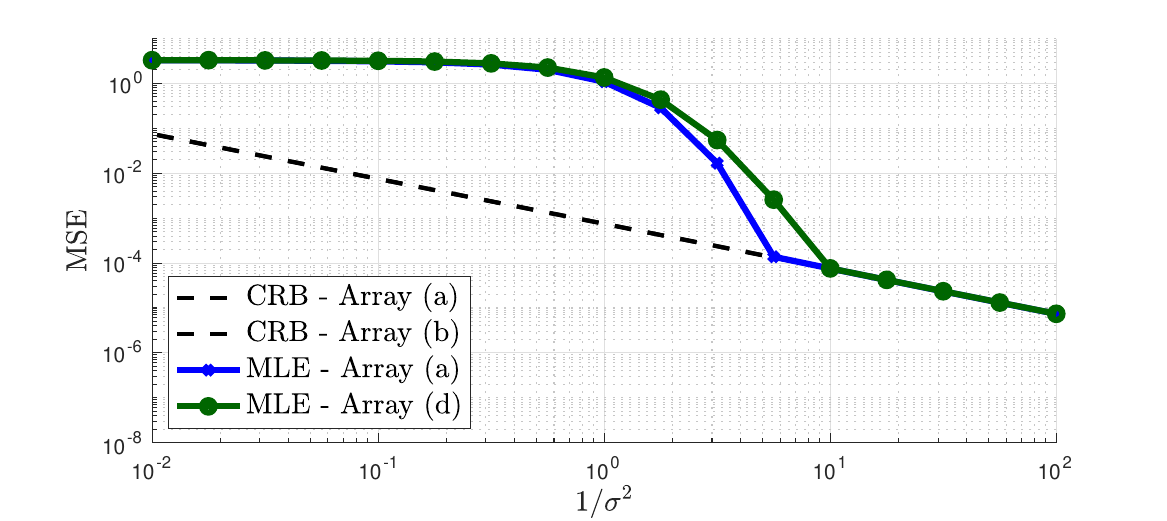}
    \label{subfig:ArrayCvsArrayDorthWav}
    }%
    \caption{MLE performance %for single target angle estimation for 
    of 
    the arrays in \cref{subfig:arrayC,subfig:arrayD} using an optimal waveform \protect\subref{subfig:ArrayCvsArrayDoptWav} and orthogonal waveforms \protect\subref{subfig:ArrayCvsArrayDorthWav}. 
    % \textcolor{red}{[Robin: Take-home message here]}
    Despite identical \ac{crb}s, different array geometries can yield distinct \ac{mle} performance in the threshold region.} 
    \vspace{-0.5cm}
    \label{fig:ArrayCvsArrayD}
\end{figure*}
% \begin{figure}[htb]
%     \centering
%     \includegraphics[width=1\linewidth]{Data/CRBresultsArrayC_vs_ArrayD.eps}
%     \caption{Single target angle estimation for the arrays in \cref{subfig:arrayC,subfig:arrayD}.}
%     \label{fig:ArrayCvsArrayD}
% \end{figure}

% \subsection{Impact of sensor allocation}
% \textcolor{red}{[Possible subsection based on your suggestion in email. However, $10\%$ improvement for the canonical MIMO is not really visible on a log scale, see \cref{subfig:MLEsensorAllocationNested}. Interestingly, it seems that by reallocating the sensors, one can get an improvement/deterioration over the whole SNR range, i.e. apparently the beampattern does not change significantly.]}
% %
% \begin{figure*}[t]
% \centering
%         \subfloat[Canonical MIMO array.]{
%     \centering
%     \includegraphics[width=0.45\linewidth]{Figures/CRBresultsOrthWaMIMOsensorDistr.eps}
%     \label{subfig:MLEsensorAllocationNested}
%     }%
%     % \vspace{-0.3cm}
%     \subfloat[Clustered \ac{tx} and \ac{rx} array.]{
%     \centering
%     \includegraphics[width=0.45\linewidth]{Figures/CRBresultsOrthWavClusteredSensorDistr.eps}
%     \label{subfig:MLEsensorAllocationClustered}
%     }%
%     \caption{MLE performance for single-target angle estimation using an orthogonal waveform \protect\subref{subfig:MLEorthogonalWav} for different sensor allocations.}
%     \label{fig:MLEsensorAllocation}
% \end{figure*}

\section{Conclusion}\label{sec:conclusions}
% This paper established fundamental limits and optimality conditions for the joint design of array geometries and transmit waveforms in active sensing. 
% We showed that, with orthogonal waveforms, the single-target \ac{crb} depends on the spatial variance of both transmit and receive arrays, while for beamforming waveforms it depends solely on the receive array. These results enabled the derivation of asymptotically optimal sensor allocations, which were validated in the finite-$N$ regime through numerical examples. The simulations demonstrated that clustered arrays achieve significant gains over canonical \ac{mimo} arrays for a single target. 
% Extending these results to planar arrays, we showed that spatial variance is replaced by spatial covariance (or 2D moment of inertia), and we established a generalized condition under which beamforming is optimal for a broad class of objective functions, including the trace, log-determinant, and maximum eigenvalue.
% Finally, we introduced a novel connection between sparse array design and sequences with equal sums of squares, illustrating how different linear and planar array geometries can be constructed to achieve identical \ac{crb}s.

This work characterized fundamental performance limits of optimal array designs employing orthogonal and coherent (\ac{crb}-optimal) waveforms in active sensing. 
We demonstrated that for orthogonal waveforms, the single-target \ac{crb} is governed by the sum of the spatial variances (or moments of inertia) of the transmit and receive arrays, or alternatively by the (multiplicity-)weighted spatial variance of the sum co-array. \textcolor{\edited}{This reveals} that CRB-optimal geometries are redundant\textcolor{\edited}{, thereby showing a fundamental trade-off between estimation accuracy and target identifiability.} % thereby exposing a fundamental trade-off between estimation accuracy and target identifiability.
We also derived the optimal allocation of a fixed number of sensors between the Tx and Rx arrays, showing that unequal allocation \textcolor{\edited}{(}favoring more Rx than Tx sensors\textcolor{\edited}{)} minimizes the \ac{crb} in case of orthogonal waveforms. These findings 
%challenge 
\textcolor{\edited}{question} 
conventional array designs with equal sensor allocations and nonredundant sum co-arrays, providing benchmarks for the \ac{mse} improvement achievable by array design alone. 
For planar arrays, we derived a novel sufficient condition for guaranteeing that \ac{crb}-optimal waveforms transmit nonzero power in the target direction. Lastly, we established a connection between sparse array design and Diophantine equations (equal sums of squares), along with a simple method for generating distinct array geometries with identical \ac{crb}s.
%These findings offer new theoretical benchmarks and practical guidelines for optimal array and waveform design. 
Directions for future work include gaining analytical insight into the multi-target \ac{crb} and exploring 
%noise-robust sparse array geometries and waveforms in the low-\ac{snr} (threshold) regime.
performance limits of array design in the low-\ac{snr} (threshold) regime for various waveforms. 

\begin{figure}
\vspace{-0.1cm}
    \centering
    \includegraphics[width=0.9\linewidth]{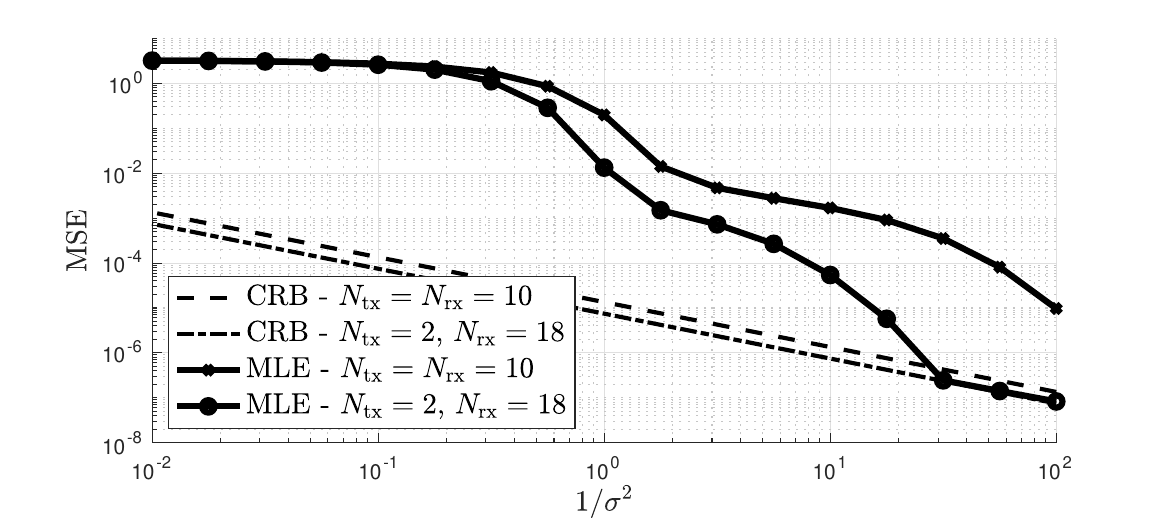}
    \vspace{-0.2cm}
    \caption{MLE performance %for single-target angle estimation for a 
    of 
    double clustered %\ac{tx} and \ac{rx} 
    array using orthogonal waveforms. 
    % \textcolor{red}{[Robin: Take-home message here]}
    Optimal sensor allocation can significantly improve \ac{mle} accuracy.} 
    \vspace{-0.5cm}
    \label{fig:MLEsensorAllocationClustered}
\end{figure}

%\appendices
\appendices\crefalias{section}{appendix}
\renewcommand{\thesectiondis}[2]{\Alph{section}:}

\section{
%Multiplicity-weighted spatial covariance of sum co-array
Proof of Proposition~\ref{thm:CRBorthWavSumcoArraydependency}}
\label{app:spatialVarianceSumCoArray}
We prove the more general two-dimensional case \eqref{eq:CRBMorthWaveformPlanar} (see \cref{sec:planarArrays} for details) of which \eqref{eq:CRBlinearArrayOrth} is a special case.
For simplicity, let $\bm{\mu}_{(.)}$ denote $\bm{\mu}(\mathcal{D}_{(.)})$. Then
\begin{align}
    & \bm{\Sigma}(\mathcal{D}_{\tx})
    + \bm{\Sigma}(\mathcal{D}_{\rx}) 
    % \notag \\
    % &
    = \frac{1}{N_{\tx}}\sum_{\mathbf{d}_{\tx}\in\mathcal{D}_{\tx}}(\mathbf{d}_{\tx}-\bm{\mu}_{\tx})(\mathbf{d}_{\tx}-\bm{\mu}_{\tx})^{\T}
    \notag\\
    & \quad 
    + \frac{1}{N_{\rx}}\sum_{\mathbf{d}_{\rx}\in\mathcal{D}_{\rx}}(\mathbf{d}_{\rx}-\bm{\mu}_{\rx})(\mathbf{d}_{\rx}-\bm{\mu}_{\rx})^{\T}\notag\\
    & = \frac{1}{N_{\tx}N_{\rx}}\sum_{\mathbf{d}_{\tx}\in\mathcal{D}_{\tx}}\sum_{\mathbf{d}_{\rx}\in\mathcal{D}_{\rx}}
    \Big(
        \mathbf{d}_{\tx}\mathbf{d}_{\tx}^{\T}
        + \mathbf{d}_{\rx}\mathbf{d}_{\rx}^{\T} \notag\\
    & \quad 
        - \bm{\mu}_{\tx}\mathbf{d}_{\tx}^{\T}
        - \mathbf{d}_{\tx}\bm{\mu}_{\tx}^{\T}
        - \bm{\mu}_{\rx}\mathbf{d}_{\rx}^{\T}
        - \mathbf{d}_{\rx}\bm{\mu}_{\rx}^{\T}
        + \bm{\mu}_{\tx}\bm{\mu}_{\tx}^{\T}
        + \bm{\mu}_{\rx}\bm{\mu}_{\rx}^{\T}
    \Big)
    \notag\\
    % & = \frac{1}{N_{\tx}N_{\rx}}\sum_{\mathcal{D}_{\tx}}\sum_{\mathcal{D}_{\rx}}
    % \Big(
    %     (\mathbf{d}_{\tx} + \mathbf{d}_{\rx})(\mathbf{d}_{\tx} + \mathbf{d}_{\rx})^{\T} \notag\\
    % & \quad 
    %     - 2\mathbf{d}_{\tx}\mathbf{d}_{\rx}^{\T}
    %     - 2\bm{\mu}_{\tx}\mathbf{d}_{\tx}^{\T}
    %     - 2\bm{\mu}_{\rx}\mathbf{d}_{\rx}^{\T}
    %     + \bm{\mu}_{\tx}\bm{\mu}_{\tx}^{\T}
    %     + \bm{\mu}_{\rx}\bm{\mu}_{\rx}^{\T}
    % \Big)
    % \notag\\
    & = \frac{1}{N_{\tx}N_{\rx}}\Big(\sum_{\mathbf{d}_{\tx}\in\mathcal{D}_{\tx}}\sum_{\mathbf{d}_{\rx}\in\mathcal{D}_{\rx}}
        (\mathbf{d}_{\tx} + \mathbf{d}_{\rx})(\mathbf{d}_{\tx} + \mathbf{d}_{\rx})^{\T} \Big) 
    \notag \\
    & \quad 
    - \big(\bm{\mu}_{\tx} + \bm{\mu}_{\rx}\big)\big(\bm{\mu}_{\tx} + \bm{\mu}_{\rx}\big)^{\T} \notag\\
    & = \frac{1}{N_{\tx}N_{\rx}}\Big(\sum_{\mathbf{d}_{\Sigma}\in\mathcal{D}_{\Sigma}} \mathbf{d}_{\Sigma}\mathbf{d}_{\Sigma}^{\T}\upsilon(\mathbf{d}_{\Sigma})\Big) - \tilde{\bm{\mu}}_\Sigma\tilde{\bm{\mu}}_\Sigma^{\T} \notag\\
    & \overset{(a)}{=} \frac{1}{N_{\tx}N_{\rx}} \sum_{\mathbf{d}_{\Sigma}\in\mathcal{D}_{\Sigma}} 
    \big(\mathbf{d}_{\Sigma} - \tilde{\bm{\mu}}_\Sigma\big)
    \big(\mathbf{d}_{\Sigma} - \tilde{\bm{\mu}}_\Sigma\big)^{\T}\upsilon(\mathbf{d}_{\Sigma}) 
    % \notag\\
    % & 
    = \tilde{\bm{\Sigma}}(\mathcal{D}_{\Sigma}), \notag
\end{align}
where in $(a)$ we have used that $\sum_{\mathbf{d}_{\Sigma}\in\mathcal{D}_\Sigma}\upsilon(\mathbf{d}_\Sigma) = N_{\tx}N_{\rx}$.
Note that $\chi(\mathcal{D}_{\tx}) + \chi(\mathcal{D}_{\rx}) = \tilde{\chi}(\mathcal{D}_{\Sigma})$ simply follows from the above equality as a special (one-dimensional) case.
%, as it corresponds to any of the two diagonal elements of the covariance.
\section{Proof of Proposition~\ref{thm:clusteredsca}}\label{app:ProofOfclusteredsca}
Under $L_{\tx} = L_{\rx}=L$, \eqref{def:sumcoarray1D} and \eqref{eq:optArrayGeometryOrthogonal} yield
    \begin{align*}
        \mathcal{D}_{\Sigma} 
        &=\mathcal{K}^L_{N_{\tx}}+\mathcal{K}^L_{N_{\rx}}\notag\\
        & = (\mathcal{U}_{N_{\tx}/2}\cup (L-\mathcal{U}_{N_{\tx}/2}) )+(\mathcal{U}_{N_{\rx}/2}\cup (L-\mathcal{U}_{N_{\rx}/2})) \notag\\
        & = (\mathcal{U}_{N_{\tx}/2}+\mathcal{U}_{N_{\rx}/2}) \cup(2L-\mathcal{U}_{N_{\tx}/2}-\mathcal{U}_{N_{\rx}/2} )\notag\\
        &\quad \cup(L+\mathcal{U}_{N_{\tx}/2}-\mathcal{U}_{N_{\rx}/2})\cup (L-\mathcal{U}_{N_{\tx}/2}+\mathcal{U}_{N_{\rx}/2} ).
       % & = \mathcal{D}_1 \textcolor{blue}{\cup} \mathcal{D}_2 \textcolor{blue}{\cup} \mathcal{D}_3 \textcolor{blue}{\cup} \mathcal{D}_4.
    \end{align*}
    Let $\mathcal{D}_1 \triangleq \mathcal{U}_{N_{\tx}/2}+\mathcal{U}_{N_{\rx}/2}=\mathcal{U}_{(N_{\tx}+N_{\rx})/2-1}$, and note that
    % \begin{align*}
    % \mathcal{D}_1 \triangleq \mathcal{U}_{N_{\tx}/2}+\mathcal{U}_{N_{\rx}/2}= \mathcal{U}_{(N_{\tx}+N_{\rx})/2}.    
    % \end{align*}
    % Note that
    \begin{align*}
        -\mathcal{U}_{N_{\rx}/2}-\mathcal{U}_{N_{\tx}/2}&=-\mathcal{D}_1=\mathcal{D}_1-(N_{\tx}+N_{\rx})/2+2,\\
        \mathcal{U}_{N_{\tx}/2}-\mathcal{U}_{N_{\rx}/2}&=\mathcal{D}_1-N_{\rx}/2+1,\\
        \mathcal{U}_{N_{\rx}/2}-\mathcal{U}_{N_{\tx}/2}&=\mathcal{D}_1-N_{\tx}/2+1.
    \end{align*}
    Furthermore, denote $\alpha\triangleq L-N_{\rx}/2+1$, $\beta\triangleq L-N_{\tx}/2+1$, and $\delta \triangleq 2L-(N_{\tx}+N_{\rx})/2+2$. Then we can write 
    \begin{align*}
    \mathcal{D}_{\Sigma}
                & =\mathcal{D}_1 \cup
                \underbrace{(\mathcal{D}_1+\delta)}_{\triangleq\mathcal{D}_2}
                \cup
                \underbrace{(\mathcal{D}_1+\alpha)}_{\triangleq\mathcal{D}_3}
                \cup \underbrace{(\mathcal{D}_1+\beta)}_{\triangleq\mathcal{D}_4}.
    \end{align*}
    By the inclusion-exclusion principle, $|\mathcal{A}\!\cup\!\mathcal{B}|\!=\!|\mathcal{A}|\!+\!|\mathcal{B}|\!-\!|\mathcal{A}\!\cap\!\mathcal{B}|$:
    \begin{align*}
        %N_\Sigma =
        |\mathcal{D}_\Sigma|\!=\!
        |\mathcal{D}_1\!\cup\!\mathcal{D}_2|\!+\!|\mathcal{D}_3\!\cup\!\mathcal{D}_4|\!-\!|(\mathcal{D}_1\cup\mathcal{D}_2)\!\cap\!(\mathcal{D}_3\cup\mathcal{D}_4)|.%\label{eq:incl_excl}
        % &=\sum_{i=1}^4|\mathcal{D}_i| -|\mathcal{D}_1\cap\mathcal{D}_2|-|\mathcal{D}_3\cap\mathcal{D}_4|\\
        % &-|(\mathcal{D}_1\cup\mathcal{D}_2)\cap(\mathcal{D}_3\cup\mathcal{D}_4)|.
    \end{align*}
    To characterize the above cardinalities, we will show that $\mathcal{D}_1$ and $\mathcal{D}_2$ are disjoint, as well as $\mathcal{D}_1\cup\mathcal{D}_2$ and $\mathcal{D}_3\cup \mathcal{D}_4$. 
    Firstly, 
     $\min\mathcal{D}_2-\max\mathcal{D}_1
     =\delta-(N_{\tx}+N_{\rx})/2+2
     %=2L-(N_{\tx}+N_{\rx})/2+2-(N_{\tx}+N_{\rx})/2+2
     =2L-N_{\tx}-N_{\rx}+4
     \geq
     \max(N_{\tx},N_{\rx})%+\textcolor{blue}{0}
     %\!>\!N_{\tx}\!+\!N_{\rx}\!+\!6
     \!>\!0$ %where the first inequality follows from 
     by assumption 
     $L\geq \max(N_{\tx},N_{\rx})+\min(N_{\rx},N_{\tx})/2-2$. 
     %$L\!>\!N_{\tx}\!+\!N_{\rx}$. 
     %$\delta > \tfrac{3}{2}(N_{\tx}+N_{\rx})+2>(N_{\tx}+N_{\rx})/2-2=\max \mathcal{D}_1$, 
     Hence, 
     \begin{align*}
         \mathcal{D}_1\cap\mathcal{D}_2=\emptyset \implies |\mathcal{D}_1\cup\mathcal{D}_2|=2|\mathcal{D}_1|=N_{\tx}+N_{\rx}-2.
     \end{align*} 
     %Moreover, $|\mathcal{D}_3\cap\mathcal{D}_4|=(N_{\tx}+N_{\rx})/2-|\alpha-\beta|=(N_{\tx}+N_{\rx})/2-|(N_{\tx}-N_{\rx})/2|=\min(N_{\tx},N_{\rx})$. 
     Moreover, since $\mathcal{D}_3\cup\mathcal{D}_4$ is contiguous,
     \begin{align*}
         |\mathcal{D}_3\cup\mathcal{D}_4|
         =(N_{\tx}+N_{\rx})/2-1+|\alpha-\beta|
         %=(N_{\tx}+N_{\rx})/2-1+\max(N_{\tx},N_{\rx})/2-\min(N_{\tx},N_{\rx})/2
         =\max(N_{\tx},N_{\rx})-1.
     \end{align*} 
     Finally,
     \begin{enumerate*}[label=(\alph*)]
     \item $\mathcal{D}_1\!\cap\!(\mathcal{D}_3\!\cup\!\mathcal{D}_4)
     %=\mathcal{D}_2\cap(\mathcal{D}_3\cup\mathcal{D}_4)
     \!=\!\emptyset$, \label{i:a}
     \item $\mathcal{D}_2\!\cap\!(\mathcal{D}_3\!\cup\!\mathcal{D}_4)\!=\!\emptyset$ \label{i:b}
     \end{enumerate*}
     imply
     \begin{align*}
         |(\mathcal{D}_1\cup\mathcal{D}_2)\cap(\mathcal{D}_3\cup\mathcal{D}_4)|=0.
     \end{align*}
     To prove \labelcref{i:a}, let $g_a\!=\! \min(\mathcal{D}_3\!\cup\!\mathcal{D}_4)\!-\!\max(\mathcal{D}_1)$ and note that
     \begin{align}
      g_a
     &=\min(\alpha,\beta)-(N_{\tx}+N_{\rx})/2+2\notag\\
     %=L+1-\max(N_{\rx},N_{\tx})/2-(N_{\tx}+N_{\rx})/2+2
     &=L-\max(N_{\tx},N_{\rx})-\min(N_{\rx},N_{\tx})/2+3.\label{eq:tmp1}
     \end{align}
     Hence, 
     $g_a
     % \min(\mathcal{D}_3\cup\mathcal{D}_4)-\max(\mathcal{D}_1)
     % =\min(\alpha,\beta)-(N_{\tx}+N_{\rx})/2+2
     % %=L+1-\max(N_{\rx},N_{\tx})/2-(N_{\tx}+N_{\rx})/2+2
     % =L-\max(N_{\tx},N_{\rx})-\min(N_{\rx},N_{\tx})/2+3
    % >\min(N_{\tx},N_{\rx})/2+3
    \!\geq\!1
     %>0
     $, by assumption 
    % $L>N_{\tx}+N_{\rx}$, 
     $L\!\geq\!\max(N_{\tx},N_{\rx})\!+\!\min(N_{\rx},N_{\tx})/2\!-\!2$,
     which implies that $\mathcal{D}_1\cap (\mathcal{D}_3\cup\mathcal{D}_4)=\emptyset$. For  
     \labelcref{i:b}, let $g_b\!=\!\min(\mathcal{D}_2)\!-\!\max(\mathcal{D}_3\!\cup\!\mathcal{D}_4)$. It can be verified 
     \begin{align}
     g_b
     %=\!\min(\mathcal{D}_2)\!-\!\max(\mathcal{D}_3\!\cup\!\mathcal{D}_4)
     =\delta-\max(\alpha,\beta)-\max \mathcal{D}_1=g_a.    \label{eq:tmp2}
     \end{align}
     Hence, 
     % \begin{align}
     %  g_b
     % &=\delta-\max(\alpha,\beta)-\max \mathcal{D}_1
     % %=2L-(N_{\tx}+N_{\rx})/2+2 -L-1+\min(N_{\rx},N_{\tx})/2-(N_{\tx}+N_{\rx})/2+2
     % &=L-\max(N_{\tx},N_{\rx})-\min(N_{\tx},N_{\rx})/2+3.\label{eq:tmp2}
     % \end{align}
     $
     g_b
     % \min(\mathcal{D}_2)-\max(\mathcal{D}_3\cup\mathcal{D}_4)
     % =\delta-\max(\alpha,\beta)-\max \mathcal{D}_1
     % %=2L-(N_{\tx}+N_{\rx})/2+2 -L-1+\min(N_{\rx},N_{\tx})/2-(N_{\tx}+N_{\rx})/2+2
     % =L-\max(N_{\tx},N_{\rx})-\min(N_{\tx},N_{\rx})/2+3
    % > \min(N_{\tx},N_{\rx})/2+3
    \geq 1
%     >0
     \implies \mathcal{D}_2\cap (\mathcal{D}_3\cup\mathcal{D}_4)=\emptyset$.
     This yields \eqref{eq:sca_size} as 
     \begin{align*}
         N_\Sigma = |\mathcal{D}_\Sigma|=% 2(N_{\tx}+N_{\rx})-0-\min(N_{\tx},N_{\rx})-
         N_{\tx}+N_{\rx}-2+\max(N_{\tx},N_{\rx})-1-0.
     \end{align*}
    Setting $L\leq \max(N_{\tx},N_{\rx})+\min(N_{\rx},N_{\tx})/2-2$ in \labelcref{eq:tmp1} yields $g_a=g_b\leq 1$ by \eqref{eq:tmp2}, which by virtue of $\mathcal{D}_1,\mathcal{D}_2$ and $\mathcal{D}_3\cup\mathcal{D}_4$ being ULAs implies that $\mathcal{D}_\Sigma= \mathcal{U}_{2L}$, i.e., $\mathcal{D}_\Sigma$ is contiguous.\hfill$\square$
    % Since $L>N_{\tx}+N_{\rx}$, we have $\mathcal{D}_1\cap\mathcal{D}_2\cap(\mathcal{D}_3\cup\mathcal{D}_4)=\emptyset$, the sum co-array has 
    % \begin{align}
    %     N_{\Sigma} = N_{\tx}+N_{\rx} + \max(N_{\tx},N_{\rx})-3,
    % \end{align} 
    % distinct elements.   
    % Moreover, the array $\mathcal{D}_1\cup \mathcal{D}_2$ would contain $2L + 1 - 2((N_{\tx}+N_{\rx})/2-1) = 2L+2-N_{\tx}-N_{\rx}$ zeros in the middle. Since $\mu(\mathcal{D}_1\cup\mathcal{D}_2) = \mu(\mathcal{D}_3\cup\mathcal{D}_4) = L+1/2$, we hence have that
    % \begin{align}
    %     \max(N_{\tx},N_{\rx})-1\geq 2L+2-N_{\tx}-N_{\rx} \notag\\
    %     \iff \notag \\
    %     N_{\tx}+N_{\rx}/2  \geq L+1 \quad \wedge \quad N_{\tx}/2+N_{\rx}  \geq L+1, \notag
    % \end{align}
    % ensures a contiguous sum co-array.
% \input{Appendices/AppendixPlanarSumBeamOptimality}
\section{Proof of Theorem~\ref{thm:sumBeamOptimality}}\label{app:sumBeamOptimality}
For notational convenience, let $\bm{\Sigma}_{(.)}\triangleq \bm{\Sigma}(\mathcal{D}_{(.)})$ and define $\alpha \triangleq 2|\gamma|^2N_{\tx}N_{\rx}/\sigma^2$, which is positive by definition.
The optimization problem in \eqref{eq:objectiveFunctionCRBM} can then be relaxed to
\begin{align}
    \min_{\lambda_{1},\;\bm{\Lambda}_{2}} 
    & \quad f\!\left(
        \big(
            \alpha
            \lambda_{1}\bm{\Sigma}_{\rx}
            +
            \alpha
            \bm{\Sigma}_{\tx}^{\frac{1}{2}}
            \bm{\Lambda}_{2}
            \bm{\Sigma}_{\tx}^{\frac{1}{2}}%^\top
        \big)^{-1}
    \right) \label{eq:opt_obj}\\
    \text{s.t.} & \quad \lambda_{1} + 
    \operatorname{tr}(\bm{\Lambda}_{2})
    = 1,\;\; \lambda_{1}\geq 0,\;\; \bm{\Lambda}_{2}\succeq 0.
    \label{eq:opt_cons}
\end{align}
Here the only relaxation is that the second constraint in \eqref{eq:opt_cons} is replaced by a (non-strict) inequality. We will show that the stationary point of this relaxed problem also satisfies the original (strict) constraint and is therefore optimal for \eqref{eq:objectiveFunctionCRBM}.
The objective function \eqref{eq:opt_obj} is convex in $\lambda_1$ and $\bm{\Lambda}_2$ (following the assumption made in the theorem \textcolor{\edited}{that $f(\mathbf{X}^{-1})$ is convex} and \textcolor{\edited}{given} that $\mathbf{X}\!=\!\alpha\lambda_{1}\bm{\Sigma}_{\rx}\!+\alpha\!\bm{\Sigma}_{\tx}^{\frac{1}{2}}\bm{\Lambda}_{2}\bm{\Sigma}_{\tx}^{\frac{1}{2}}$ \textcolor{\edited}{is linear in $\lambda_1$ and $\bm{\Lambda}_2$}), the feasible set is convex and strictly feasible, and therefore strong duality holds. The \ac{kkt} conditions are necessary and sufficient for optimality.
Introducing dual variables $\mu$, $\nu$, $\mathbf{Z}$ for the constraints in \eqref{eq:opt_cons}, the Lagrangian is
\begin{align}
    \mathcal{L}(\lambda_1,\bm{\Lambda}_2,\mu,\nu,\mathbf{Z}) 
    % \notag\\
    % & 
    &= f(\mathbf{X}^{-1}) + \mu (\lambda_1 + \textrm{trace}(\bm{\Lambda}_2) - 1) \notag\\
    & \qquad - \nu \lambda_1 - \textrm{trace}(\mathbf{Z}\bm{\Lambda}_2).
\end{align}
Let $\mathbf{G} = - \nabla_{\mathbf{X}} f(\mathbf{X}^{-1})  = \mathbf{X}^{-1}\left(\nabla_{\mathbf{X}^{-1}}f(\mathbf{X}^{-1})\right)\mathbf{X}^{-1}$, 
and 
denote the differential operator %with 
by 
$d$. First note that $d\mathbf{X}^{-1} = -\mathbf{X}^{-1}(d \mathbf{X})\mathbf{X}^{-1}$,
% \begin{align}
%     d\mathbf{X}^{-1} &= -\mathbf{X}^{-1}(d \mathbf{X})\mathbf{X}^{-1},
% \end{align}
and that $d \mathbf{X} = 
    \alpha
    \bm{\Sigma}_{\rx} 
    \left( d \lambda_1\right) 
    + 
    \alpha
    \bm{\Sigma}_{\tx}^{\frac{1}{2}}
    \left(d\bm{\Lambda}_{2}\right)
    \bm{\Sigma}_{\tx}^{\frac{1}{2}}.$
% \begin{align}
%     d \mathbf{X} = 
%     \bm{\Sigma}_{\rx} 
%     \left( d \lambda_1\right) 
%     + 
%     \bm{\Sigma}_{\tx}^{\frac{1}{2}}
%     \left(d\bm{\Lambda}_{2}\right)
%     \bm{\Sigma}_{\tx}^{\frac{1}{2}}.
% \end{align}
Let $\mathbf{Y}= \mathbf{X}^{-1}$, such that $f(\mathbf{Y})$. Then applying the chain rule, and substituting the expressions above, we can write
\begin{align}
    & d f 
    % & 
    = \textrm{trace}\left(
    \left(
    % \partial f(\mathbf{Y}) / \partial \mathbf{Y} 
    % \frac{\partial f(\mathbf{Y})}{\partial \mathbf{Y}}
    \nabla_{\mathbf{Y}} f(\mathbf{Y})
    \right)^{\T}
    (d\mathbf{Y})
    \right) 
    \notag\\
    &
    = 
    \textrm{trace}\left(
    -\mathbf{X}^{-1}
    \nabla_{\mathbf{X}^{-1}}f(\mathbf{X}^{-1}) \mathbf{X}^{-1}
    (d \mathbf{X})
    \right) 
    % \notag\\
    % & 
    = 
    \textrm{trace}\left(
    -\mathbf{G}^{\T}
    (d \mathbf{X})
    \right) \notag\\
    & =  
    -
    \alpha
    \textrm{trace}\left(
    \mathbf{G}^{\T}
    \bm{\Sigma}_{\rx}
    \right)  
    \left( d \lambda_1\right)
    +
    \textrm{trace}\left(
    -
    \alpha
    \bm{\Sigma}_{\tx}^{\frac{1}{2}}
    \mathbf{G}^{\T}
    \bm{\Sigma}_{\tx}^{\frac{1}{2}}
    \left(d\bm{\Lambda}_{2}\right)
    \right). \notag
\end{align}
Hence the partial derivatives of $f(\mathbf{X}^{-1})$ %gradients 
are given by 
\begin{align}
    \frac{\partial f(\mathbf{X}^{-1})}{\partial \lambda_1}
    = -\alpha\textrm{trace}(\mathbf{G}\bm{\Sigma}_{\rx}), 
    \
    \frac{\partial f(\mathbf{X}^{-1})}{\partial \bm{\Lambda}_2} 
    = -\alpha\bm{\Sigma}_{\tx}^{\frac{1}{2}}%^{\T}
    \mathbf{G}\bm{\Sigma}_{\tx}^{\frac{1}{2}}. \notag
\end{align}

% \textcolor{red}{[Robin: Does this equality hold for any $f$? Needs reference or proof.]}

% %\textcolor{red}{[Robin: Again, above is nonobvious. Reference or elaboration needed.]}

% \textcolor{red}{[Ids: See in blue above. This derivation holds for any continuously differentiable $f$.]}
%
The KKT conditions are
\begin{align}
    % primal feasibility:
    & \lambda\textcolor{\edited}{^\star}_{1} + \textrm{trace}(\bm{\Lambda}\textcolor{\edited}{^\star}_{2})= 1,\quad \lambda_{1}\textcolor{\edited}{^\star}\geq 0,\quad \bm{\Lambda}\textcolor{\edited}{^\star}_{2}\succeq 0, \\
    % dual feasibility
    & \textcolor{\edited}{\nu^\star} \geq 0, \quad \mathbf{Z}\textcolor{\edited}{^\star}\succeq \mathbf 0, \label{eq:dualFeasibility} \\
    % complementary slackness:
    & \nu\textcolor{\edited}{^\star}\lambda\textcolor{\edited}{^\star}_1 = 0, \quad \textrm{trace}(\mathbf{Z}\textcolor{\edited}{^\star}\bm{\Lambda}\textcolor{\edited}{^\star}_2) = 0, \label{eq:complementarySlackness}\\
    % stationarity/vanishing gradients
    &-\alpha\textrm{trace}(\mathbf{G}\bm{\Sigma}_{\rx}) +\mu\textcolor{\edited}{^\star} - \nu\textcolor{\edited}{^\star} = 0, 
    \label{eq:stationarity1}\\
    & 
    -\alpha\bm{\Sigma}_{\tx}^{\frac{1}{2}} \mathbf{G}\bm{\Sigma}_{\tx}^{\frac{1}{2}} + \mu\textcolor{\edited}{^\star}\mathbf{I} - {\mathbf{Z}\textcolor{\edited}{^\star}}^{\T} = \mathbf{0}.\label{eq:stationarity2}
\end{align}
Finally, \textcolor{\edited}{we can show that \eqref{eq:condition} is a necessary and sufficient condition for $(\lambda_1^\star,\bm{\Lambda}_2)= (1,\mathbf{0})$ to be optimal.}

\textcolor{\edited}{\textit{(Sufficiency)} Observe that if \eqref{eq:condition} holds, then the tuple 
\begin{align}
    & \left(\lambda\textcolor{\edited}{^\star}_1,\bm{\Lambda}\textcolor{\edited}{^\star}_2,\mu\textcolor{\edited}{^\star},\nu\textcolor{\edited}{^\star},\mathbf{Z}\textcolor{\edited}{^\star}\right) = \notag \\
    & \left(
        1,
        \mathbf{0},
        \alpha \textrm{trace}\left(\mathbf{G}\bm{\Sigma}_{\rx}\right), 
        0,  
        \alpha\textrm{trace}\left(\mathbf{G}\bm{\Sigma}_{\rx}\right)\mathbf{I} - \alpha\bm{\Sigma}_{\tx}^{\frac{1}{2}}%^{\T}
        \mathbf{G}\bm{\Sigma}_{\tx}^{\frac{1}{2}}
    \right)
\end{align} 
satisfies all \ac{kkt} conditions. Satisfying the \ac{kkt} conditions is a sufficient condition for optimality of $(\lambda^\star_1,\bm{\Lambda}^\star_2)= (1,\mathbf{0})$.}

\textcolor{\edited}{\textit{(Necessity)} Now suppose that $(\lambda^\star_1,\bm{\Lambda}^\star_2)= (1,\mathbf{0})$ are optimal. Since the problem in \eqref{eq:opt_obj} and \eqref{eq:opt_cons} is strictly feasible (Slater's condition holds), the \ac{kkt} conditions are a necessary condition for optimality of $(\lambda^\star_1,\bm{\Lambda}^\star_2)$. Then, since the \ac{kkt} conditions hold, $\mathbf{Z}\succeq \mathbf{0}$, and hence \eqref{eq:condition} must hold.}
We conclude that, \eqref{eq:condition} is necessary and sufficient for $(\lambda_1,\bm{\Lambda}_2) = (1,\mathbf{0})$ to be a solution to \eqref{eq:opt_obj}-\eqref{eq:opt_cons} and hence also to \eqref{eq:objectiveFunctionCRBM}. Substituting $(\lambda_1,\bm{\Lambda}_2) = (1,\mathbf{0})$ into \eqref{eq:optWavPlanar} then yields \eqref{eq:optTxWaveform_sumPlanar}. \hfill$\square$
%\end{proof}
% \input{Appendices/AppendixG}
% \input{Appendices/AppendixF}
% \input{Appendices/AppendixSensorDistrCoprime}

\bibliographystyle{IEEEtran}
% \bibliography{references}
\bibliography{referencesAbreviated}

 

% \end{IEEEbiography}

%\newpage
%\appendices
% IDS VAN DER WERF
\begin{IEEEbiography}[{\includegraphics[width=1in,height=1.25in,clip,keepaspectratio]{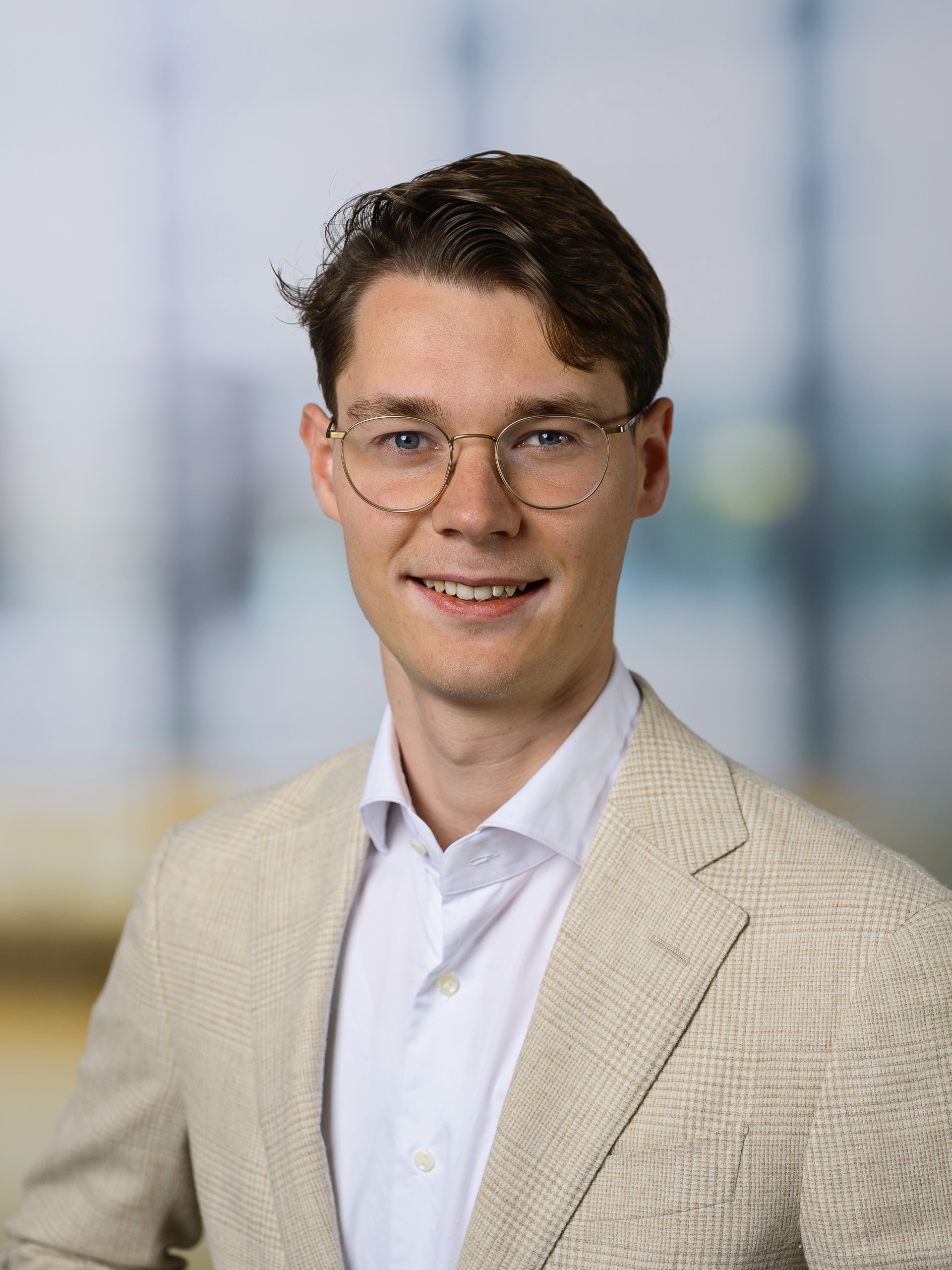}}]{Ids van der Werf} (Graduate Student Member, IEEE) received the B.Sc. and MSc degrees in electrical engineering from the Delft University of Technology, The Netherlands, in 2019 and 2022, respectively. He is currently pursuing the Ph.D. degree in the Signal Processing Systems group, Delft University of Technology, The Netherlands.
\end{IEEEbiography}

\begin{IEEEbiography}
[{\includegraphics[width=1in,height=1.25in,clip,keepaspectratio]{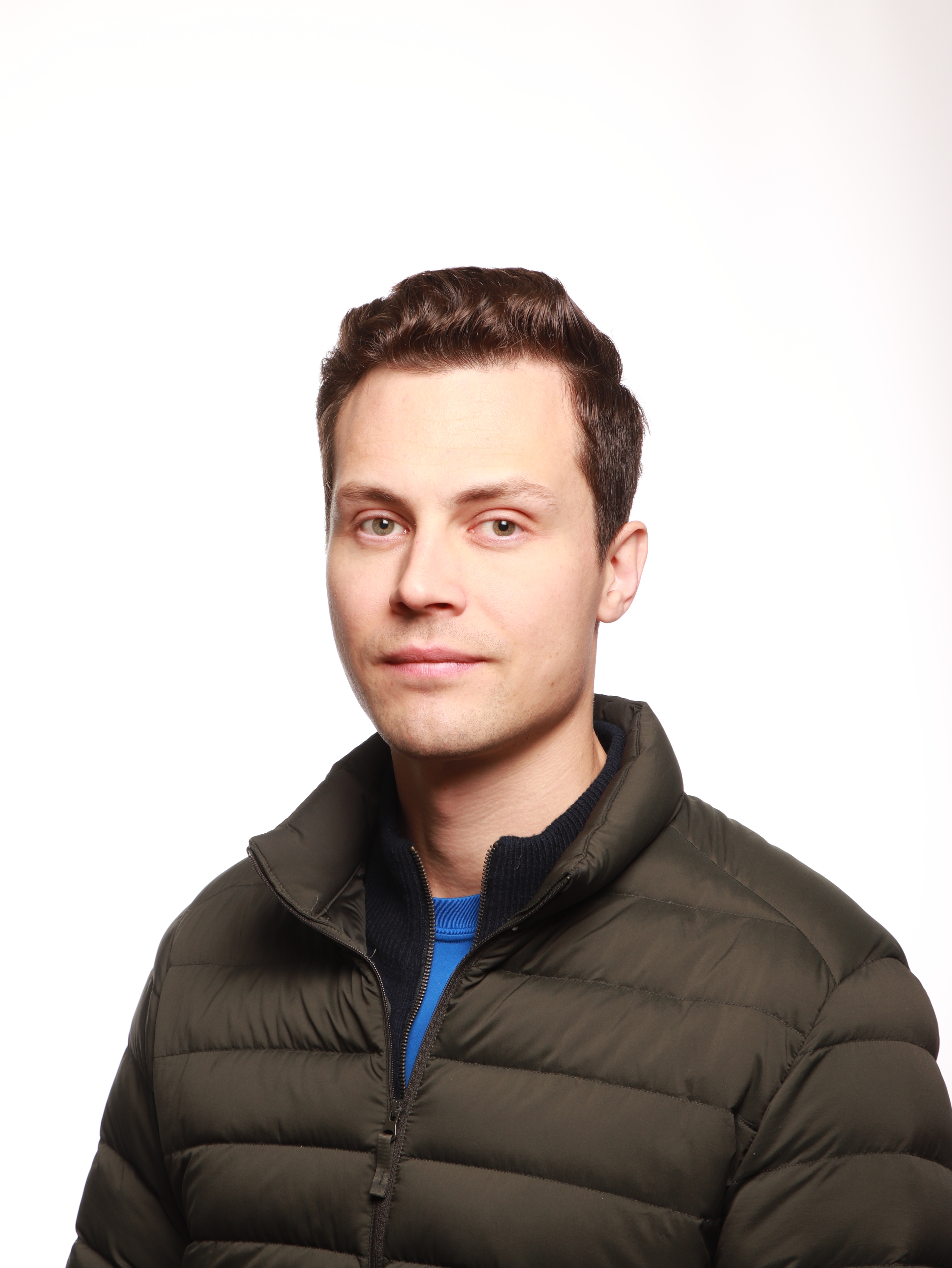}}]{Robin Rajamäki}
(Member, IEEE) received his D.Sc. degree in electrical engineering in 2021 from Aalto University, Finland. He was a postdoctoral scholar at the University of California San Diego in 2022-2024, and Aalto University in 2021-2022 / 2024-2025. He is currently an assistant professor at Tampere University, Finland. His research interests lie in the intersection of theory and applications of statistical signal processing, with a focus on multisensor systems in sensing and communications.
\end{IEEEbiography}

% GEERT LEUS
\begin{IEEEbiography}[{\includegraphics[width=1in,height=1.25in,clip,keepaspectratio]{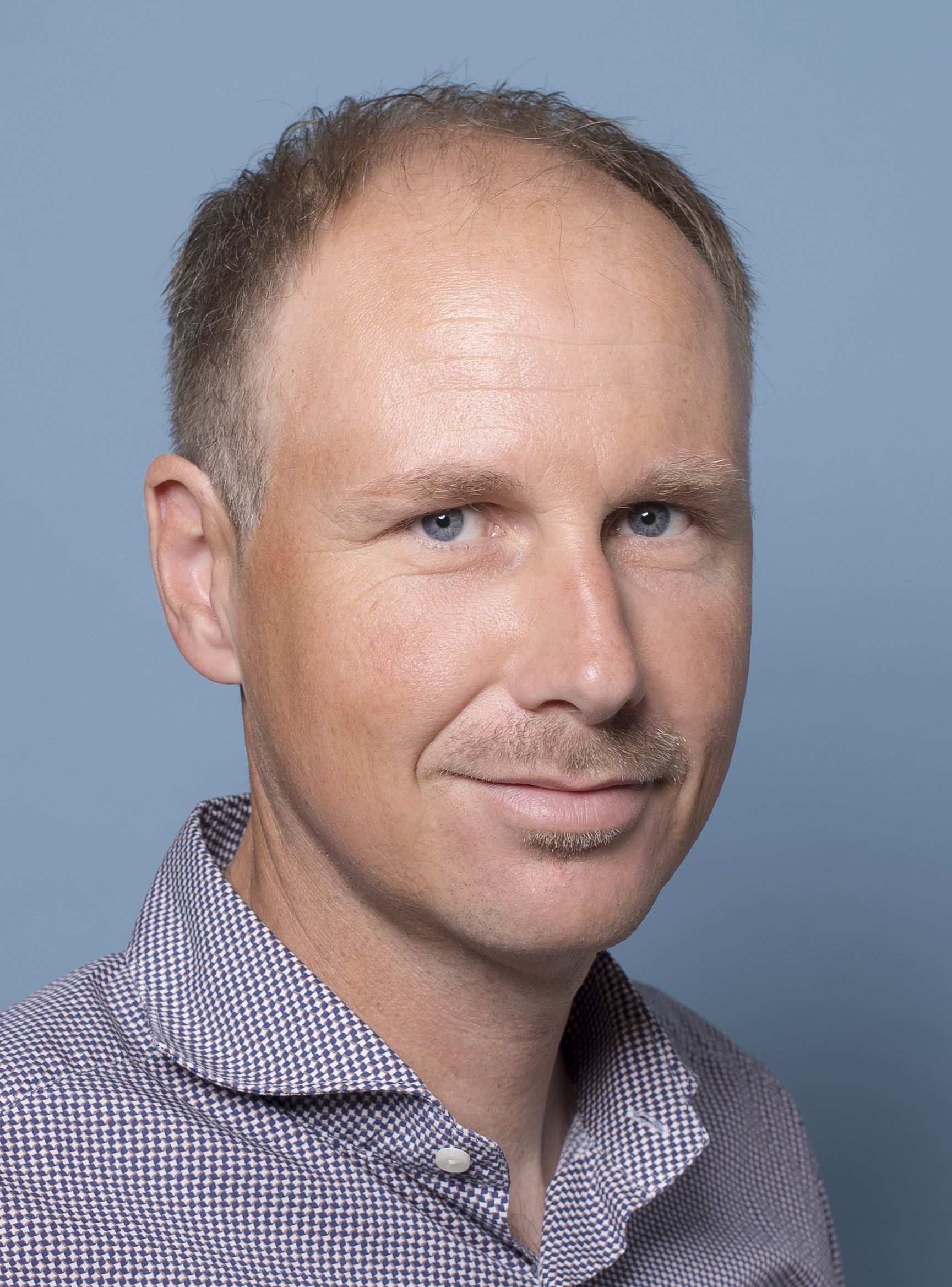}}]{Geert Leus} (Fellow, IEEE) is a Full Professor at the Faculty of EEMCS at the Delft University of Technology, The Netherlands. Geert Leus received the 2021 EURASIP Individual Technical Achievement Award, a 2024 and 2005 IEEE SPS Best Paper Award, and a 2002 IEEE SPS Young Author Best Paper Award. He is a Fellow of the IEEE and EURASIP. Geert Leus was a Member-at-Large of the BoG of the IEEE SPS, the Chair of the IEEE SPCOM TC, the Chair of the EURASIP SPMuS TAC, the EiC of EURASIP JASP and the EiC of EURASIP Signal Processing.
\end{IEEEbiography}


\clearpage

\end{document}